\documentclass[prd,aps,amsfonts,superscriptaddress,nofootinbib,longbibliography,notitlepage]{revtex4-1}

\usepackage{graphicx}
\usepackage{xcolor}
\usepackage{subcaption}
\usepackage{rotating}
\usepackage{amsmath,amssymb,graphics,amsthm,isomath}
\usepackage{physics }
\usepackage{verbatim}
\usepackage{tikz}
\usepackage{blkarray}
\usepackage{mathrsfs}
\usepackage{algorithm}
\usepackage{algpseudocode}

\usepackage{xcolor}
\definecolor{mycitecolor}{rgb}{0.0, 0.45, 0.85}   

\usepackage[
  colorlinks=true,
  linkcolor=mycitecolor,
  citecolor=mycitecolor,
  urlcolor=mycitecolor,
  hyperindex=true,
  linktocpage=true
]{hyperref}

\usepackage[capitalise,compress]{cleveref}
\usepackage{tcolorbox}

\newtcolorbox{shadedtheorem}{
  colback=gray!15,    
  colframe=white,     
  boxrule=0pt,        
  arc=0pt,            
  left=5pt,           
  right=5pt,
  top=5pt,
  bottom=5pt
}

\newcommand\LRvel{v_{\rm LR}}
\newcommand{\e}{\mathrm{e}}
\newcommand{\ii}{\mathrm{i}}

\newcommand\diam{\operatorname{diam}}

\newcommand\loc{\operatorname{loc}}
\newcommand\rng{\operatorname{rng}}
\newcommand\CLR{c_{\mathrm {LR}}}
\renewcommand\equiv{:=}
\renewcommand\epsilon{\varepsilon}
\newcommand\SW{\text{SW}}
\newcommand\tildemathcal[1]{\widetilde{\mathcal #1}}
\newcommand\Sint[1]{\mathrm S_{\mathrm{int}}^{(#1)}}
\newcommand\Sext[1]{\mathrm S_{\mathrm{ext}}^{(#1)}}
\newcommand\overlinebold[1]{\overline{\mathbb {#1}}}
\newcommand\resext{|_{\Sext{i-1}}}

\newcommand\isbig{\operatorname{Big}}
\newcommand\polylog{\operatorname{polylog}}

\newtheorem{thm}{Theorem}
\numberwithin{thm}{section}
\newtheorem{cor}[thm]{Corollary}
\newtheorem{lem}[thm]{Lemma}

\newtheorem{prop}[thm]{Proposition}

\newtheorem{defn}[thm]{Definition}
\newtheorem{rmk}[thm]{Remark}

\renewcommand{\thesection}{\arabic{section}}
\renewcommand{\thesubsection}{\thesection.\arabic{subsection}}
\renewcommand{\thesubsubsection}{\thesubsection.\arabic{subsubsection}}

\makeatletter
\renewcommand{\p@subsection}{}
\renewcommand{\p@subsubsection}{}
\makeatother

\usepackage{mathtools}
\tikzstyle{densely dashed}= [dash pattern=on 4pt off 3pt]

\newcommand{\ad}{\operatorname{ad}}

\newcommand{\supp}{\operatorname{supp}}

\usepackage{dsfont}

\begin{document}

\title{Non-perturbatively slow spread of quantum correlations in non-resonant systems}

\author{Ben T. McDonough}
\email{ben.mcdonough@colorado.edu}
\affiliation{Department of Physics and Center for Theory of Quantum Matter, University of Colorado, Boulder CO 80309, USA}

\author{Marius Lemm}
\email{marius.lemm@uni-tuebingen.de}
\affiliation{Department of Mathematics, University of T\"ubingen, 72076 T\"ubingen, Germany}

\author{Andrew Lucas}
\email{andrew.j.lucas@colorado.edu}
\affiliation{Department of Physics and Center for Theory of Quantum Matter, University of Colorado, Boulder CO 80309, USA}

\begin{abstract}
Strong disorder often has drastic consequences for quantum dynamics.  This is best illustrated by the phenomenon of Anderson localization in non-interacting systems, where destructive quantum wave interference leads to the complete absence of particle and information transport over macroscopic distances.  In this work, we investigate the extent to which strong disorder leads to provably slow dynamics in many-body quantum lattice models.  We show that in any spatial dimension, strong disorder leads to a non-perturbatively small velocity for ballistic information transport under unitary quantum dynamics, almost surely in the thermodynamic limit, in every many-body state.  In these models, we also prove the existence of a ``prethermal many-body localized regime", where entanglement spreads logarithmically slowly, up to non-perturbatively long time scales.  More generally, these conclusions hold for all models corresponding to quantum perturbations to a classical Hamiltonian obeying a simple non-resonant condition. Deterministic non-resonant models are found, including spin systems in strong incommensurate lattice potentials. Consequently, quantum dynamics in non-resonant potentials is  asymptotically easier to simulate on both classical or quantum computers, compared to a generic many-body system.  
\end{abstract}

\maketitle

\tableofcontents

\section{Introduction}
Disorder profoundly affects quantum mechanics. The
most well-known example is the Anderson localization \cite{anderson,frohlich1983absence,aizenman1993localization,aizenman2015random} of single-particle wave functions in a disordered potential: almost surely, eigenstates of a single-particle Hamiltonian in a sufficiently disordered background are spatially localized.  Anderson localization represents a complete breakdown of ergodicity: the particle does not explore all of ``phase space" consistent with energy conservation.  

In recent years, cold atoms, trapped ions, and superconducting circuits are all able to accurately probe the quantum dynamics of states at finite energy density in isolation from a thermal bath.  Can such experiments observe a breakdown of ergodicity and statistical mechanics in \emph{many-body} systems?  One possible mechanism is \textit{many-body localization} (MBL), the analogue of Anderson localization for many interacting particles. The suggestion that localization persists in the presence of interactions was suggested in Anderson's original work \cite{anderson} and refined 50 years later with arguments in favor of Anderson localization in fermionic systems with short-range interactions \cite{Basko_MBL_foundations, Gornyi_MBL_foundations}. It has recently been established rigorously that in  many-body quantum models with few-body interactions and extensive energy barriers between low-energy states, eigenstate localization 
is possible and is robust to almost every perturbation \cite{Yin:2024jad}.  However, such models cannot be realized in any finite-dimensional lattice. Hence, the much more intriguing problem, which has not been rigorously resolved, is whether simple many-body quantum systems of locally-interacting spins or qudits on a $d$-dimensional lattice, which never have extensive energy barriers,  can nevertheless exhibit MBL without extreme fine-tuning of the Hamiltonian \cite{MBLColloquium}.  It is now expected \cite{deroeckhuveneers,thiery2018many,morningstar2020many} that if this is possible, a natural setting is in strongly disordered one-dimensional spin chains \cite{pal2010many}:
\begin{equation}\label{eq:Hintro}
    H = \sum_{i=1}^L h_i Z_i + \epsilon V
\end{equation}
where the $h_i$'s are disordered and $\epsilon \ll 1$, with \cite{alukin_entanglement, Choi_2016,Li:2025kje, hur2025stabilitymanybodylocalizationdimensions, Schreiber_2015, Smith2016, localization_2d_quasiperiodic,MBL_powerlaw, MBL_long_range} discussing additional possibilities.  Proving MBL  (naively) requires having exquisite control over the many-body spectrum, which is extremely sensitive to perturbations: see \cite{imbrie2016many,absenceofconduction} for recent efforts.    For certain 1D Hamiltonians, subdiffusive transport and/or slow dynamics has been proved by leveraging additional assumptions, e.g., sparsity of interactions and heavy-tailed disorder \cite{mastropietro2015localization,baldwingorshkov,gebert2022lieb,absenceofconduction,toniolo2025logarithmic,baldwin2025subballistic}.

In spatial dimension $d>1$, there exist heuristic arguments which strongly suggest that MBL is dynamically unstable \cite{deroeckhuveneers,thiery2018many,morningstar2020many}.  Even so, it is physically reasonable to  expect that the disorder still has a marked and observable slowing-down effect on the quantum dynamics. Indeed, an extensive literature has  explored the potential implications of MBL on observable physics, with a particularly striking and testable prediction being the slow generation of entanglement between distant qubits under unitary dynamics generated by $H$: for two qubits separated by distance $r$, it is expected that the time needed to entangle them in a fully MBL system is \cite{Serbyn_2013,huse2014phenomenology} \begin{equation}
    t \sim \mathrm{e}^{a\cdot r \log (1/\epsilon)}.  \label{eq:texpr}
\end{equation}
for an O(1) constant $a$.    This physics will persist to long time scales even if eigenstates do not actually localize \cite{Suntajs:2019lmb,Morningstar:2021pcy}, leading to a ``prethermal MBL" regime characterized by slow dynamics.

In this paper, we ask whether---independently of the fate of MBL---strong disorder will generically lead to extremely slow many-body quantum dynamics.   Under what conditions do many-body quantum lattice models exhibit \eqref{eq:texpr}?  Is dynamics slow at all times, or only on ``prethermal" time scales?  These questions are particularly relevant for experiments using quantum simulators  \cite{Choi_2016,Li:2025kje,hur2025stabilitymanybodylocalizationdimensions}, which cannot probe eigenstate localization anyway, but must instead look for signatures of ``MBL" on shorter time scales.

Our result is about many-body quantum models with Hamiltonian  $H=H_0+\epsilon V$, 
where $V$ again is spatially local in $d$ dimensions and contains few-body interactions, while $H_0$ is a sum of Pauli-$Z$ strings (thus commuting and solvable). We assume that $H_0$ obeys a non-resonant condition: bit flips on $r$ adjacent lattice sites change the energy of $H_0$ by \begin{equation}
    \Delta H_0(r) \gtrsim \e^{-b\cdot r^\xi } \label{eq:intrononresonant}
\end{equation}
for some $O(1)$ constants $b$ and $\xi \ge d$. Crucially, this non-resonance assumption only involves the exactly solvable part $H_0$ and can thus be analytically verified in examples, as we explain further below.
We prove that the full Hamiltonian $H$ then exhibits non-perturbatively slow dynamics.
First, we use \eqref{eq:intrononresonant} to prove that \eqref{eq:texpr} holds for distances \begin{equation}
    r < r_* \sim  \log^{\frac{1}{\xi}}(\epsilon^{-1}). \label{eq:main_rstar}
\end{equation}
For distances $r>r_*$, we find that any ballistic information spreading is limited to occur with maximal velocity \begin{equation}
    v \sim \epsilon^{r_*}\label{eq:vscalingintro}
\end{equation}
which is non-perturbatively small in the interaction strength, i.e. vanishes faster than any power law in $\epsilon$.  
These slow dynamics hold for \emph{arbitrary initial states}---even highly entangled initial states cannot change their pattern of entanglement over short time scales! We emphasize that our results hold for any spatial dimension $d$, in contrast to the expected sensitivity of MBL physics on spatial dimension \cite{deroeckhuveneers,thiery2018many,morningstar2020many}; hence, a proposed signature of MBL that could be observed by experiment is \emph{universal} for any non-resonant system, and is not sensitive to whether eigenstates are thermal or not.

To prove these results, we leverage a number of well-established methods in mathematical physics.  Lieb-Robinson bounds (LRBs) \cite{Lieb1972,AnthonyChen:2023bbe} show that local quantum lattice models can generate quantum correlations only at a finite velocity, starting from arbitrary initial states.  
The recently-developed irreducible path construction \cite{chen2021operator} provides a drastic improvement over standard Lieb-Robinson bounds \cite{Lieb1972} that is essential to our proof. 
However, existing LRBs on Hamiltonians such as \eqref{eq:Hintro} in the literature \cite{premont2010lieb,haah2021quantum,wanghazzard,baldwingorshkov,lemm2025enhanced}   only guarantee that $v\sim \epsilon$, far weaker than \eqref{eq:vscalingintro}, and only exploit commutativity, not disorder, of $H_0$. 

To prove stronger LRBs for a non-resonant model with bounded disorder in any dimension, we employ a local Schrieffer-Wolff transformation \cite{bravyi2011schrieffer} (operator perturbation theory) to account for the destructive interference that arises in many-body quantum dynamics in non-resonant models.  Similar techniques have previously established slow \emph{prethermal} dynamics in non-resonant models \cite{deroeck2023} as well as the stability of quantum topological order \cite{bravyi2010} and the slow decay of false vacua \cite{yin2023prethermalization,Yin:2024hjm}.  In the present paper, we marry these two techniques (Lieb-Robinson bounds and Schrieffer-Wolff transformations) and obtain non-perturbative bounds on quantum dynamics in non-resonant models.

Our calculations reveal a key difference between a typical ``prethermal" phenomenon, such as the metastability of a false vacuum, and non-resonant dynamics.  Intuitively in the former, rare quantum fluctuations nucleate a ``critical bubble" of true vacuum, which is energetically resonant with the original false vacuum state; on the prethermal time scale, these bubbles have engulfed the whole system.  The nucleation rate of critical bubbles is associated with a prethermal timescale $t_\ast$, but for times $t\gg t_*$ there is no remaining notion of metastability or slow dynamics, because signals can be sent quickly in the thermal phase (``true vacuum").   In contrast, \eqref{eq:vscalingintro} implies that non-resonant systems have slow dynamics on \emph{all time scales}---even after a ``resonance" has allowed for information to propagate a distance $r_*$, we must wait \emph{another} time $t_*$ for information to propagate another $r_*$ sites away.   Hence, prethermal MBL is a qualitatively distinct kind of ``prethermal" phenomenon than, e.g., the decay of a metastable state.\footnote{For this reason, previous authors \cite{deroeck2023} referred to such phenomena by other names such as ``quasilocalization".  But it appears that prethermal MBL is the commonly used name in the recent literature, so we will use it as well.}

In order for our theorems to be useful, it must be the case that \eqref{eq:intrononresonant} is realizable in reasonable models of $H_0$.  This is indeed the case. Notice that since $H_0$ is commuting, \eqref{eq:intrononresonant} reduces to a statement about lack of resonances among the fields $h_i$, and so it is simpler to verify than a many-body off-resonance condition as appeared in \cite{imbrie2016many}.  We prove that this criterion holds in the vast majority of spatial regions for disordered models, establishing a key signature \eqref{eq:texpr} of MBL physics over non-perturbatively long time scales, in every spatial dimension $d$.  Moreover, we also present two \emph{non-random} $H_0$ which obey \eqref{eq:intrononresonant} in sufficiently many regions for our main results to apply.   One of them simply corresponds to placing the qubits in an incommensurate lattice potential \cite{aubry1980analyticity,sarang_quasiperiodic, huse_quasiperiodic}, which may be more accessible in experiments using ultracold atoms trapped in optical lattices, where an incommensurate potential can be easier to introduce than a fully random one \cite{localization_2d_quasiperiodic, hur2025stabilitymanybodylocalizationdimensions, Schreiber_2015}.

\section{Motivation}
To understand heuristically why we should expect \eqref{eq:texpr} and \eqref{eq:vscalingintro} to hold, it is instructive to first consider the quantum mechanics of a single particle hopping in a non-resonant landscape.   For simplicity, let us consider the following one-dimensional model: \begin{equation}
    H = \sum_{n\in\mathbb{Z}} \left[ h_n |n\rangle\langle n| + \epsilon |n\rangle\langle n+1| +\epsilon |n+1\rangle\langle n|\right] .
\end{equation}
Suppose that we now take $h_n$ to be a generic periodic sequence with period $r_*$, i.e. $h_n=h_{n+r_*}$. Then the dynamics can then be understood by Bloch theory.  Arguing formally for simplicity, the   (non-normalizable) formal eigenstates of $H$ are  classified by a wave number $k\in [-\frac{\pi}{r_*},\frac{\pi}{r_*})$ and a set of $r_*$ numbers $c_0,\ldots, c_{r_*-1}$: \begin{equation}
    |E\rangle = \sum_{n\in\mathbb{Z}}c_{n\text{ (mod $r_*$)}}\mathrm{e}^{\mathrm{i}kn}|n\rangle.
\end{equation}  
If $\epsilon=0$, then exactly one $c_n \ne 0$ for any given eigenstate; if $c_m\ne 0$ then $E=h_m$.  There is a trivial dispersion relation \begin{equation}
    E_m(k)=h_m. \label{eq:trivialdispersion}
\end{equation} 

Now suppose that \begin{equation}
    \epsilon \ll |h_j-h_{k}| \label{eq:nonresonantsingleparticle}
\end{equation}for any $j\neq k$.   The Schr\"odinger equation $H|E\rangle = E|E\rangle$ reads \begin{equation}
   \left(\begin{array}{ccccc} h_0 &\ \epsilon \mathrm{e}^{\mathrm{i}k} &\ 0 &\ \cdots &\ \epsilon \mathrm{e}^{-\mathrm{i}k} \\ \epsilon \mathrm{e}^{-\mathrm{i}k} &\ h_1 &\ \epsilon \mathrm{e}^{\mathrm{i}k} &\cdots &\ 0  \\ 0 &\ \epsilon \mathrm{e}^{-\mathrm{i}k} &\ h_2 &\cdots &\ 0 \\ \vdots &\ \vdots &\ \vdots &\ &\ \vdots \\ \epsilon\mathrm{e}^{\mathrm{i}k} &\ 0 &\ 0 &\ \cdots &\  h_{r_*-1} \end{array}\right)\left(\begin{array}{c} c_0 \\ c_1 \\ c_2 \\ \vdots \\ c_{r_*-1} \end{array}\right)  =E \left(\begin{array}{c} c_0 \\ c_1 \\ c_2 \\ \vdots \\ c_{r_*-1} \end{array}\right)
\end{equation}
The dispersion relation $E_m(k)$ is found by solving for the roots of the $r_*$-order characteristic polynomial of the matrix above.  A short calculation reveals that 
\begin{equation}
    \det(E \cdot 1 - H) = f(E,h_0,\ldots,h_{r_*-1},\epsilon)  - (-1)^{r_*} 2 \cos (kr_*)\epsilon^{r_*},  \label{eq:detHsingleparticle}
\end{equation}
where the polynomial $f(E)$ has $r_*$ distinct roots for $\epsilon$ sufficiently small. Critically all $k$-dependence is in the last term alone.  The velocity at which the quantum wave packets propagate is the group velocity \begin{equation}
    v_{\mathrm{g},m}(k) := \frac{\mathrm{d}E_m(k)}{\mathrm{d}k}.
\end{equation}
For this derivative to be non-trivial, the solutions $E_m(k)$ to \eqref{eq:detHsingleparticle} must be corrected by the $\mathrm{O}(\epsilon^{r_*})$ term, which is the only $k$-dependent term.  We immediately deduce that for \emph{every} band in the system, $v_{\mathrm{g},m}\sim \epsilon^{r_*}$. The argument can be made rigorous and it extends to higher dimensions \cite{abdul2025sharp}.\footnote{If some $h_j=h_k$, then the above conclusion no longer holds in general.  The reason is that the polynomial $f$ has higher-order zeros at $\epsilon=0$, so the perturbative correction to the location of these zeros does not need to occur at $\mathrm{O}(\epsilon^{r_*})$.}  The intuition is that a higher period leads to flattening of the bands and thus to smaller curvature (group velocity), and a numerical example of this effect is shown in Fig.~\ref{fig:bandstructure}.

\begin{figure}
    \centering
    \def\svgwidth{0.45\linewidth} 
    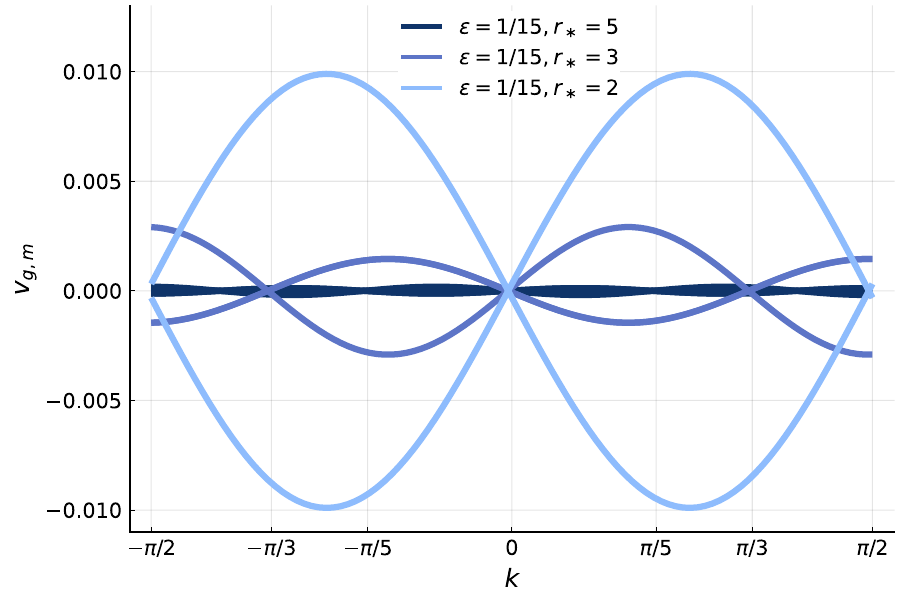
    \def\svgwidth{0.45\linewidth} 
    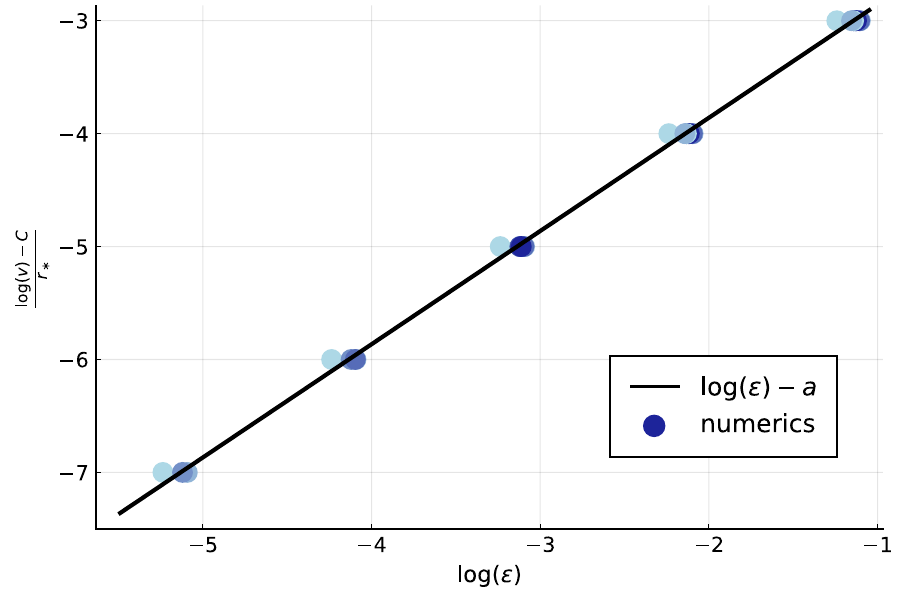
    \caption{\textbf{Left:} Group velocity of each band for $h_n = n/r_\ast$ where $\epsilon = \frac{1}{15}$ and $r_\ast = 2,3,5$. The bands become asymptotically flat as $r_\ast$ is increased. \textbf{Right:} Scaling of maximum group velocity with $\epsilon$ for different system sizes $r_\ast$, which we find scales with $v_g(r_\ast, \epsilon) \approx C(a\epsilon)^{r_\ast}$, where $C$ and $a$ are fitting parameters.}
    \label{fig:bandstructure}
\end{figure}
Our goal in this paper is to rigorously show that a similar result holds for non-resonant \emph{interacting many-body} quantum systems, where the simplifications due to band theory are no longer applicable.
 
\section{Main results}
We study interacting qubits, placed on the vertices of a finite bounded degree graph $\Lambda$, which we assume is $d$-dimensional in the sense that any vertex has at most $\sim r^d$ vertices within graph distance $r$.  Unitary dynamics is generated by a time-independent Hamiltonian $H$:
\begin{align}
H = H_0 + \sum_{r=1}^\infty\epsilon^r V_r \label{eq:H_main_sum}
\end{align}
where $H_0$ is a geometrically-local Hamiltonian that contains only Pauli $Z$s, and $V_r$ consists of a sum of bounded operators acting on sites which are at most  graph distance $r$ apart.  The precise definitions and assumptions can be found in the Supplementary Material.

Our goal is to derive Lieb-Robinson bounds on the resulting dynamics, which are usually phrased in terms of the operator norm (largest singular value of its argument) of a real-time commutator:
\begin{align}
\Vert [B_y, A_x(t)] \Vert \lesssim \lVert B_y\rVert \lVert A_x\rVert \times  \exp(\mu[\LRvel t - \mathsf d(x, y)])
\end{align}
where $A_x, B_y$ act nontrivially only on $x,y \in \Lambda$ respectively and $A(t)=e^{\mathrm{i}Ht}Ae^{-\mathrm{i}Ht}$ is Heisenberg-picture time evolution. The Lieb-Robinson bound defines an emergent ``lightcone" in the system with velocity $\LRvel$.  To understand why, it is useful to consider the following thought experiment:  take an arbitrary many-body state $|\psi\rangle$, and then consider a similar state perturbed locally near site $y$: $|\psi^\prime\rangle = \mathrm{e}^{\mathrm{i}\delta B_y}|\psi\rangle$ with $\delta$ a small parameter.  After evolving both states for time $t$ with Hamiltonian $H$ to $|\psi(t)\rangle := \mathrm{e}^{-\mathrm{i}Ht}|\psi\rangle$ and $|\psi'(t)\rangle := \mathrm{e}^{-\mathrm{i}Ht}|\psi'\rangle$, we ask whether an observable near site $x$ can detect the perturbation we made at $y$. Formally expanding to first order in $\delta$, one finds that
\begin{equation}
|\langle \psi^\prime(t)|A_x|\psi^\prime(t)\rangle - \langle \psi(t)|A_x|\psi(t)\rangle| \le     \delta \Vert [B_y, A_x(t)] \Vert + \cdots . 
\end{equation} Since $A_x$ and $B_y$ are arbitrary, we conclude that the Lieb-Robinson bound controls the generation of entanglement and quantum correlations \cite{bravyi2006lieb}.  Moreover, Lieb-Robinson bounds also control the clustering of correlations in a gapped ground state \cite{hastings2006spectral,exponential_clustering} and classical and quantum simulation complexity \cite{haah2021quantum,mcdonough2025, Osborne_2007, osborne2006}; see \cite{chen2023speed} for a recent review.

It is well-established that for the many-body systems of interest in this work, $v_{\mathrm{LR}}\lesssim \epsilon$.  At a high-level, one can gain intuition for this result by considering a formal series.  For simplicity, suppose that we have only nearest-neighbor interactions.  If $x$ and $y$ are far apart on the lattice, then if $A_x(t) = A_x + \mathrm{i}[H,A_x]t + \cdots $,  we will need to pull down at least $r=\mathsf{d}(x,y)$---the distance between $x$ and $y$---powers of $H$ in order for $A_x(t)$ to act non-trivially on site $y$, which is required in order for the commutator with an operator acting only on $y$ to be nonzero; the leading-order term in the Taylor series that contributes has a coefficient of $\epsilon^r t^r/r! \sim (\epsilon t/r)^r \ll  \mathrm{e}^{\epsilon t - r}$.  

Motivated by our single-particle example, it makes sense to postulate that a \emph{non-resonance} condition on $H_0$ may shrink $v_{\mathrm{LR}}$ to a higher order in $\epsilon$.  
More precisely, we introduce the following definition for a non-resonant classical Hamiltonian $H_0$ at scale $r_*$, formalizing \eqref{eq:intrononresonant}.

\begin{shadedtheorem}
\begin{defn}[Definition \ref{defn:noresonance}, schematic]\label{defn:schematic_noresonance}
A geometrically $k-$local Hamiltonian 
\[
H_0= \sum_{I:\diam(I) \leq k}h_IZ_I
\]
satisfies a $(h, r_\ast, \Delta)$ non-resonance condition if $\sup_{I} |h_I|\leq h$, and if the restriction of $H_0$ to any subset within any closed ball of radius $r_\ast > k$ has a minimal spectral gap $h\Delta$ anywhere in the spectrum. 
\end{defn}
\end{shadedtheorem}

This condition only concerns the explicitly solvable $H_0$ and thus avoids assumptions such as limited level repulsion from \cite{imbrie2016many}, which are not proven for any specific choices of $H$ of interest.  In contrast, we explicitly show in the appendices that there exist models that are non-resonant in the above sense. Indeed, we give a   deterministic construction of a non-resonant potential, the dyadic-triadic scheme, which avoids resonances by modular arithmetic arguments and ensures non-resonance \eqref{eq:intrononresonant} with $\xi=d$. Having established that non-resonant potentials exist, we now present our main result: 

\begin{shadedtheorem}
\begin{thm}[Thm.~\ref{thm:main_thm}, schematic]
\label{thm:schematic_main_thm}
On a finite graph $\Lambda$, consider the time evolution generated by $H$ given in \eqref{eq:H_main_sum}.
If $H_0$ is $k$-local and satisfies a $(h, \Delta, r_\ast)$ non-resonance condition, 
then for any $\alpha < \frac{1}{5k+1}$ there exist  constants $C, C'$ such that for any $S_0, S_1 \subseteq \Lambda$ and $r = \min(r_\ast+1, \mathsf d(S_0,S_1))$, we have
\begin{align}
\Vert [B_{S_0}, A_{S_1}(t)] \Vert \leq \Vert B_{S_0}\Vert \Vert A_{S_1}\Vert  \times C\min(|\partial S_0|,|\partial S_1|)\qty[\exp(C'ht \tilde \epsilon^{\alpha ( r/8 - k)})-1]\exp(-\tfrac{\alpha}{8}\mathsf d(S_0, S_1)) \label{eq:thmmaintext}
\end{align}
as long as $\tilde \epsilon \sim \frac{\epsilon V_\ast^2} {\Delta^2}<1$ where $V_\ast=\sup_{x\in \Lambda}|B_{r_\ast}(x)|$ and $B_{S_0},A_{S_1}$ are any bounded operators supported within $S_0, S_1 \subset \Lambda$ respectively.
\end{thm}
\end{shadedtheorem}

Then as in \eqref{eq:intrononresonant}, if $\Delta \sim \e^{-b \cdot r_\ast^{\xi}}$, we can take $r_\ast \sim \log^{\frac{1}{\xi}}(\epsilon^{-1})$ to find that information can spread at most with a non-perturbatively slow velocity. The conclusion of Thm.~\ref{thm:schematic_main_thm} holds for potentials which only approximately obey Def.~\ref{defn:schematic_noresonance} up to modifying the constant $\frac{1}{8}$ in $ \epsilon^{\alpha ( r/8 - k)}$, which shows the result is robust and  allows us to apply it to non-fine-tuned examples.  At early times, \eqref{eq:thmmaintext} implies that correlations between two qubits separated by distance $r<r_*$ require a time $t\sim \epsilon^{-r}$ to be established, just as in prethermal MBL.  A Lieb-Robinson bound for this early-time phenomenon was recently derived in certain random spin chains \cite{elgart2024slow}.

The full result only requires the weaker assumption of partial non-resonance (formalized in Def.~\ref{defn:partialnr}), in which one allows the failure of Definition \ref{defn:schematic_noresonance} along a small fraction of paths on the lattice. This weakening of the non-resonance assumption is important because it allows us to cover several further examples of physical interest.  Focusing on the simplest case $H_0 = \sum_v h_v Z_v$ where we have strong on-site potentials that act independently on each qubit, we prove that if $h_v$ are independent and identically distributed Gaussian\footnote{We remark that the argument straightforwardly extends to non-Gaussian random variables with the property that i.i.d.\ sums have an explicit continuous distribution, e.g., uniform or exponential random variables. However, having a central limit theorem is not sufficient, as we need a fine anti-concentration bound that is much more precise than the Berry-Esseen error term.}
random variables, $H_0$ satisfies a partial non-resonance condition for any $r_*$ with high probability with a spectral gap \begin{equation}
    \Delta \sim \exp(-r_*^d).  \label{eq:Deltascaling}
\end{equation}
The non-resonant condition thus holds for $r_* \sim \log^{1/d}(\epsilon^{-1})$, i.e. we obtain \eqref{eq:main_rstar} with $\xi=d$.
We can easily see that \eqref{eq:Deltascaling} is optimal: in a ball of radius $r_*$, we must pack $\exp(\mathrm{O}(r_*^d))$ energy levels of $H_0$ into an interval of size $\mathrm{O}(r_*)$ on the real line.  
Another important example is where \begin{equation}
    h_{\mathbf{n}} = \cos(2\pi \boldsymbol{\alpha}\cdot \mathbf{n})
\end{equation}
where $\mathbf{n} \in \mathbb{Z}^d$ label the sites of a $d$-dimensional cubic lattice.  If  $\boldsymbol{\alpha}\in [0,1]^d$ is chosen uniformly at random, we prove that with high probability $h_{\vb n}$ satisfies the partial non-resonant condition in many regions with $\Delta \sim \exp(-r_*^{d+1})$.  This proof is based on algebraic complexity arguments (counting parameters of polynomials) and small-denominator estimates which are also relevant in KAM theory \cite{kolmogorov2005preservation}. We conjecture that the stronger result \eqref{eq:Deltascaling} holds for incommensurate lattices as well.

To illustrate the idea behind the proof, first consider a system of linear size $r_\ast$, a fixed minimal gap $\Delta$, and a nearest-neighbor perturbation $V_\epsilon = \epsilon V_1$. If $\epsilon \ll \Delta$, then perturbation theory converges to all orders. 
The bare eigenkets $\ket{n^{(0)}}$ of $H_0$ will receive a correction at each order $j$ in $\epsilon$ by $V^j$, which means that $\ket{n} = U\ket{n^{(0)}}$, where $U = 1 + \sum_{j=1}^\infty\epsilon^j O_j$, with $O_j$ acting on  $\le j+1$ sites.

Now, consider sites $x,y$ a distance $r \gg r_\ast$ apart. 
If $r = r_\ast n_\ast$, then we can divide the region separating $x$ from $y$ into $n_\ast$ blocks of size $r_\ast$. 
In the rotated basis, the coupling $V_{i,i+1}$ between two systems is transformed into $U^\dagger V_{i,i+1}U = \sum_{j=0}^\infty\epsilon^{j}V^{(i)}_j$, where $V_{j}^{(i)}$ is a $j+1$-local operator, as depicted in Fig.~\ref{fig:intro_proof_sketch}.
Formally expanding the time-evolution of $A_x(t)$, the lowest-order terms that contribute to the Taylor series are 
\begin{equation}
\frac{t^{n_\ast}}{n_\ast!}[\epsilon^{r_1}V^{(1)}_{r_1}, [\epsilon^{r_2} V^{(2)}_{r_2}, [\dots, [\epsilon^{r_{n_\ast}} V_{r_{n_\ast}}^{(n_\ast)}, A_x]]] \sim \frac{\epsilon^r t^{n_\ast}}{n_\ast!} \ll \e^{\epsilon^{r_\ast}t-r} \ ,
\end{equation}where $r_1 + r_2 + \dots + r_{n_\ast} = r$.
Thus, the Lieb-Robinson bound allows us to show that for non-perturbatively long times, 
the quantum dynamics are well-approximated by dephasing evolution under the commuting Hamiltonian $U^\dagger HU = H_0 + \sum_{j=1}^\infty\epsilon^j Z_j$ where $Z_j$ is also a $j$-local operator.  
This is exactly the expected mechanism leading to the phenomenology of MBL \cite{huse2014phenomenology}, and our result proves that this phenomenology manifests on prethermal time scales in any dimension. On the much longer time scales where MBL is expected to break down in dimensions larger than $1$, our result proves that one still has to wait a time $t_* \gtrsim n_\ast\e^{\log^{1+1/\xi}(1/\epsilon)}$ for information to propagate through all of the intervening regions, representing an anomalously slow ballistic spread of information, even if there is no outright localization.

\begin{figure}[t]
    \centering
    \def\svgwidth{0.95\linewidth} 
    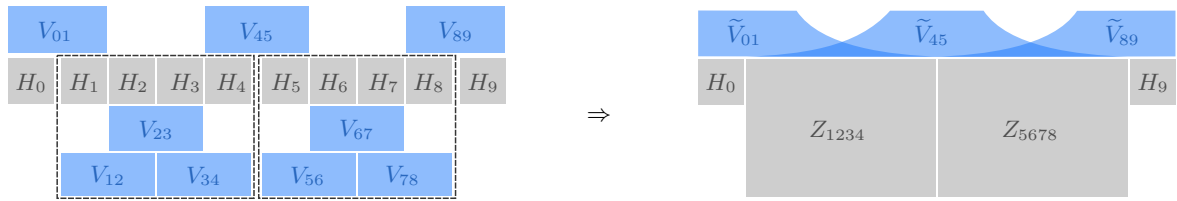
    \caption{This figure illustrates the idea behind the proof on a one-dimensional chain with 10 sites; After breaking the system into regions of size $r_\ast$, time-independent perturbation theory is applied to diagonalize each, which results in the terms from $V$ that couple the regions together spreading out with tails that decay as $\sim \epsilon^{r}$. Lieb-Robinson bounds can then be applied to the transformed system on the right to bound the coupling between sites $0$ and $9$.}
    \label{fig:intro_proof_sketch}
\end{figure}

\section{Implications}
Our results have a number of immediate, and perhaps surprising, implications about quantum dynamics in disordered systems.  As we emphasized in the introduction, there has been an enormous effort to characterize the proposed many-body localized regime, where strongly disordered one-dimensional models would fail to thermalize.  Of course, to truly check \emph{eigenstate localization} in the thermodynamic limit seems completely impossible experimentally, since the number of eigenstates is exponentially large and the time scales to distinguish similar eigenstates are exponentially long.  

Therefore, more practical signatures of MBL were proposed, such as the logarithmically slow dynamical growth of quantum entanglement \eqref{eq:texpr}. From the form of \eqref{eq:thmmaintext}, we immediately see that for $r_*$ sufficiently large, we indeed obtain \eqref{eq:texpr} up to distances \eqref{eq:main_rstar}, which diverge as $\epsilon \rightarrow 0$.  The timescales over which the physics will look like MBL are \begin{equation}
    t\lesssim t_* \sim \epsilon^{-r_*}\sim \epsilon^{-\log^{1/\xi}\epsilon^{-1}}.
\end{equation}
Note that $t_*$ grows superpolynomially as $\epsilon \rightarrow 0$.  Given the fairly short coherence times of quantum simulators, our results show that the (at most) logarithmically slow growth of entanglement is a \emph{universal} phenomenon in all non-resonant systems.

Excitingly, the relevant non-resonant condition holds both for models with strong disorder and in incommensurate lattice potentials, suggesting that both may exhibit similar quantum dynamics on short time scales, such as those accessible in experiments using quantum simulators.  Indeed, some experiments \cite{hur2025stabilitymanybodylocalizationdimensions} which looked for signatures of MBL replaced disorder with an incommensurate lattice, which is easier to realize, with \cite{sarang_quasiperiodic,huse_quasiperiodic} arguing that MBL may  be stable in $d = 2$ with incommensurate lattices. 
We have proved that---at least in the limit $\epsilon \rightarrow 0$---each class of models is guaranteed to exhibit the MBL phenomenology \eqref{eq:texpr} at early times, in any dimension.

Moving beyond MBL, our results also imply that non-resonant many-body systems are \emph{more tractable} for classical or quantum simulation than generic many-body systems, in a certain formal sense. Operator locality provides a quantifier of the error from truncating dynamics to a finite region, leading to a straightforward classical simulation algorithm \cite{chen2023speed, mcdonough2025}. In the appendices, we describe how many practical computations can be done in disordered models using $\exp[\mathrm{O}(\min(\log(t),\epsilon^{r_*}t)^d)]$ resources on classical computers.  The strong $\epsilon$-dependence in the bound shows that non-resonant systems are easier to simulate with provable accuracy than typical systems.  In this respect, while quantum advantage may be found both in the simulation of non-resonant or resonant many-body systems, our results do call into question the extent to which quantum supremacy can be demonstrated with a non-resonant model.  Of course, our theorem does not rule out that efficient approximations, such as hydrodynamics \cite{tiborhydro}, may give very accurate descriptions of typical chaotic systems, rendering them easier to simulate for many practical purposes.

It has been argued \cite{hartman} that in systems with a Lieb-Robinson velocity $v_{\mathrm{LR}}$, the time scale $\tau_{\mathrm{eq}}$ after which hydrodynamic diffusion (with diffusion constant $D$) onsets obeys $D \lesssim v_{\mathrm{LR}}^2\tau_{\mathrm{eq}}$.  Re-writing this inequality and using our rigorous bound that correlations over distance $r<r_*$ require time $\tau(r)\sim \epsilon^{-r}$ to be prepared, we deduce that unless $\tau_{\mathrm{eq}} \gg \epsilon^{-r_*}$, the diffusion constant $D \lesssim r_*^2 \epsilon^{r_*}$ is non-perturbatively small in any non-resonant model.  $D$ being non-perturbatively small as $\epsilon \rightarrow 0$ does seem reasonable, but we cannot rule out the possibility of a large diffusion constant if $\tau_{\mathrm{eq}}$ is exceptionally long. Interestingly, there are classical models where a non-resonance condition can be used to prove that diffusion constants are very small \cite{Huveneers_2013,De_Roeck_2014}.\footnote{While intuitively non-resonance in these classical models plays the same role as in our work, we note that in classical Hamiltonian mechanics, it is never possible for a non-resonance condition to hold in \emph{every state}; it can only hold in typical states.  The existence of strong Lieb-Robinson bounds (which hold in every state) for non-resonant models is a uniquely \emph{quantum} effect.}

\section{Outlook}
We have established rigorous bounds on \emph{many-body} quantum dynamics that show how destructive quantum interference slows down the spreading of quantum information in strongly disordered, and other non-resonant, systems. Previously observed dynamical ``signatures" of MBL on experimentally accessible time scales are in fact universal phenomena in non-resonant systems, so our conclusions may explain why many experimental results \cite{Choi_2016,Li:2025kje,hur2025stabilitymanybodylocalizationdimensions} look like MBL in regimes where it is not expected to exist without fine-tuning (e.g. $d=2$).   Constant prefactors in our bounds are not tight, so we expect that in practice the very slow dynamics implied by our main theorem persists to much larger values of $\epsilon$ than we can guarantee.\footnote{We remark that in the KAM Theorem, which guarantees the stability of integrable systems to small perturbations, there is a similarly vast gulf between mathematical guarantees and practical observations of how large a perturbation needs to be to destroy integrability.}

Looking forward, the methods we have developed should generalize to many more settings.  For example, a quantum CSS error-correcting code \cite{chaocode,vedikacode} can be interpreted as a many-body Hamiltonian, where each term in $H$ is a stabilizer of the codespace.  Weighting each of these terms by non-resonant coefficients, we conjecture that a similar slowdown of quantum dynamics arises in these models, which may help to protect quantum information.  In this example, and many others, a global symmetry may imply exact degeneracies of the non-resonant $H_0$, but following \cite{Yin:2024hjm}, only \emph{local} non-resonant conditions are expected to be necessary to generalize our theorem.   

Systems such as incommensurate lattices or disordered potentials are realizable in experiments, such as quantum simulators based on cold atoms. This opens broad experimental avenues for testing the slowing down of transport through non-resonance.  Another important generalization of our work will be to systems with power-law interactions, where slow dynamics due to non-resonant conditions have been found \cite{longrange_metastab_Ising,longrange_metastab_spec,lerose}.  Strong Lieb-Robinson bounds for systems with power-law interactions do exist \cite{Chen:2019hou,Kuwahara:2019rlw,Tran:2021ogo,lemm2025quantum}, but we expect that generalizing our results to power-law models with a non-resonance condition will be a formidable technical effort.

\section*{Acknowledgements}
We thank Wojciech de Roeck, Francois Huveneers and Rahul Nandkishore for useful discussions.  This work was supported in part by the Department of Energy under Quantum Pathfinder Grant DE-SC0024324 (BTM, AL), and by the Department of Defense  through the National Defense Science \& Engineering Graduate  Fellowship Program (BTM). The research of ML is supported by the DFG through the grant TRR 352 – Project-ID 470903074 and by the European Union (ERC Starting Grant MathQuantProp, Grant Agreement 101163620).\footnote{Views and opinions expressed are however those of the authors only and do not necessarily reflect those of the European Union or the European Research Council Executive Agency. Neither the European Union nor the granting authority can be held responsible for them.}

\begin{appendix}
\renewcommand{\thesubsection}{\thesection.\arabic{subsection}}
\renewcommand{\thesubsubsection}{\thesubsection.\arabic{subsubsection}}
\section{Mathematical preliminaries}
This appendix collects some established facts which will be useful in the proofs of our main results.

\subsection{Graphs}

In this paper, we are studying many-body quantum models of interacting qubits.   It is often useful to consider each qudit as associated with the vertices of an unoriented finite graph $(\Lambda,E)$, namely the Hilbert space of interest is $(\mathbb{C}^D)^{\otimes \Lambda}$ if $\Lambda$ is a finite set and $D$ is the qudit dimension. Although we restrict to a finite graph to avoid unbounded operators, our bounds will hold uniformly in the size of the graph. We will focus on models defined on a graph embeddable in $d$-spatial dimensions.

\begin{defn}[Graph distance]
Given $u,v\in \Lambda$, let $\mathsf{d}(u,v)$ denote the graph distance between two vertices on the edge set $E$.  Define the ball \begin{equation}
    B_r(u) := \lbrace v\in \Lambda : \mathsf{d}(u,v) \le r\rbrace. \label{eq:ballclosed}
\end{equation}
\end{defn}
\begin{defn}[$d$-dimensional graph] \label{defn:d-dim}
    We say that $(\Lambda,E)$ is a $d$-dimensional graph (or lives in $d$ spatial dimensions) if there exists $0<M<\infty$ such that for all $u\in \Lambda$, and all $r$, \begin{equation}
        |B_r(u)| \le Mr^d. \label{eq:Mdimension}
    \end{equation}
\end{defn}

\begin{defn}[$d$-dimensional square lattice]
    The $d$-dimensional cubic lattice has $\Lambda=\mathbb{Z}^d$ and $(x,y)\in E$ if and only if $|x-y|_1=1$, namely $x$ and $y$ differ in exactly one component by $\pm 1$.   This edge set is implicit henceforth when we refer to the lattice $\mathbb{Z}^d$.
\end{defn}

\begin{prop}\label{prop:Msquarelattice}
    For the $d$-dimensional lattice $\mathbb{Z}^d$, \begin{equation}
        M < \min[(d+1)(d+2), 3\times 2^{d-1}]. \label{eq:ZdMbound}
    \end{equation}
\end{prop}
\begin{proof}
    This is a simple exercise in combinatorics.  Let \begin{equation}
        N_d(r_*):= |\lbrace x\in\Lambda : \mathsf{d}(x,0)\le r_* \rbrace |.
    \end{equation}
    Defining a slightly modified function \begin{equation}
        N^\prime_d(r_*):= |\lbrace x\in\Lambda : \mathsf{d}(x,0)\le r_* \rbrace, x_j \ge 0 |,
    \end{equation}
    we notice that \begin{equation}
       N^\prime_d(r) > 2^d N_d(r) 
    \end{equation}
    while we can use the method of generating functions to exactly compute $N^\prime_d(r_*)$: \begin{equation}
        G_d(x):=\sum_{r=0}^\infty x^r N^\prime_d(r) = \frac{1}{(1-x)^{d+1}}.
    \end{equation}
    Hence we see that \begin{equation}
        N_d(r) < 2^d \frac{(r+d)!}{r!d!},
    \end{equation}
    and in particular that \begin{equation}
        \frac{N_d(r)}{r^d} < \prod_{j=1}^d \left(\frac{2}{r}+\frac{2}{j}\right).
    \end{equation}
    This bound decreases as $r$ increases.  It is best to avoid using this bound at $r=1$ as it is quite loose and the exact answer $N_d(1)=2d+1$ is obvious.  But we already get a ``reasonable" bound for $r=2$, where $r^{-d}N_d(r)= (d+2)!/d!$, which leads to \eqref{eq:ZdMbound}. 
    
    When $d \leq 4$, a tighter bound on $M$ can be found via the following argument: we can decompose a $d$-dimensional ball into lower dimensional balls as
    \begin{align}
    N_d(r) = 2\sum_{q = 1}^r N_{d-1}(r-q) + N_{d-1}(r) \ .
    \end{align}
    Upper-bounding $N_d(r) \leq M_d r^d$ for some $M_d$, for $d > 1$ we readily find the bound $M_{d} \leq 2M_{d-1}$, with the base case $M_1 = 3$.
\end{proof}

\subsection{Operator locality}

With these notions of locality coming from the graph $(\Lambda,E)$, we now turn to characterizing local operators acting on the global Hilbert space $(\mathbb C^D)^{|\Lambda|}$.

\begin{defn} [Range and locality]
We say that an operator $O$ acting on $(\mathbb C^D)^{|\Lambda|}$  is strictly local if it can be written as $O = \sum_S O_S$, where each $O_S$ is supported within a subset $S$ and every $\diam(S)$ is smaller than a finite number. Given an operator $O$, we define $\rng(O) = \inf_{\{S\}}\sup_S \diam(S)$, where the infimum is taken over local decompositions $O = \sum_S O_S$. Similarly, we define the locality $\loc(O) = \inf_{\{S\}} \sup_S |S|$.
\end{defn}
For an operator that is not strictly local, operator locality can be quantified by the $\kappa$-norm, as defined in Ref.~\cite{chen2021operator}:
\begin{defn}
Let $O$ be an operator and $\kappa\geq 0$. The $\kappa$-norm is defined as
\begin{align}
\Vert O \Vert_{\kappa} = \inf_{\{O_S\}}\sup_{v \in \Lambda}  \sum_{S \ni v} \e^{\kappa \diam(S)}\Vert O_S\Vert \label{eq:kappanormdef}
\end{align}
where $\Vert \cdot \Vert$ is the operator norm (maximum singular value of the
argument) and the infimum is taken over decompositions $\sum_{S\subseteq \Lambda}O_S$, where $O_S$ is supported within $S$.
\end{defn}
We notice that the $\kappa$-norm of a local operator grows at most exponentially with $\kappa$:\footnote{We note here that the analysis of the locality of the generator also works with a stronger version of the $\kappa$-norm with volume-law scaling, defined $\Vert O \Vert_\kappa = \inf_{\{O_S\}} \sup_{v \in \Lambda}\sum_{ S \ni v}\e^{\kappa (|S|-1)} \Vert O_S \Vert$, which may prove useful for other applications. However, the weaker $\kappa$-norm above is sufficient for LR bounds and allows us to use perturbations of the form $\epsilon^{\rng(V)}V$ instead of $\epsilon^{\loc(V)}V$.  The volume-law tailed Lieb-Robinson bounds recently developed in \cite{mcdonough2025} would be necessary to generalize our analysis to account for these stronger tails.}
\begin{prop}
If $O$ is a local operator, then $\Vert O \Vert_{\kappa} \leq \e^{\kappa \rng(O)}\Vert O \Vert_{\kappa = 0}$.
\label{prop:kappa_norm_scaling}
\end{prop}
Next, we show that the $\kappa$-norm is indeed a well-defined norm:
\begin{lem}
The $\kappa$-norm defines a norm on operators.
\end{lem}

\begin{proof}
Positive-definiteness and absolute-value homogeneity follow directly from the same properties of the operator norm. To prove the triangle inequality, let $A = \sum_{S \in \mathcal C_1}A_S$ and $B = \sum_{S \in \mathcal C_2}B_S$ be the decompositions which realize the $\kappa$-norms within a tolerance $\delta_1, \delta_2$. This gives a decomposition $A+B = \sum_{S \in \mathcal C_1}A_S + \sum_{S \in \mathcal C_2}B_S$, so we have
\begin{align}
\Vert A + B \Vert_{\kappa} \leq \sum_{S \in \mathcal C_1}\e^{\kappa \diam(S)}\Vert A_S\Vert + \sum_{S \in \mathcal C_2}\e^{\kappa\diam(S)}\Vert B_S \Vert \leq \Vert A \Vert_{\kappa} + \Vert B \Vert_{\kappa} + \delta_1 + \delta_2
\end{align}
Then taking $\delta_1, \delta_2 \to 0$ proves the triangle inequality.
\end{proof}
Central to our bounds is the notion of submultiplicativity. Many norms do not
satisfy this property, with the most notable example being the Frobenius
norm $\Vert O \Vert_{\rm F} = \sqrt{\frac{1}{d}\sum_{ij}|O_{ij}|^2}$.
We find that the $\kappa$-norm is \emph{almost} submultiplicative, in that we
incur an additional multiplicative factor of system volume (much smaller than
the Hilbert space dimension of the system):
\begin{lem}
Let $A, B$ be two operators. Then 
\begin{align}
\Vert A B \Vert_{\kappa} \leq [\loc(A)+\loc(B)]\Vert A \Vert_{\kappa} \Vert B \Vert_{\kappa}
\end{align}
\label{lem:submultiplicative}
\end{lem}

\begin{proof}
Let $u$ be the vertex which maximizes $\Vert  A B\Vert_{\kappa}$.
\begin{align}
\Vert AB\Vert_{\kappa} &\leq \sum_{S_1 \ni u}\sum_{S_2: S_2 \cap S_1 \neq \emptyset}\e^{\kappa \diam(S_2 \cup S_1)} \Vert A_{S_1}\Vert \Vert B_{S_2}\Vert + \sum_{S_2 \ni u}\sum_{S_1: S_1 \cap S_2 \neq \emptyset}\e^{\kappa \diam(S_1 \cup S_2)}\Vert \Vert A_{S_1}\Vert \Vert B_{S_2}\Vert \notag \\
&\leq \sum_{S_1 \ni u}\e^{\kappa \diam(S_1)
}\Vert A_{S_1} \Vert \sum_{v \in S_1}\sum_{S_2 \ni v}\e^{\kappa \diam(S_2)}\Vert B_{S_2}\Vert + \sum_{S_2 \ni u} \e^{\kappa \diam(S_2)}\Vert B_{S_2}\Vert\sum_{v \in S_2}\sum_{S_1 \ni v}\e^{\kappa \diam(S_1)}\Vert A_{S_1}\Vert \notag \\
&\leq \sum_{S_1 \ni u}\e^{\kappa \diam(S_1)
}\Vert A_{S_1} \Vert |S_1|\Vert B \Vert_{\kappa} + \sum_{S_2 \ni u} \e^{\kappa \diam(S_2)}\Vert B_{S_2}\Vert|S_2| \Vert A\Vert_{\kappa} \notag \\
&\leq \Vert B \Vert_{\kappa}\loc(A)\sum_{S_1 \ni u}\e^{\kappa \diam(S_1 
)}\Vert A_{S_1} \Vert  + \Vert A\Vert_{\kappa}\loc(B)\sum_{S_2 \ni u} \e^{\kappa \diam(S_2)}\Vert B_{S_2}\Vert \notag \\
&\leq [\loc(A) + \loc(B)]\Vert A\Vert_\kappa \Vert B \Vert_{\kappa}
\end{align}
In the first line we used the freedom to choose any local decomposition to sum over $S_1, S_2$ in the decomposition which realizes the $\kappa-$norms of $A$ and $B$.
\end{proof}

 In particular, if $A$ and $B$ are both supported on a region of size $V$, then $\Vert AB \Vert_{\kappa} \leq 2V \Vert A \Vert_\kappa \Vert B \Vert_\kappa$.
Since the minimum gap $\Delta$ must scale at best exponentially with
$V$, this extra factor of $V$ will be of
little consequence. 

\subsection{Lieb-Robinson bounds\label{sec:lrbound_duhamel}}
\begin{prop}[Duhamel identity]
Let $A$ and $B$ be bounded linear operators. Then
\begin{align}
\e^{(A + B)t} = \e^{At} + \int\limits_0^t \dd s \e^{(A+B)(t-s)}B \e^{A s}
 \label{eq:duhamel}
\end{align}
\end{prop}

The above identity can easily be proven by showing the two sides satisfy the same differential equation. However, it is also illuminating to look at the Taylor series expansion to understand the relationship between LR bounds and so-called ``irreducible paths" of terms in the Hamiltonian. To expound on this, we will work with time-dependent operators instead. We write $\mathcal T$ for the usual time-ordering operator.

\begin{prop}[Time-dependent Duhamel identity]
Let $A:\mathbb R \to \mathcal B(\mathcal H)$ and $B:\mathbb R \to \mathcal B(\mathcal H)$ be time-dependent operators valued in $\mathcal B(\mathcal H)$, the space of bounded linear operators on a Hilbert space $\mathcal H$. Then the following identity holds:
\begin{align}
\mathcal T\e^{\int_0^t \dd s (A(s) + B(s))} = \mathcal T\e^{\int_0^t \dd sA(s)} + \int\limits_0^t \dd s \mathcal T\e^{\int_s^t \dd u (A(u)+B(u))}B(s) \mathcal T\e^{\int_0^s \dd uA(u)}
\end{align}
\end{prop}
\begin{proof}
We begin by expanding the Taylor series of the left-hand side. First, for a time-dependent operator $F$, let 
\begin{align}
F^{(n)}(t,s) \equiv \int_s^t \dd s_n \int_{0}^{s_1}\dd s_2 \dots \int_0^{s_{n-1}}\dd s_n F(s_n)\dots F(s_2)F(s_1)
\end{align}
with $F^{(0)}(t) = 1$.
By definition, $U(t) = \mathcal T\exp(\int_0^t \dd s F(s))$ solves the differential equation $\dv{t} U = F(t)U(t)$. By repeatedly integrating this equation, we get the classic Dyson series solution:
\begin{align}
\mathcal T\exp(\int_s^t \dd u F(u)) = \sum_{n=0}^\infty F^{(n)}(t,s)
\end{align}
Applying this to the LHS of \eqref{eq:duhamel}, we have
\begin{align}
\mathcal T\e^{\int_0^t \dd s A(s) + B(s)} = \sum_{n\geq 0}\sum_{\vec X}\int_{0}^{t}\dd s_n \dots \int_0^{s_2}\dd s_1 X_n(s_n)\dots X_1(S_1)
\end{align}
where $\vec X = (X_1, \dots, X_n) \in \{A,B\}^{ n}$ index all possible sequences of $A, B$ of length $n$.

Now, we rewrite the sum by sorting the sequences by the index of the first $B$:
\begin{align}
\mathcal T\e^{\int_0^t \dd s A(s) + B(s)} &=  \sum_{n}A^{(n)}(t,0)+ \sum_{n,m}\sum_{\vec X}\int_{0}^{t}\dd s_n \dots \int_0^{s_{m+1}}\dd s_m X_n(s_n)\dots X_{m+1}(s_{m+1})B(s_m)A^{(m)}(s_m)
\end{align}
The first term corresponds to all of the sequences in which no $B$ ever appears, and the sum over $m$ in the second indexes the first place in the sequence $B$ is found. We can continue to manipulate the second term:
\begin{align}
\sum_{n,m}\sum_{\vec X}\int_{0}^{t}\dd s_n \dots &\int_0^{s_{m+1}}\dd s_m X_n(s_n)\dots X_{m+1}(s_{m+1})B(s_m)A^{(m)}(s_m) \notag \\
&=\sum_{n}\sum_{\vec X}\int_{0}^{t}\dd s_n \dots \int_0^{s_2}\dd s_1 \int_{0}^{s_1} \dd s X_n(s_n)\dots X_{1}(s_1)B(s)\mathcal T\e^{\int_0^{s} \dd u A(u)} \notag \\
&=\sum_{n}\sum_{\vec X}\int_{0}^{t}\dd s \qty(\int_{s}^t \dd s_n \dots \int_{s}^{s_{2}}\dd s_{1} X_n(s_n)\dots X_{1}(s_1))B(s)\mathcal T\e^{\int_0^{s} \dd u A(u)} \notag \\
&=\int_{0}^{t}\dd s \mathcal T \e^{\int_0^{s}\dd u A(u)+B(u)}B(s)\mathcal T\e^{\int_0^{s} \dd u A(u)},
\end{align}
which completes the proof.
\end{proof}
What this proof reveals is that the Duhamel expansion provides a systematic way
to count sequences of terms in the Taylor series of $\e^{-\ii H t}$ when $H =
\sum_{S}H_S$.
To avoid the proliferation of integrals, we introduce some new
notation. First, given an operator $F$, let $\mathcal F \equiv [F, \cdot]$ be
the adjoint action of $F$. Let $\ast$ represent convolution, and $\e(\mathcal F)(t)
= \exp(\ii \mathcal F t)$ represent exponentiation. In this notation, the
time-independent Duhamel identity
may be repackaged as
\begin{align}
 \e(A+B)  = \e(A) +\e(A+B) \ast B \e(A) = \e(A)+ \e(A+B)B \ast \e(A)
\end{align}
which makes transparent the function of the Duhamel identity as a way to count paths.
Using this fact, we will now prove the following:\footnote{Although a Lieb-Robinson bound of this form is not a new result, the proof presented here is new and contains an intuitive perspective, which will be very helpful for us in this paper.}

\begin{prop}[Theorem 3 in \cite{chen2021operator}]
\label{prop:irreducible_paths}
Suppose that $H = \sum_{S}H_S$ is a local Hamiltonian. The supports $S$ of the terms $H_{S}$ form the vertices of a graph, with edges between $S_1$ and $S_2$ if they intersect nontrivially. Given operators $A_x, B_y$ as above, we have the following bound
\begin{align}
\Vert [B_y, A_x(t)] \Vert \leq \Vert A_x \Vert \Vert B_y \Vert\sum_{\vec S \in \Gamma(x \to y)}\frac{(2t)^{|\Gamma|}}{|\Gamma|!}\omega(\vec S)
\end{align}
where the sum is taken over ``irreducible," or non-self-intersecting, paths $\Gamma(x \to y)$ through the graph for which the initial set contains $x$ and the final set contains $y$. The function $\omega$ is the weight of the path, defined as
\begin{align}
\omega(\vec S) = \prod_{i}\Vert H_{S_i} \Vert
\end{align}
\end{prop}

\begin{proof}
Given an operator $F$, let $\mathcal F \equiv [F, \cdot]$. First, consider the decomposition $H = H_1 + H_1'$, where $H_1 = \sum_{S \ni y}H_S$ and $H_1' = \sum_{S \not\ni y}H_S$. We then apply the Duhamel expansion:
\begin{align}
  \mathcal B_y\e(\mathcal H)A_x &=  \mathcal B_y\qty[\e(\mathcal H_1') + \ii\e(\mathcal H)\ast \mathcal H_{1}\e(\mathcal H_1')]A_x 
                                  = \ii\mathcal B_y\sum_{S_1 \ni y}\e(\mathcal H)\ast\mathcal H_{S_1}\e(\mathcal H_1')A_x 
\end{align}
The first term vanishes because $[\e(\mathcal H_1')A_x](t)$ is not supported on $y$, so the commutator with $B_y$ vanishes. Next, for each $S_1$ in the sum above, we expand $H_1^\prime = H_2' + H_2$, where $H_2 = \sum_{S:S \cap S_1 \neq \emptyset}(H_1')_S$ and $H_2' = \sum_{S : S\cap S_2 = \emptyset}(H_1')_S$. Then we have
\begin{align}
\mathcal B_y\sum_{S_1 \ni y}\e(\mathcal H )\ast\mathcal H_{S_1}\e(\mathcal H_1')A_x  &= \mathcal B_y\sum_{S_1 \ni y}\e(\mathcal H)\ast\mathcal H_{S_1}\qty[\e(\mathcal H_2') +  \ii\e(\mathcal H_1')\ast\mathcal H_2 \e(\mathcal H_2')]A_x  \notag\\
&= \ii\mathcal B_y\sum_{S_1 \ni y, S_2 \cap S_1 \neq \emptyset} \e(\mathcal H )\ast \mathcal H_{S_1}\e(\mathcal H_1')\ast \mathcal (\mathcal H_1')_{S_2} \e(\mathcal H_2')A_x
\end{align}
We can continue to iterate this process, choosing $H_m^\prime = H_{m+1} + H_{m+1}'$ where $H_{m+1} = \sum_{S_{m+1} \cap S_m \neq \emptyset}H_{S_{m+1}}$ and $H_{m+1}' = \sum_{S_{m+1} \cap S_m = \emptyset}H_{S_{m+1}}$
at each step, until $y \in S_{m+1}$, at which point the first term in the expansion no longer vanishes. Thus we obtain
\begin{align}
[B_y, A_x(t)] = \mathcal B_y\sum_{\substack{\vec S \in \Gamma(x \to y) \\|\vec S| = m}}\ii^m \e(\mathcal H)\ast \mathcal H_{S_1}\e(\mathcal H_1')\ast (\mathcal H_1')_{S_2} \dots  \e(\mathcal H_{m-1}') \ast (\mathcal H_{m-1}')_{S_m}\e(\mathcal H_m's_m)A_x
\label{eq:Schwinger_Karplus}
\end{align}
where the notation $\Gamma(x \to y)$ denotes the set of sequences $(S_1 \ni y, \dots, S_m \ni x)$ of subsets such that $S_k \cap S_{k+1} \neq \emptyset$. Because of the way $H_m$ is constructed from $H_{m-1}'$, no subset can appear in the sequence twice, so the paths involves are the \textit{irreducible paths} in $\Gamma(x \to y)$.

Taking norms on both sides,
\begin{align}
  \Vert [B_y, A_x(t)] \Vert &\leq \Vert B_y\Vert\Vert A_x \Vert\sum_{\substack{\vec S \in \Gamma(x \to y) \\|\vec S| = m}}2^m\Vert H_{S_1}\Vert \Vert H_{S_2}\Vert \dots \Vert H_{S_m}\Vert 1^{\ast m}(t)
  \notag = \Vert B_y\Vert\Vert A_x \Vert\sum_{\substack{\vec S \in \Gamma(x \to y) \\|\vec S| = m}}\frac{(2t)^m}{m!}\omega(\vec S)
  \label{eq:irreducible_path_bound}
\end{align}
where $1^{\ast m}$ denotes convolving $1$ with itself $m$ times.
In the first step, we repeatedly used the triangle inequality, the unitary invariance of the norm, and the submultiplicativity of the operator norm.
\end{proof}
Note that this works for open systems as well if we replace the superoperator $\mathcal H$ with the appropriate Lindbladian $\mathcal L$. This makes it convenient to define the following norm on super-operators:
\begin{defn}
\label{defn:inftyinftynorm}
Given a super-operator $\mathcal L$, define the $\infty-\infty$ norm to be the operator norm induced on $\mathcal L$ by the spectral norm on the space of (bounded) operators on the base Hilbert space:
\begin{align}
\Vert\mathcal L\Vert_{\infty-\infty} = \sup_{\Vert O \Vert = 1}\Vert \mathcal L O \Vert
\end{align}
\end{defn}
To generalize to open systems, we replace $\Vert \mathcal H \Vert_{\infty-\infty} \leq 2 \Vert H \Vert$ with $\Vert \mathcal L \Vert_{\infty-\infty}$ and observe that $\Vert \e^{\mathcal Lt}\Vert_{\infty-\infty} \leq 1$.

Now we will establish a basic Lieb-Robinson (LR) bound to illustrate the manipulations involved.
\begin{prop}
Suppose that $H$ is a $k-$local Hamiltonian with local terms $\Vert H_S \Vert \leq h$ for some $h$. Let $G$ be the graph on subsets $S$ obtained by adding an edge between $S_1, S_2$ whenever $H_{S_1}, H_{S_2}$ are terms in the Hamiltonian and $S_1 \cap S_2 \neq \emptyset$. Then we have
\begin{align}
\Vert [B_y, A_x(t)]\Vert \leq \Vert A_x \Vert \Vert B_y\Vert\qty(\frac{2\mathrm{e}hkt}{\mathsf d(x,y)})^{\mathsf d(x,y)}
\end{align}
where $\mathsf d(x,y) = \min_{S_1 \ni x, S_2 \ni y}\mathsf d(S_1, S_2)$ is the distance induced by the graph distance on $G$.
\end{prop}

\begin{proof}
We can bound $\omega(\vec S) \leq h^{|\vec S|}$ by assumption, so
\begin{align}
\Vert [B_y, A_x(t)] \Vert \leq \Vert A_x \Vert \Vert B_y \Vert \sum_{m=\mathsf d(x,y)}^{\infty}\sum_{\substack{\vec S \in \Gamma(x\to y) \\ |\vec S| = m}}\frac{(2t)^m}{m!}w(\vec S) \leq \Vert A_x \Vert \Vert B_y \Vert \sum_{m=\mathsf d(x,y)}^{\infty}\sum_{\substack{\vec S \in \Gamma(x\to y) \\ |\vec S| = m}}\frac{(2th)^m}{m!} 
\end{align}
The size of $\{\vec S \in \Gamma(x\to y):|\vec S| = m\}$ is bounded by $k^m$, where we loosen the bound to drop the non-self-intersecting requirement and use the fact that $H$ is $k-$local. Thus we have
\begin{align}
\sum_{m=\mathsf d(x,y)}^{\infty}\sum_{\substack{\vec S \in \Gamma(x\to y) \\ |\vec S| = m}}\frac{(2t)^m}{m!}w(\vec S) \leq  \sum_{m=\mathsf d(x,y)}^{\infty}\frac{(2hkt)^m}{m!} 
\end{align}
Suppose $\alpha > 1$ is arbitrary. Then $\alpha^{m-\mathsf d(x,y)} > 1$ for any $m \geq \mathsf d(x,y)$, so
\begin{align}
\sum_{m=\mathsf d(x,y)}^{\infty}\frac{(2hkt)^m}{m!}  \leq \alpha^{-\mathsf d(x,y)}\sum_{m=0}^{\infty}\frac{(2\alpha hk t)^m}{m!} = \alpha^{-\mathsf d(x,y)}\e^{2\alpha hk t}
\end{align}
Choosing $\alpha = \frac{\mathsf d(x,y)}{2hkt}$ (restricting the validity of the bound to outside the lightcone, but since the trivial bound is better inside the lightcone, this is inconsequential) we have
\begin{align}
\alpha^{-\mathsf d(x,y)}\e^{2hk\alpha t} = \qty(\frac{2\mathrm{e}hkt}{\mathsf d(x,y)})^{\mathsf d(x,y)}
\end{align}
\end{proof}
The above defines an emergent LR velocity because it depends on the dimensionless ratio $\frac{vt}{\mathsf d(x,y)}$, where $v = 2\mathrm{e}hk$, and so the commutator can only be appreciable once $vt \sim \mathsf d(x,y)$.

This version of the LR bound was for a strictly local evolution, but it can also
be established for quasilocal evolutions as well.
The key ingredient in extending the proof to the quasi-local case is assuming a
reproducing condition on the terms in the Hamiltonian. First, it is convenient to introduce some notation to talk about the norm of the operator restricted to large, connected subsets:
\begin{defn}
  \label{def:resnorm}
Given any $u,v \in \Lambda$ and a fixed local decomposition $A = \sum_{S
  \subseteq \Lambda}A_S$, define the restricted norm (although not properly a norm or semi-norm itself) of $A$ to
subsets containing $u,v$ as 
\begin{equation}
  \Vert A\Vert_{u,v} = \sum_{S \ni u,v}\Vert A_S\Vert
\end{equation}
\end{defn}
Although the $\kappa-$norm is a convenient way to package locality into two
numbers, in order to prove quasi-local LR bounds, we frequently want to place
restrictions on $\Vert H\Vert_{u,v}$, if $H$ is the system Hamiltonian. In
particular, for some local decomposition, we want the map $(u,v) \mapsto \Vert H\Vert_{u,v}$ to be a
\emph{reproducing} function:
\begin{defn}
Let $f: \Lambda \times \Lambda \to \mathbb R^+$ be a positive function. Then $f$ is called reproducing with parameter $K$ if
\begin{align}
\sum_{z \in \Lambda}f(x,z)f(z,y) \leq K f(x,y)
\end{align}
for any $x,y \in \Lambda$.
\end{defn}

In this way, we will see that a Hamiltonian $H$ satisfies a Lieb-Robinson bound with exponential tails if $\Vert H_{S}\Vert_{u,v} \leq hf(x,y)$ (where $h$ is a constant) for any $x,y\in \Lambda$ if $f$ is both reproducing and decays at most exponentially in $\mathsf d(x,y)$ (although this is sufficient but not necessary; LR bounds in systems with long-range interactions are known \cite{Chen:2019hou,Kuwahara:2019rlw,Tran:2021ogo}. To see how this definition is operationalized, consider the problem of summing over irreducible path weights. If $v_0, v_n \in \Lambda$ are arbitrary, then we define $\vec S \in \Gamma^{(m)}(v_0 \to v_n)$ as sequences of $m$ sets satisfying $S_k \cap S_{k+1} \neq \emptyset$. We can over-bound this sum by choosing one point $v_k$ in $S_k \cap S_{k+1}$ and summing over $v_k \in \Lambda, S_k, S_{k+1} \ni v_k$:
\begin{align}
  \sum_{\vec S \in \Gamma^{(m)}(v_0 \to v_n)} w(\vec S) &\leq \sum_{v_1, \dots, v_{n-1}}\sum_{S_1 \ni v_0,v_1}\dots \sum_{S_n \ni v_{n-1}, v_n}\Vert H_{S_1} \Vert \dots \Vert H_{S_n}\Vert \notag\\
                          &\leq 
\sum_{v_1, \dots, v_{n-1}}\Vert H\Vert_{v_0, v_1}\dots \Vert H\Vert_{v_{n-1}, v_n} 
\leq h^n\sum_{v_1, \dots, v_{n-1}}f(v_0, v_1)\dots f(v_{n-1}, v_n) \leq \frac{(hK)^n}{K}f(v_0,v_n)
  \label{eq:set-intersection}
\end{align}
Using the assumption $f(x,y) \leq \e^{-\mu \mathsf d(x,y)}$ for some $\mu$, we have
\begin{align}
\sum_{m=1}^{\infty}\frac{(2t)^m}{m!}\sum_{\vec S \in \Gamma^{(m)}(x \to y)}w(\vec S) \leq \frac{1}{K}\sum_{m=1}^{\infty}\frac{(2hKt)^m}{m!}f(x,y) \leq \frac{1}{K}\e^{-\mu\mathsf d(x,y)}(\e^{2hKt}-1)
\end{align}
We can then show that such a requirement is naturally satisfied by any $H$ with a bounded $\kappa-$norm:
\begin{lem}
\label{lem:kappa-norm-is-local}
Let $\Vert H \Vert_{\kappa} = h$. Considering the local decomposition which
realizes this norm, for any $\mu < \kappa$, $\beta > d+1$, and $c$ dependent on
$\mu, \beta$, we have
\begin{equation}
\Vert H\Vert_{u,v} \leq \frac{ch\e^{-\mu \mathsf d(u,v)}}{\mathsf d(u,v)^\beta}
\end{equation}
\end{lem}
\begin{proof}
Notice that 
\begin{align}
\sum_{S \ni u,v} \Vert H_S \Vert \leq \e^{-\kappa \mathsf d(u,v)}\sum_{S}\e^{\kappa \diam(S)}\Vert H_S \Vert \leq  h\e^{-\kappa \mathsf d(u,v)} \leq \frac{ch\e^{-\mu \mathsf d(u,v)}}{\mathsf d(u,v)^\beta}
\end{align}
\end{proof}

\begin{lem}[Def. 12 in \cite{hastings_mobilitygap}]
The function $f(u,v) = \e^{-\mu \mathsf d(u,v)}/\mathsf d(u,v)^\beta$ for $\mu >
0, \beta > d+1$ is reproducing. We will call the reproducing parameter $K$.
\end{lem}
The above is proven by bounding the sum with an integral, which is where the
finite dimensionality of the lattice comes into play. 
Putting the last few assertions together leads to the following bound, which we will employ heavily in the remainder of the paper:
\begin{lem}[Thm. 3.7, \cite{chen2023speed}]
\label{lem:standardLRbound}
Suppose that $O_S$ is an operator supported on $S$ which is disjoint from another $B \subset \Lambda$. Further suppose that $H$ is a Hamiltonian satisfying $\Vert H \Vert_{\kappa} \leq h$ for some constant $h$. Then for any $\mu < \kappa$ and $\beta > d+1$ there exists a constant $\CLR$ such that
\begin{align}
\Vert O_S(t)\Vert_{S, B} \leq \CLR\min(|\partial S|,|\partial B|)\Vert O_S \Vert\frac{\e^{-\mu \mathsf d(S, B) + vt}}{\mathsf d(S,B)^\beta}
\end{align}
where $c$ is a constant dependent on $\mu, \beta$ and $v = 2Kh$.
\end{lem}
The prefactor comes from summing over the possible origins of each path in
$\partial S$ and endpoints in $\partial B$, leading to a prefactor of $|\partial
S||\partial B|$ which can be further improved to $\min(|\partial B|, |\partial
S|)$ by recognizing that all paths of length $\mathsf d(S,B)$ are being counted,
so only origin or final points of the paths need to be summed
over \cite{chen2023speed}.

\section{Local Schrieffer-Wolff transformation}
The Schrieffer-Wolff transformation is an operator version of perturbation
theory (i.e. in the Heisenberg picture). To begin, we consider a system of the form $H = H_0 + V_{\epsilon}$ acting on $V_\ast$ sites, where $H_0$ satisfies the following non-resonance condition:
\begin{defn}[Non-resonant potentials]\label{defn:noresonance}
Let $\Lambda$ denote the vertex set of a graph.  We say that a geometrically $k-$local Hamiltonian
\begin{equation}
  H_0 = \sum_{I:\diam(I) \leq k}h_IZ_I
\end{equation}
where $Z_I$ denotes a product of Pauli-Z operators\footnote{This works for commuting Hamiltonians in general, not just those diagonal in the computational basis. However, non-trivial examples of such Hamiltonians, such as quantum error correcting codes, feature symmetries leading to degeneracies that violate the naive non-resonance condition. Since local operators cannot have matrix elements between states related by such symmetries, we expect these results to extend to such systems, but more analysis is required.}acting on
$I$, obeys an $(h, r_\ast,\Delta)$-non-resonant condition if $|h_I|\leq h$ and
for any $x$, $ S \subseteq B_{r_\ast}(x)$, we have
\begin{equation}
  \Delta_{S}(H_0) \geq h\Delta
\end{equation}
where $\Delta_S(H_0)$ represents the minimal gap of $H_0|_{S} = \sum_{I \subseteq S}h_IZ_I$.
\end{defn}

The perturbation $V_{\epsilon}$ we consider takes the
generic form
\begin{align}
V_{\epsilon} = \sum_{r \geq 1}^{\infty}\epsilon^rV_r \ \text{ with } \ \rng(V_r) \leq r \ \text{ and } \ \sum_r \Vert V_r\Vert_{\kappa = 0}  \leq h
\end{align}
which generalizes the perturbations considered in \cite{absenceofconduction}.

The goal is to look for an anti-Hermitian generator \begin{equation}
    T = \sum_{q =1}^\infty \epsilon^q T_q
\end{equation} such that $\mathrm{e}^T(H_0 + \epsilon V)\mathrm{e}^{-T}$ is diagonal in the eigenbasis of $H_{0}$ at all orders in $\epsilon$. 
The Schrieffer-Wolff transformation for many-body systems is developed in
\cite{bravyi2011schrieffer}, and focuses on
\emph{block-diagonalizing} $H$ between low- and high-energy subspaces
on an \emph{infinite} system to a \emph{finite} order in $\epsilon$ when $H_0$ has a \emph{single gap} in its spectrum. We
will adapt this approach to \emph{diagonalize} $H$ to \emph{all} orders in $\epsilon$ when
$H_0$ is \emph{non-degenerate} and acts on a \emph{finite} system. 
The main purpose of this section will be to prove the following:
\begin{thm}
\label{thm:quasilocal_generator}
Let $H = H_0 + V_{\epsilon}$ be restricted to the terms acting within a region of volume $V_\ast$, and let $H_0$ be
non-degenerate with minimal gap $\Delta$, $\rng(H_0) = k$, and $\Vert H \Vert_{\kappa = 0} = h$ (this is to say that $H = \sum_{I \subseteq \Lambda}H_I$ where $\sup_{I}\Vert H_I\Vert \leq h$, or in other words, $H_0$ has bounded local strength). 
Further suppose that $\rng(H_0) = k$ and $
\sum_{r = 1}^\infty\Vert V_r\Vert_{\kappa=0} \leq h$ with $\rng(V_r) \leq r$.
Then for
any $\tilde\epsilon \equiv A \epsilon \leq 1$,
\begin{equation}
\Vert T \Vert_{\kappa} < 1 \text{, where}\hspace{.2cm} A = 32\pi^2\qty(\frac{V_\ast h}{\Delta})^2 \ ,
\end{equation}
$\kappa = \frac{1}{5k + 1}\log(\tilde \epsilon^{-1})$ 
and $T$ is the
antihermitian generator of a unitary transformation $U_{\mathrm{SW}} = \e^{T}$
diagonalizing $H$.
\end{thm}
To prove the statement above, we will construct an absolutely convergent series
$T = \sum_{q=1}^{\epsilon} \epsilon^q T_q$. Stronger than this is that $T$ is a \textit{local} operator with locality controlled by $\epsilon$, which is captured by the $\kappa-$norm.

\subsection{Constructing the generator}
\begin{defn}
If $S \subseteq \Lambda$ and $X_S$ is an operator supported within $S$, define
\begin{equation}
L_S(X_S) = \sum_{E_i \neq E_j}\ketbra{i}{j} \frac{\bra{i}X_S\ket{j}}{E_i - E_j}
\end{equation}
where $\{\ket{i}\}$ are the eigenstates of of $H_S$, defined as the sum of the terms in $H_0$ acting non-trivially in $S$. This operator is well-defined in the case when $H_0$ is non-degenerate, which will follow from the non-resonance condition.

\end{defn}

\begin{prop} \label{prop:H0LScommutator}
For any $X_S$, we have $[H_0, L_S(X_S)] = O_S(X_S)$, where 
\begin{equation}
O_S(X_S) = \sum_{E_i \neq E_j}\ketbra{i}{j}\bra{i}X_S \ket{j} \ ,
\end{equation}
i.e. $O_S$ projects onto the off-diagonal matrix elements between eigenstates of $H_S$.
\end{prop}

\begin{proof}
Since the support of $L_S(X_S)$ does not intersect with the support of any term in $H_0$ outside of $S$, we have
\begin{equation}
[H_0, L_S(X_S)] = [H_S, L_S(X_S)] = \sum_{E_i \neq E_j}\qty[H_S\ketbra{i}{j}\frac{\bra{i}X_S \ket{j}}{E_i-E_j} - \ketbra{i}{j}H_S\frac{\bra{i}X_S \ket{j}}{E_i-E_j}] = O_S(X_S)
\end{equation}
This proves the claim.
\end{proof}

\begin{lem}
The operator norm of $L_S(X_S)$ is bounded by
\begin{align}
\Vert L_S(X_S) \Vert \leq \frac{\pi\Vert X_S \Vert}{\Delta_S\sqrt{3}} \label{eq:LSXSbound}
\end{align}
where $\Delta_S$ is the minimum spectral gap of $H_S$.
\end{lem}

\begin{proof}
Let $A$ be a matrix defined element-wise by \begin{equation}
    A_{ij} = \left\lbrace \begin{array}{ll} (E_i-E_j)^{-1} &\ i\ne j\\ 0 &\ i = j\end{array}\right..
\end{equation} Then we observe that
\begin{align}
L_S(X_S) = \sum_{E_i \neq E_j}\ketbra{i}{j} \frac{\bra{i}X_S\ket{j}}{E_i - E_j} = A \circ X_S
\end{align}
where $A \circ X_S$ denotes the Hadamard product, defined by $(A \circ B)_{ij} =
A_{ij}B_{ij}$. The norm \begin{equation}
    \Vert A \Vert_{\mathrm{S}} \equiv \sup_{\Vert B \Vert = 1}\Vert A \circ B
  \Vert
\end{equation} is commonly called the Schur multiplier norm. The Cauchy-Schwartz inequality implies  $\Vert A \Vert_{\mathrm{S}} \leq \max_i \Vert A
  \ket{i}\Vert_2$ 
  \cite{shurmultipliers}, where $\Vert \cdot \Vert_2$ is the 2-norm on vectors. Since
  $\Delta$ is the minimum gap between \emph{adjacent} energy levels, we can see that $|A_{ij}| \leq 1/(|i-j|\Delta)$, and so we have
  \begin{align}
 \Vert L_S(X_S) \Vert^2 \leq  2\Vert X_S \Vert^2\sum_{n=1}^{\infty}\frac{1}{n^2\Delta^2}  = \frac{\pi^2\Vert X_S \Vert^2}{3\Delta^2},
  \end{align}
  which implies \eqref{eq:LSXSbound}.
\end{proof}

\begin{defn}
Fix a tolerance $\delta>0$ and a local decomposition $X = \sum_S X_S$ achieving the $\kappa$-norm up to this tolerance. We define
\begin{align}
L(X) = \sum_S L_S(X_S).
\end{align}
We will leave the aforementioned dependence of $L$ on $\delta$ and the chosen local decomposition implicit in the notation.
\end{defn}
Strictly speaking, we always need to work with a local decomposition achieving a norm that is $\delta$-close to the $\kappa$-norm. However, this will not affect the structure of our proofs, so we will leave the dependence on $\delta$ implicit and take $\delta \to 0$ at the end of the computation.

\begin{prop}
If $\rng(H_0) = k$, then $\rng[L(X)] \leq 2k + \rng(X)$.
\label{prop:LSlocality}
\end{prop}

\begin{proof}
By assumption, $H_0 = \sum_{n}h_nZ_n$, where $Z_n, Z_{n'}$ are local products of Pauli-$Z$ operators. The energy levels of $H_0$ are given by $E_{\vec \sigma} = \vec \sigma \cdot
\vec h$, and so the level splittings are given by
\begin{align}
\Delta_{\sigma \sigma'} = E_{\vec \sigma} - E_{\vec \sigma'} = (\vec \sigma - \vec \sigma') \cdot \vec h = \prod_n \Delta\sigma_nh_n
\end{align}
where $\Delta \sigma_n$ measures the change in the eigenvalue of $Z_n$.
Therefore, given $n$ fixed, we can write down a projector
\begin{align}
\mathbb P^{\Delta \vec \sigma}_n X_S = \frac{1}{2\pi}\int_0^{2\pi} \dd t\e^{\ii\Delta\sigma_n t}\e^{-\ii Z_n t}X_S\e^{\ii Z_n t}
\label{eq:offdiagonalproj}
\end{align}
onto operators that change the eigenvalue of $Z_n$ by $\Delta\sigma_n$.
Let $A = \{k : \supp(Z_k) \cap S \neq \emptyset\}$. 
If $k \notin A$, then
we clearly have $\e^{-\ii Z_n t}X_S \e^{\ii Z_n t} = X_S$, so $\mathbb P^{\Delta\vec \sigma}_{n}X_S$ is only nonzero if $\Delta\sigma_k = 0$ for all $k \in A$.
Let $E =
\{\Delta\vec \sigma : \Delta\vec \sigma \cdot \vec h > 0\}$. Thus we may write the superoperator projecting onto the lower-right triangle as
\begin{align}
P^+X_S = \sum_{\substack{\vec \sigma \in E \\ \sigma_k = 0 \  \forall k \notin A}}\prod_{k \in S}\mathbb P^{\vec \sigma}_{k}X_S
\end{align}
Then, if $X_S$ is Hermitian, we can write $L_S(X_S)$ as
\begin{align}
&L_S(X_S)  = \int_{0}^{\infty} \dd t \e^{-H_0t}P^+X_S \e^{H_0t}  - \text{h.c.}
  \notag\\
  &= 
\sum_{\substack{\vec \sigma\in E: \\ \sigma_k = 0 \ \forall k \ \notin A}}
\frac{1}{(2\pi)^{|A|}}\int_0^{2\pi}\dd ^{|A|}\vec\tau\int_0^\infty \dd t \e^{\ii \vec \sigma \cdot \vec \tau}\exp(-\sum_{k \in A}(h_kt + \ii\tau_k\sigma_k)Z_k) X_S \exp(\sum_{k' \in A}(h_{k'}t + \ii \tau_{k'}\sigma_{k'})Z_{k'}) - \text{h.c.}
\label{eq:commutatorexp}
\end{align}
This expresses $L_S(X_S)$ as an operator supported within $A' = \bigcup_{k \in A}\supp(Z_k)$, and $\diam(A')
\leq \diam(S) + 2k$.  
\end{proof}

\begin{prop}
Suppose that $\Delta$ is the minimal gap of $H_0$ and $k = \rng(H_0)$
\begin{align}
\Vert L(X) \Vert_{\kappa} \leq \frac{\pi}{\Delta\sqrt{3}}\e^{2\kappa k}\Vert X \Vert_{\kappa}
\end{align}
\label{prop:kappaLS}
\end{prop}
\begin{proof}
Using Prop. \ref{prop:LSlocality},
\begin{align}
\Vert L(X) \Vert_{\kappa} &\leq \sum_{S} \e^{\kappa(\diam(S)+2k)}\Vert L_S(X_S)\Vert \notag \\
&=  \e^{2k\kappa}\sum_{S} \e^{\kappa \diam(S)}\Vert L_S(X_S)\Vert \notag \\
&\leq  \frac{\pi}{\Delta\sqrt{3}}\e^{2\kappa k}\sum_{S} \e^{\kappa \diam(S)}\Vert X_S\Vert =  \frac{\pi}{\Delta\sqrt{3}}\e^{2\kappa k}\Vert X \Vert_{\kappa}
\end{align}
where we chose the decomposition (or limit of decompositions) $X = \sum_S X_S$ which realize the $\kappa$-norm of $X$.
\end{proof}

Given a positive integer $N$ and an ordered tuple of non-negative integers $\vec{\lambda}$, we write $\vec{\lambda}\vdash [N]$ to express that $\vec{\lambda}$ partitions $N$, i.e., the components of $\vec{\lambda}$ sum to $N$.

\begin{lem}
  Define $T_q, V^{(q)}$ by $V^{(0)} = V_1$ and $T_q = L(V^{(q-1)})$, and
  \label{def:Vq}
\begin{align} \label{eq:lemB6eq}
V^{(q-1)} = \sum_{r=0}^{q}\sum_{\substack{\vec \lambda \vdash [q-r]\\ (r,|\lambda|) \neq (0,1)}}\frac{1}{|\lambda|!}\mathcal T_{\lambda}V_r
\end{align}
where for notational compactness, we denote $\mathcal T_\lambda \equiv \mathcal
T_{\lambda_1} \dots \mathcal T_{\lambda_l}$, $V_0 \equiv H_0$, and $\mathcal
T_{\lambda} = 1$ if $\lambda \vdash [0]$.
Then if the series converges, $\e^TH\e^{-T}$ is diagonal in the eigenbasis of $H_0$.
\end{lem}
\begin{proof}
Explicitly expanding $\e^{\mathcal T}$, where $\mathcal T \equiv \ad_{T}$, gives
\begin{align}
\e^{\mathcal T} &= \sum_{n}\frac{1}{n!}\qty(\sum_{q=1}^\infty \epsilon^q \mathcal T_q)^n
= \sum_{q}\epsilon^q \sum_{\vec \lambda \vdash [q]}\frac{1}{|\lambda|!}\mathcal T_{\lambda_1} \dots \mathcal T_{\lambda_{|\lambda|}}
\end{align}
Applying this to $H = \sum_{r=0}^\infty \epsilon^{r}V_r$, where we denote $H_0$ by $V_0$ for compactness, we have
\begin{align}
\e^{\mathcal T}(H_0 + V_{\epsilon}) = H_0+\sum_{q=1}\epsilon^q([T_q, H_0] + V^{(q-1)})
  \label{eq:localdecomp}
\end{align}
where $V^{(q-1)}$
is a function of $V_{r \leq q}, H_0$, and $T_{k < q}$. Since
\begin{align}
[T_q, H_0]+V^{(q-1)} = \sum_S \qty[V^{(q-1)}_S - O_S(V^{(q-1)}_S)] = \sum_S D_S(V^{(q-1)}_S) \ ,
\end{align}
which follows from our definition of $T_q$ and Proposition \ref{prop:H0LScommutator}, this choice diagonalizes $H$.
\end{proof}
Although we will not use this explicitly, we can bound the locality of $T$ at
each order:
\begin{prop}
  Suppose that $\rng(H_0) = k$. Then $\rng(T_q) \leq (1+5k)q-3k$. 
  \label{prop:Tlocality}
\end{prop}

\begin{proof}
Since $V^{(0)} = V_1$, we clearly have $\rng(V^{(0)}) = 1$. Assume
$\rng(V^{(q'-1)}) \leq (1+5k)q'-3k$ for all $q' < q$.  Using \eqref{eq:lemB6eq}, we see that
\begin{align}
 \rng(V^{(q-1)}) &\leq
  \max_{\substack{r+p = q \\ \lambda \vdash [q-r] \\ |\lambda| = l, (r, l) \neq (0,1)}}\rng(\mathcal T_{\lambda}V_{r}) \leq \max_{(r,l) \neq (0,1)}(1+5k)(q-r)-3kl + r + \delta_{r0}k
\end{align}
This is clearly maximized when $r = 0, l = 2$, in which case we find $\rng(V^{(q-1)})
\leq (1+5k)q-5k$, as desired. Thus we have $\rng(T_q) \leq 2k+ \rng(V^{(q-1)})
\leq (5k+1)q-3k$.\footnote{With further analysis, this could be improved to $(1+3k)q-k$ by observing that $[L(X_{S_1}), L(X_{S_2})] = 0$ if $\mathsf d(S_1, S_2) \geq k$. From \eqref{eq:commutatorexp}, the operators grow into the interstitial region by stabilizers, which commute with each other, so the support of $L(X_{S_1})$ must overlap $S_2$ or vice-versa. We will also not treat the case $\rng(V_{\epsilon}) = 0$ separately, which would result in a mild improvement to $3kq$.}
\end{proof}

\subsection{Locality bounds on $T$}
We have constructed $T = \sum_{q = 1}^\infty \epsilon^q T_q$ just in case the
series converges. 
We will therefore choose the largest $\kappa$ such that $\Vert T \Vert_{\kappa}$ is finite,
and apply traditional LR bounds to $U_{\text{SW}} = \mathrm{e}^{T}$ to show that a local operator transformed under $U_{\text{SW}}$ is again quasilocal.

\begin{proof}[Proof of Theorem~\ref{thm:quasilocal_generator}]
We will bound $\Vert T \Vert_{\kappa}$ by repeated applying Prop.~\ref{prop:kappaLS}, following the strategy of \cite{bravyi2011schrieffer}.
Let $k = \rng(H_0)$. Assume that $V_{\epsilon} = \sum_{r}\epsilon^r V_r$, where
$\diam(V_r) = r$, and $\sum_{r=1}^\infty\Vert V_r \Vert_{0} \leq h$.
First, given operators $A$,$B$, the triangle inquality clearly bounds $\Vert [A, B]\Vert_{\kappa} \leq \Vert A B \Vert_{\kappa} + \Vert B A \Vert_{\kappa}$. Since we assume in the theorem hypothesis that $H$ acts on a system of volume $V_\ast$, we can then apply the almost-submultiplicativity of the $\kappa$-norm (Lemma~\ref{lem:submultiplicative}) to find $\Vert A B \Vert_{\kappa} \leq 2V_{\ast} \Vert A \Vert_{\kappa}\Vert B \Vert_{\kappa}$. Applying this iteratively, we have
\begin{align}
\label{eq:recursivebound}
\Vert V^{(q-1)} \Vert_{\kappa} &\leq 
\Vert H_0 \Vert_{\kappa}\sum_{\substack{\vec \lambda \vdash [q] \\ |\lambda|  = l \neq 1}}(4V_\ast)^l\frac{1}{l!}\Vert T_{\lambda_1}\Vert_{\kappa} \dots \Vert T_{\lambda_l}\Vert_{\kappa}+\sum_{r=1}^{q}\Vert V_r \Vert_{\kappa} \sum_{\substack{\vec \lambda \vdash [q-r] \\ |\lambda|  = l \neq 1}}\qty(4V_\ast)^l\frac{1}{l!}\Vert T_{\lambda_1}\Vert_{\kappa} \dots \Vert T_{\lambda_{l}}\Vert_{\kappa} \notag \\
& \leq 
he^{k\kappa}\sum_{\substack{\vec \lambda \vdash [q] \\ |\lambda|  = l \neq 1}}\qty(\frac{4\pi V_\ast \e^{2\kappa}}{\Delta \sqrt{3}})^l\frac{1}{l!}\Vert V^{(\lambda_1-1)}\Vert_{\kappa} \dots \Vert V^{(\lambda_l-1)}\Vert_{\kappa} \notag
\\&+\sum_{r=1}^{q}\e^{\kappa r}\Vert V_r \Vert_{0} \sum_{\substack{\vec \lambda \vdash [q-r] \\ |\lambda|  = l \neq 1}}\qty(\frac{4\pi V_\ast \e^{2\kappa}}{\Delta \sqrt{3}})^l\frac{1}{l!}\Vert V^{(\lambda_1-1)}\Vert_{\kappa} \dots \Vert V^{(\lambda_{l}-1)}\Vert_{\kappa}
\end{align}
where in the second line we bounded the $\kappa-$norm of $T_{\lambda_i} = L(V^{(\lambda_i - 1)})$ using Prop.~\ref{prop:kappaLS}, and we applied the $\kappa-$norm estimate for local operators in Prop.~\ref{prop:kappa_norm_scaling} to $V_r$ and $H_0$.

Let $h \e^{\kappa k}\mu_q$ saturate the inequality above, i.e.
\begin{align}
\mu_q = 
\sum_{\substack{\vec \lambda \vdash q \\ |\lambda|  = l \neq 1}}\frac{b^l}{l!}\mu_{\lambda_1} \dots \mu_{\lambda_l} + \frac{\e^{-\kappa k}}{h}\sum_{r= 1}^q\e^{\kappa r}\Vert V_r \Vert\sum_{\substack{\vec \lambda \vdash [q-r] \\ |\lambda| = l}}\frac{b^l}{l!}\mu_{\lambda_1}\dots \mu_{\lambda_l}
\end{align}
for $q \geq 1$, where we defined \begin{equation}
    b = \frac{4\pi}{\sqrt{3}}V_\ast\e^{3\kappa
  k}\frac{h}{\Delta}. \label{eq:define_b}
\end{equation} We then have $\Vert V^{(q-1)}\Vert \leq h\e^{\kappa k}\mu_q$, and thus
\begin{align}
\Vert T \Vert_{\kappa}\leq \frac{\pi \e^{\kappa k} h}{\sqrt{3}\Delta}\mu(\epsilon) \leq b\mu(\epsilon)  \equiv b\sum_{q=1}^{\infty}\epsilon^q \mu_q
\label{eq:generating_function}
\end{align}
Then we can derive an explicit equation for $\mu(\epsilon)$:
\begin{align}
\mu(\epsilon) &= \sum_{q=1} \epsilon^q\qty[\sum_{\vec \lambda \vdash [q]}\frac{b^l}{l!}\mu_{\lambda_1}\dots \mu_{\lambda_l} - b\mu_q + \frac{\e^{-k\kappa}}{h}\sum_{r = 1}^q\sum_{\vec \lambda \vdash [q-r]}\frac{b^l}{l!}\mu_{\lambda_1} \dots \mu_{\lambda_l}] \notag \\
&= \qty(1+\frac{\e^{-k\kappa}}{ h}\Vert V_{\epsilon} \Vert_{\kappa})\exp[b\mu(\epsilon)] - b\mu(\epsilon) -1
  \label{eq:equation_for_mu}
\end{align}
This abuses notation by writing $\Vert V_{\epsilon} \Vert_{\kappa} = \sum_{r = 1}\e^{\kappa r}\epsilon \Vert V_r \Vert_0$, which by assumption satisfies $\Vert V_\epsilon \Vert_{\kappa} \leq
\e^{\kappa}\epsilon h$. Put $\overline \epsilon = \e^{(1-k)\kappa}\epsilon$.
Consider the related function 
\begin{align}
\widetilde\mu(\overline \epsilon) &=  (1+\overline\epsilon)\exp[b\widetilde\mu(\overline \epsilon)] - b\widetilde\mu(\overline \epsilon) -1
\label{eq:mumax}
\end{align}
Let $f(\overline\epsilon) \equiv \e^{-k\kappa}\Vert V_{\epsilon}\Vert_{\kappa}/h$.  Comparing \eqref{eq:equation_for_mu} and \eqref{eq:mumax} we observe that \begin{equation}
    \mu(\overline \epsilon) = \widetilde\mu(f(\overline \epsilon)).
\end{equation} 
By
hypothesis, $f(\overline\epsilon) \leq \overline\epsilon$, so $f$ is a
contraction mapping. From e.g. \eqref{eq:equation_for_mu} we observe that $\mu$ and $\widetilde{\mu}$ have a Taylor series with positive coefficients since each term in the power series expansion of $\lVert V_\epsilon\rVert_\kappa$ is non-negative.  Hence $\mu$ is a monotonically increasing function, and
and $\mu(\overline \epsilon) \leq \widetilde \mu(\overline \epsilon)$ holds on the domain where the Taylor expansion converges.

Now we can find the radius of
convergence of $\widetilde\mu$. By the implicit function theorem, $\widetilde\mu$ is given
by a smooth function in a neighborhood of any point where $\widetilde \mu$ is
defined. This implies that the only non-analyticity that can occur is a
singularity in the first derivative. Differentiating implicitly, we have
\begin{align}
  \pdv{\widetilde\mu}{\overline \epsilon} &= \frac{\e^{b\widetilde\mu}}{(1+b)-b\e^{b\widetilde\mu}(1+\overline \epsilon)}
\end{align}
this shows that the derivative is bounded unless the denominator vanishes.
Since $\widetilde \mu$ is monotonically increasing, it achieves a maximum at the point of non-analyticity. Setting the denominator to zero and plugging this into \eqref{eq:mumax}, we have
\begin{align}
\widetilde \mu_{\rm{max}} = \frac{1}{b(b+1)} \leq \frac{1}{b^2} 
\end{align}
Then solving for $\overline \epsilon_{\rm{crit}}$ where this maximum is achieved, we have
\begin{align}
  \overline \epsilon_{\rm{crit}} = \frac{b+1}{b}\exp(-\frac{1}{1+b})-1
\end{align}
We can get a simpler bound using the following observation: 
\begin{align}
 \overline \epsilon_{\rm{crit}} = \frac{b+1}{b}\exp(-\frac{1}{1+b})-1 > \frac{1}{2(2\sqrt{\e}-1)}\qty(\frac{1}{b(b+1)}) > \frac{1}{3}\mu_{\rm{max}} > \frac{1}{6b^2}
\end{align}
where we have used that $b > 1$.  After all, $V_\ast > 1$ for any non-trivial non-resonance condition and $h/\Delta \geq \frac{1}{2}$ (which follows from the spectral gap from flipping a single bit), so $b>1$ follows immediately from \eqref{eq:define_b}.

Thus, $\mu(\epsilon) < \frac{1}{b} $ for all $\overline \epsilon \leq
\frac{1}{6b^2}$. Define $\tilde \epsilon$ via
\begin{align}
  \tilde \epsilon = 32\pi^2\qty(\frac{V_\ast h }{\Delta}) ^2\epsilon \label{eq:epsilontilde}
\end{align}
Then plugging in \eqref{eq:define_b}, 
\begin{align}
6b^2 \overline \epsilon = 32\pi^2\qty(\frac{V_\ast h \e^{3k\kappa}}{\Delta})^2 \e^{(1-k)\kappa} = \tilde \epsilon \e^{(5k+1)\kappa}
\end{align}
Let\footnote{Note that this matches the locality estimate from Prop.~\ref{prop:Tlocality}, and similarly, a more careful analysis could be used to tighten the locality to $\sim (3k+1)q$ here as well. In addition, the special case $\rng(V_{\epsilon}) = 0$ would result in a slight improvement to $\kappa \sim \frac{1}{3k}\log(\tilde \epsilon^{-1})$.}
\begin{equation}\kappa = \frac{1}{5k+1}\log(\tilde \epsilon^{-1})\end{equation} 
so that from \eqref{eq:generating_function} and $\mu(\widetilde\epsilon) < \frac{1}{b}$, we have
\begin{align}
\Vert T \Vert_{\kappa} \leq  b\mu_{\rm{max}} < 1
\end{align}
for all $\widetilde \epsilon \leq 1$.
\end{proof}

\subsection{Locality of transformed operators}
The ultimate goal is to perform SW transformations on finite regions, and then
stitch them together. In order to model how the locality of operators changes at the
boundary of the stitch, we have the following lemma;
\begin{lem}
\label{lem:SW_transformed_LRbound}
Suppose that $O_S$ is an operator supported on a subset $S$, and let
$U_{\mathrm{SW}}$ be the transformation constructed in the previous section.
Denote $O_S^{\mathrm{SW}} \equiv U_{\mathrm{SW}}^\dagger O_S U_{\mathrm{SW}}$.
Then for any $\alpha < \frac{1}{5k + 1}$ and $\beta > d+1$ there exists a constant $c$ dependent on $\alpha$ such that for any $x$, 
\begin{align}
\Vert  O_S^{\mathrm{SW}}\Vert_{S,x} \leq c|\partial S|\Vert O_S \Vert \frac{\tilde \epsilon^{\alpha  \mathsf{d}(S,x)}}{\max[1,\mathsf{d}(S, x)]^\beta}
\label{eq:SW_transformed_LRbound}
\end{align}
\end{lem}
\begin{proof}
This follows by applying the standard LR bound (Lem.~\ref{lem:standardLRbound}) with $\Vert T \Vert_{\kappa} \leq 1$, $\kappa = \frac{1}{5k + 1}\log(\tilde \epsilon^{-1})$, and a fictitious time $t = 1$. 
\end{proof}

Next, since $T$ is quasilocal, $H^{\text{SW}}$ will be quasilocal as well. We
must specifically bound the decomposition of $H^\SW$ into a sum of commuting terms. We will
see that this follows naturally from our bound on $\mu(\epsilon)$:
\begin{lem}
  Consider the decomposition of $H^{\text{SW}} = \sum_{S\subseteq \Lambda} Z_S$ into local, diagonal terms given in
  \eqref{eq:localdecomp}. This local decomposition realizes a $\kappa-$norm
  bounded by 
\begin{align}
\Vert H^{\text{SW}}-H_0\Vert_{\kappa}^{\{S\}} \leq h
\end{align}
where again $\kappa = \log(\tilde \epsilon^{-1})$. We wrote $\Vert \cdot
\Vert^{\{S\}}_\kappa$ to emphasize that this differs from the $\kappa-$norm
because the local decomposition is fixed.
\end{lem}

\begin{proof}
We have
\begin{align}
  H^{\text{SW}}-H_0 = \sum_{q = 1}^{q}\epsilon^q\sum_S \mathcal D_S(V^{(q-1)}_S) 
\end{align}
The superoperator $\mathcal D_S$ is a dilation \cite{bhatia2000pinching}, and by
Eq.~\eqref{eq:offdiagonalproj}, it increases the locality by at most $2k$, so we have
\begin{align}
\sum_S \e^{\kappa \diam(S)} \Vert\mathcal D_S(V^{(q-1)}_S)  \Vert \leq \sum_S \e^{\kappa [2k +\diam(S)]}\Vert V^{(q-1)}_S \Vert = \e^{2k\kappa}\Vert V^{(q-1)} \Vert_{\kappa}
\end{align}
and therefore 
\begin{equation}
\Vert H^{\text{SW}}-H_0 \Vert^{\{S\}}_{\kappa} \leq h\e^{2k\kappa}\mu_{\rm{max}} < h
\end{equation}
\end{proof}
This immediately implies the following bound on the decomposition of $H^\SW$ over distances:
\begin{cor}
\label{cor:HSW_locality}
There exists a constant $c$ (from Lem.~\ref{lem:kappa-norm-is-local}) such that
  \begin{align}
  \Vert H^{\SW} \Vert_{u,v} \leq 2hck^\beta\tilde \epsilon^{-\alpha k}\frac{\tilde \epsilon^{\alpha \mathsf
      d(u,v)}}{\max[1,\mathsf d(u,v)]^\beta}
   \end{align}
\end{cor}

\begin{proof}
  \begin{align}
\Vert H^{\SW} \Vert_{u,v} \leq \Vert H^{\SW} - H_0 \Vert_{u,v} + \Vert H_0
\Vert_{u,v} \leq h\delta_{\mathsf d(u,v) \leq k} + ch\frac{\tilde
  \epsilon^{-\alpha \mathsf d(u,v)}}{\mathsf d(u,v)^{\beta}} \leq \frac{2ch k^\beta \tilde
  \epsilon^{-\alpha (\mathsf d(u,v) - k)}}{\mathsf d(u,v)^\beta}
    \end{align}
\end{proof}

\section{Lieb-Robinson bounds for non-resonant models}
Now that we have established the construction of the local Schrieffer-Wolff transformation when $H_0$ is non-resonant up to a length $r_\ast$, we can show how to formulate Lieb-Robinson bounds with a velocity that is suppressed up to an order proportional to $r_\ast$ in $\epsilon$. Our goal is to develop an equivalence-class expansion similar to the one introduced in Sec.~\ref{sec:lrbound_duhamel}.

\subsection{Non-resonance at fixed scale $r_*$}
The strategy underlying our new Lieb-Robinson bounds leverages the Duhamel identity to construct irreducible paths with ``jumps" which are at least as large as $r_*$, where we can then use the Schrieffer-Wolff transformations locally to show how non-resonance forbids fast dynamics.  The set of all irreducible paths of interest is $\Gamma(S_0 \to S_1)$, and is constructed explicitly in the following definition.  See  Fig.~\ref{fig:examplepaths} for an illustration. 

\begin{defn}
\label{def:generate_paths}
Considering a local decomposition $H = \sum_{S \subseteq \Lambda}H_S$ and subsets $S_0, S_1$ with $S_0\cap S_1=\emptyset$, 
let
$\Gamma(S_0 \to S_1)$ be the set of paths constructed as follows. 

Put $n_\ast = \lfloor (\mathsf d(S_0, S_1)-1)/(r_\ast+1)\rfloor$. Let $\Gamma = (\Gamma_1, \Gamma_2, \dots, \Gamma_l)$ such that $l = n_\ast+1$ or $i=l$ is the smallest $i$ with $\min_{y \in S_0}\max_{x \in \Gamma_i}\mathsf d(x,y) \geq n_\ast r_\ast$.  Now define $\Sint{-1} = \emptyset$ and $\Sint{0} = S_0$, and inductively, we will require $\Gamma_i \subseteq \Lambda$ to satisfy the conditions $\Gamma_i \in \Sint{i-2}(\Gamma)^{\mathrm{c}}$, $\Gamma_{i} \cap \Sint{i-1}(\Gamma) \neq \emptyset$, and $\Gamma_{i} \cap \Sint{i-1}(\Gamma)^{\mathrm{c}} \neq \emptyset$. 

Given $\Gamma$ with $\Gamma_i$ satisfying the above, we define $\Sint{i}$ as follows: if $\min_{x \in \partial \Sint{i-1}(\Gamma)}\max_{y \in \Gamma_{i}\cap \Sint{i-1}(\Gamma)^{\mathrm{c}}} > r_\ast$ then we say $\Gamma_i$ is ``big", otherwise $\Gamma_i$ is ``small". If $\Gamma_i$ is big, then put $\Sint{i}(\Gamma) = \Sint{i-1} \cup \Gamma_i$, otherwise pick the least (with respect to some total ordering) $x \in \partial \Sint{i-1}(\Gamma)^{\mathrm{c}}$ such that $\Gamma_{i+1}\cap \Sint{i-1}(\Gamma)^{\mathrm{c}} \subseteq B_{r_\ast}(x)$ and put $\Sint{i}(\Gamma) = \Sint{i-1}(\Gamma) \cup B_{r_\ast}(x)$. 

We will also denote $\Gamma(S_0 \to S_1) = \Gamma^{(n_\ast)}(S_0)$, because $\Gamma(S_0 \to S_1)$ only depends on $S_1$ through $n_\ast$.
\end{defn}

\begin{figure}[t]
    \def\svgwidth{0.95\linewidth} 
    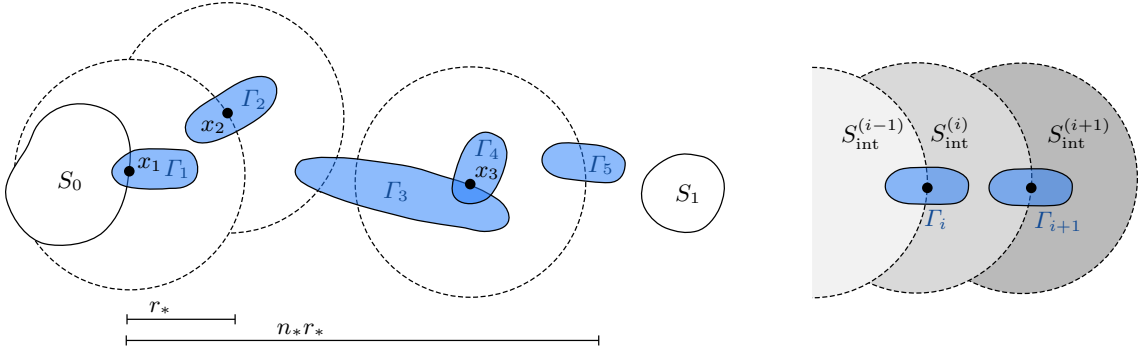
  \caption{\textbf{Left:} A coarse-grained picture of an example path $\Gamma = (\Gamma_1, \dots, \Gamma_5)$. \textbf{Right:} Depiction of the relationship among $\Sint{i-1}$, $\Sint{i}$, and $\Gamma_i$. The darkening shades of gray show the addition of subsequent regions to $\mathrm S_{\mathrm{int}}$ (or equivalently, removal from $\mathrm S_{\mathrm{ext}}$).}
  \label{fig:examplepaths}
\end{figure}

Then as in \eqref{eq:irreducible_path_bound}, we can expand the commutator into an irreducible path bound. We generalize the irreducible path bound by way of the following definitions.

\begin{defn}
Given the Hilbert space of bounded operators $\mathcal B = \mathcal B_{S} \otimes \mathcal B_{S^{\mathrm{c}}}$ where $S \subseteq \Lambda$, define $\overlinebold P_{S}$ to be the orthogonal projector onto the space $1_{S} \otimes \mathcal B_{S^{\mathrm{c}}}$, and let $\mathbb P_S$ be its complement projector.
\end{defn}

\begin{defn}
Let
  \begin{align}
    C_{i}^{(\Gamma)}(t) = \left\Vert \mathcal H_{\Gamma_{i}}\e\qty(\overline{\mathbb P}_{S_{\mathrm{int}}^{(i-1)}(\Gamma)}\mathcal H)(t)\overline{\mathbb P}_{S_{\mathrm{int}}^{(i)}(\Gamma)\backslash\Gamma_{i+1}}\right\Vert_{\infty-\infty}
   \end{align}
Furthermore, let \begin{equation}
    \widetilde C_{i}^{(\Gamma)}(s) := \int_0^\infty \dd s \; 
   \e^{-st}C_{i}^{(\Gamma)}(t)
\end{equation} be the Laplace transform of $C_{i}^{(\Gamma)}$.\footnote{Since we can
     easily bound $C_i^{(\Gamma)}(t)$ with an exponential function (e.g. with
     standard LR bounds) the Laplace transform and its inverse are both
     well-defined. However, the Laplace transform is mostly a convenient
     accounting method, and one may verify that the manipulations can also be
     carried out directly in real time.}
\end{defn}

\begin{lem}
\label{lem:jump_paths}
  If $\Gamma(S_0 \to S_1)$ is the set of paths constructed in Def.~\ref{def:generate_paths}, then we have
\begin{align}
\Vert [B_{S_0}, A_{S_1}(t)]\Vert \leq 2\Vert A_{S_1}\Vert \Vert B_{S_0}\Vert \sum_{\Gamma \in \Gamma(S_0 \to S_1)} [1 \ast C^{(\Gamma)}_1 \ast \dots \ast C^{(\Gamma)}_{|\Gamma|}](t)
\end{align}
In terms of the Laplace transform,
\begin{align}
\mathcal L\{\left\Vert\mathcal B_{S_0}\e(\mathcal H)(t)A_{S_1}\right\Vert\} \leq 2\Vert A_{S_1}\Vert \Vert B_{S_0}\Vert\frac{1}{s}\sum_{\Gamma \in \Gamma(S_0 \to S_1)} \prod_{i =1}^{|\Gamma|}\widetilde C_{i}^{(\Gamma)}(s)
  \label{eq:comm}
\end{align}
\end{lem}
Note that this lemma reduces to the earlier Lieb-Robinson bound, Proposition \ref{prop:irreducible_paths}, based on irreducible paths if we use the trivial $\widetilde C_{i}^{(\Gamma)}(s)
\leq \frac{2}{s}\Vert H_{\Gamma_i} \Vert$.  The proof of this result is similar to the proof of Proposition \ref{prop:irreducible_paths}, while making use of the detailed properties in Definition \ref{def:generate_paths}; we present it in Section \ref{sec:proof_jumppaths}.

Our bound will then follow from proving that $C_{i}^{(\Gamma)}(s) \lesssim s^{-1} \epsilon^{r_\ast}$ (schematically).
If $\Gamma_i$ is ``big," then this is a direct consequence of $\lVert V_r\rVert_\kappa \lesssim \epsilon^r$ with $r\ge r_*$. If $\Gamma_i$ is ``small", then we can construct the Schrieffer-Wolff transformation around the corresponding $x_i$, and show that the growth to the boundary of $B_{r_\ast}(x_i)$ is similarly suppressed. 

The big advantage of our approach that makes it readily applicable to disordered potentials and incommensurate lattices is robustness to failures of the non-resonance condition. Essentially, within some choices of $x_i$ we are not able to construct an SW transformation, and we are forced to use the trivial $C_{i}^{(\Gamma)}(s) \lesssim \Vert H_{\Gamma_i} \Vert / s$. 
With the strategy for constructing irreducible paths explained, we are now in a position state precisely the weaker notion of partial non-resonance allowing for these failures. Loosely, the definition below ensures that the non-resonance condition fails with probability $\le \zeta$ in any region along a path in the lattice.

\begin{defn}\label{defn:partialnr}
We will say that a potential $H$ satisfies an $(h, r_\ast, \Delta, n_\ast, \zeta)$ partial non-resonance condition at $x \in \Lambda$ if for any $\Gamma \in \Gamma^{(n_\ast)}(x)$, there is no subsequence $i_1, \dots, i_k$ with $k > (1-\zeta)n_\ast$ for which $\Gamma_{i_j}$ is small and the non-resonance condition fails in $B_{r_\ast}(x_{i_j})$ (with $x_{i_j}$ chosen as in Def.~\ref{def:generate_paths}). 
\label{def:partialnonresonance}
\end{defn}

Here we remark that this definition of partial non-resonance is key to our proof strategy, because we are able to explicitly forbid the existence of a percolating path of non-resonant regions.

\subsection{Main theorem}
In this section, we will formally combine the SW transform with the Lieb-Robinson formalism to establish the following main theorem:
\begin{thm}
  \label{thm:main_thm}
Let $H = H_0 + V_{\epsilon}$, where $H_0$ is an onsite potential with local
strength $h$ and $\rng(H_0) = k$. Consider two subsets $S_0,S_1\subset\Lambda$. Assume that $H_0$ satisfies a $(h, r_\ast, \Delta, n_\ast, \zeta)-$non-resonance
condition at every $x \in \partial S_0$ if $|\partial S_0| \leq |\partial S_1|$ and at every $y \in \partial S_1$ otherwise on a graph in $d$ dimensions.  Let $r = \min(r_\ast +1, \mathsf d(S_0,
S_1))$. Suppose that $A_{S_1}$ is an operator
supported on $S_1$ and $B_{S_0}$ is an operator supported on 
$S_0$, and let $n_\ast = \lfloor (\mathsf d(S_0,S_1)-1)/(r_\ast + 1)\rfloor$.
For any $\alpha < 1$, there exist constants $C,C'$ such that if $\tilde \epsilon < \e^{-1}$, we have
\begin{align}
\Vert [B_{S_1}, A_{S_0}(t)] \Vert \leq
    C\min(|\partial S_0|, |\partial S_1|)\Vert A_{S_0} \Vert \Vert B_{S_1} \Vert\qty[\exp(C'h(\e\tilde \epsilon)^{\alpha[\zeta r/8-k]}t)-1]\exp(-\frac{\alpha \zeta}{8}\mathsf d(S_0,S_1))
\end{align}
for any $r > 8k/\zeta$.
\end{thm}

\begin{proof}
As explained above, we first bound $C_{\Gamma_i}$, and then bound the sum over all such paths using a set-intersection counting argument similar to the one in Sec.~\ref{sec:lrbound_duhamel}. For the first task, if $\diam(\Gamma_i) > r_\ast$, then the trivial bound gives $\widetilde
C^{(\Gamma)}_i(s) \leq s^{-1}\Vert
H_{\Gamma_i}\Vert \leq s^{-1}\epsilon^{\diam(\Gamma_i)}$, so this case is easy. The
hard case will be $\diam(\Gamma_i) \leq r_\ast$, and we will defer this to a
technical lemma proved in Section \ref{sec:proof_C4}:

\begin{lem}
\label{lem:small_terms}
Let $\Gamma \in \Gamma(S_0 \to S_1)$. Then if $\Gamma_i$ is ``small", we have
\begin{align}
\widetilde C_{i}^{\Gamma} \leq  \frac{2K\Vert H_{\Gamma_i}\Vert}{s}\qty(1+\frac{\widetilde h}{s})^2\sum_{u \in \Gamma_{i-1}}\sum_{v \in \Gamma_i}\frac{\tilde\epsilon^{\alpha\mathsf d(u,v)/2}}{\mathsf d(u,v)^\beta}
\end{align}
where \begin{equation}
    \widetilde h = 4hK(1+k)^{2\beta}(r_\ast+1)^{\beta}\tilde \epsilon^{-\alpha k}. \label{eq:widetilde_h}
\end{equation}
\end{lem}

The summation over $u \in \Gamma_{i-1}, v \in \Gamma_{i}$ is an accounting tool similar to the ideas in Sec.~\ref{sec:lrbound_duhamel} to bound set intersections.
The central idea of the lemma is that a ``small" term $\Gamma_i$  can be captured within a ball of radius $r_\ast$, within which we can construct an SW transformation, which will suppress the growth of operators to the boundary. 

Given a path $\Gamma$, each small $\Gamma_i$ corresponds to an $x_i$ as chosen in Def.~\ref{def:generate_paths}. By choosing $x_i \in \partial \Sint{i-1} \cap \Gamma_{i}$ when $\Gamma_i$ is big, we can sort the paths by the points $x_1, \dots, x_{n}$ for $n\leq n_\ast+1$.  Let \begin{equation}f(u,v) = \frac{\tilde \epsilon^{\alpha \mathsf d(u,v)/2}}{\mathsf d(u,v)^{\beta}}\end{equation}
be reproducing with parameter $K$. Let $\isbig(\Gamma, i)$ denote whether $\Gamma_i$ is big, and let
$\operatorname{NR}(x_i)$ indicate whether the non-resonance condition is obeyed
at $x_i$. Applying Lem.~\ref{lem:small_terms}, we have

\begin{align}
  \sum_{\Gamma \in \Gamma(S_0\to S_1)}\prod_{i=1}^{|\Gamma|} C_i^{(\Gamma)} &\leq
  \sum_{n=1}^{n_\ast+1}\qty(\frac{2}{s})^n\sum_{x_1, \dots, x_{n+1} \in \Lambda}\sum_{\substack{u_1, \dots, u_n \in \Lambda \\ v_1, \dots, v_n \in \Lambda}}\prod_{i=1}^{n}\qty(K\qty[1+\frac{\widetilde h}{s}]^2f(u_i, v_i)\delta_{I_1}  + \delta_{I_2} + \delta_{I_3}) \notag\\
  &\times\sum_{\Gamma_i \ni x_i, u_i, v_{i-1}}\Vert H_{\Gamma_i}\Vert 
  \label{eq:sumoverpaths}
\end{align}
where $\delta_{I_{1,2,3}}$ denotes an indicator function on the sets $I_1, I_2, I_3$,
defined as follows:
\begin{subequations}
\begin{equation}
  I_1 = \{(x_i, u_i, v_i) | \neg \isbig(\Gamma,i), i < n+1, \mathrm{NR}(x_i), \mathsf d(x_i, x_{i+1}) = r_\ast+1\}\\
\end{equation}
\begin{equation}
  I_2 = \{(x_i, u_i, v_i) | \isbig(\Gamma, i)\text { or } i = n+1 , u_i = x_i, v_i = x_{i+1}\} \\
\end{equation}
\begin{equation}
  I_3 = \{(x_i, u_i, v_i) | \neg\isbig(\Gamma, i),i < n+1, \neg\mathrm{NR}(x_i), \mathsf d(x_i, x_{i+1}) = r_\ast + 1, u_i = x_i, v_i = x_{i+1}\}
\end{equation}
\end{subequations}
These conditions represent the three possible ways the path can grow: ($I_1$) via a non-resonant bubble, ($I_2$) via a ``non-local" (diameter $\ge r_*)$ term in the potential which occurs at a high order in $\epsilon$ or one that occurs at the end of the sequence, or ($I_3$) via a resonant bubble: see Fig.~\ref{fig:placeholder_2}.
\begin{figure}[t]
    \centering
    \def\svgwidth{0.7\linewidth} 
    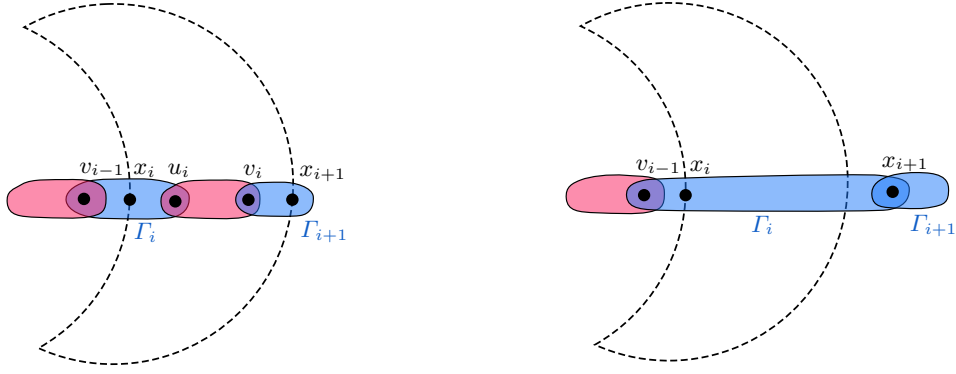
    \caption{Depiction of the summation strategy. \textbf{Left:} A ``small" term $\Gamma_i$. The red represents the factor of $f(u_i, v_i)$ from Lem.~\ref{lem:small_terms}, which we can add if $B_{r_\ast}(x_i)$ is non-resonant. The dotted circle represents the condition $\mathsf d(x_i, x_{i+1}) = r_\ast+1$, reflected in $I_1$ and $I_3$.
    \textbf{Right:} If $\Gamma_i$ is big, then the summations over $u_i, v_i$ can be removed, as reflected in condition $I_2$.}
    \label{fig:placeholder_2}
\end{figure}

To estimate the contribution from the sum over $I_1$, if $x_i, u_i, v_{i-1} \in \Gamma_i$ then we have
\begin{equation}\diam(\Gamma_i) \geq \max(\mathsf d(x_i, u_i), \mathsf d(x_i, v_{i-1})) \geq \frac{1}{2}[\mathsf d(x_i, u_i)+\mathsf d(x_i,
v_{i-1})]
\ ,
\end{equation}which implies that\footnote{Taking this into account in $d>1$ is what leads to the
  $r_\ast/8$ in the exponent, and in $d=1$ or if $V$ is 2-local, this technique yields an improved
  exponent of $r_\ast/4$. In the interest of describing qualitative behavior, we treat the most general case.} 
\begin{align}
  &\sum_{\Gamma_i \ni x_i, u_i, v_{i-1}}\Vert H_{\Gamma_i} \Vert \leq h\delta_{\max(\mathsf d(u_i, x_i), \mathsf d(v_{i-1}, x_i)) \leq k}+ \delta_{\max(\mathsf d(u_i, x_i), \mathsf d(v_{i-1}, x_i)) > k}h\epsilon^{\frac{1}{2}[\mathsf d(u_i, x_i) + \mathsf d(v_{i-1}, x_i)]} \notag \\
 &\leq h(1+k)^{2\beta}\tilde \epsilon^{-\alpha k}f(u_i, x_i)f(v_{i-1}, x_i)
 \label{eq:est1}
\end{align}
where $(1+k)^{2\beta}$ is a constant with respect to $k, d$, which are taken to be fixed. Now we turn to the terms in $I_2$.  Let
\begin{equation}f^{(n)}_\ast(x,y) = \frac{\tilde \epsilon^{\alpha\max(n(r_\ast + 1), \mathsf d(x,y))/2}}{\mathsf d(x,y)^\beta}\end{equation} and notice that $f_{\ast}^{(n)}$ has the reproducing property
\begin{equation}
\sum_{u \in \Lambda}f^{(n)}_\ast(x, u) f^{(m)}_\ast(u, y) \leq K f^{(n+m)}_\ast(x,y)
\end{equation}
If $i < n+1$ then $\Gamma_i$ is big, which implies that there is a point $y \in \Gamma_i$ such that $\mathsf d(x_i, y) > r_\ast$. Since
\begin{equation}
\diam(\Gamma_i) \geq \frac{1}{2}[\mathsf d(v_{i-1}, x_i) + \max(\mathsf d(x_i, y), \mathsf d(x_i, x_{i+1}))] \ ,
\end{equation}
we can bound
\begin{align}
  &\sum_{\Gamma_i \ni v_{i-1}, x_i, x_{i+1}}\Vert H_{\Gamma_i} \Vert\delta_{I_2} \leq \epsilon^{\frac{1}{2}[\mathsf d(v_{i-1}, x_i)+ \max(\mathsf d(x_i, y), \mathsf d(x_i, x_{i+1})]} \leq hf^{(1)}_\ast(x_{i}, x_{i+1})f(v_{i-1}, x_i)
    \label{eq:est2}
\end{align}

To incorporate the terms in $I_3$, we observe that
\begin{equation}
 \delta_{I_3} \leq (r_\ast+1)^\beta\tilde \epsilon^{-\alpha (r_\ast+1)/2}f_{\ast}^{(1)}(x_{i},x_{i+1})\delta_{\neg\operatorname{NR}(x_i)}
  \label{eq:est3}
\end{equation}
Thus for each ball in the path in which the non-resonance condition fails, we
obtain a factor $\epsilon^{-\alpha (r_\ast+1)/2}$, which forces us to decrease
$\lambda$ in the bound.

Putting \eqref{eq:est1}, \eqref{eq:est2}, and \eqref{eq:est3} together, we can
now carry out the sum over $u_i, v_i$ first. Define \begin{equation}
    h' = K^3(1+2k)^\beta \tilde
\epsilon^{-\alpha k}h,
\end{equation}
with which we can uniformly bound all three terms in \eqref{eq:sumoverpaths}:
\begin{align}
  &\sum_{\substack{u_1, \dots, u_n \\ v_1, \dots, v_n}}\prod_{i=1}^n\sum_{\Gamma_i \ni x_i, u_i, v_{i-1}}\Vert H_{\Gamma_i}\Vert \qty(K\qty[1+\frac{ \widetilde h}{s}]^2 \delta_{I_1} f(u_i, v_i) + \delta_{I_2} + \delta_{I_3}) \notag\\
  &\leq h'f(x_{n}, x_{n+1})\prod_{i=1}^{n-1}h'f_\ast(x_{i}, x_{i+1})\qty(\qty[1+\frac{\widetilde h}{s}]^2\delta_{\mathrm{NR}(x_i)} + (r_\ast+1)^\beta \tilde \epsilon^{-\alpha r_\ast/2}\delta_{\neg\mathrm{NR}(x_i)})
\end{align}
Since condition $\neg\operatorname{NR}(x_i)$ is met for at most $(1-\zeta)
n_{\ast}$ choices of $x_i$ by the assumption of condition~\ref{def:partialnonresonance}, we can then bound
the sum over $x_1, \dots, x_{n+1}$ in \eqref{eq:sumoverpaths} by
\begin{align}
  &\sum_{x_1, \dots, x_{n+1}}h'f(x_n, x_{n+1})\prod_{i=1}^{n-1}h'f_\ast(x_{i}, x_{i+1})\qty(\qty[1+\frac{ \widetilde h}{s}]^2\delta_{\mathrm{NR}(x_i)} + r_\ast^\beta \tilde \epsilon^{-\alpha (r_\ast+1)/2}\delta_{\neg\mathrm{NR}(x_i)})\notag\\
  &\leq \epsilon^{-\alpha (r_\ast+1) (1-\zeta)n_\ast/2}(K^2\widetilde h)^{n}\qty(1+\frac{\widetilde h}{s})^{2(n-1)}\sum_{x_1, \dots, x_{n+1}}\prod_{i=1}^{n-1}f_\ast^{(1)}(x_{i}, x_{i+1})
\end{align}
where the factor of $(1+r_\ast)^\beta$ is included in $\widetilde h$. By construction, if $\Gamma \in \Gamma^{(n_\ast)}(S_0)$ terminates at $|\Gamma| < n_\ast+1$, we must have $\mathsf d(x_1, x_n) \geq
n_\ast(r_\ast+1)$. This implies
\begin{align}
\sum_{x_1, \dots, x_{n+1}}\prod_{i=1}^{n-1}f_\ast^{(1)}(x_{i}, x_{i+1})
&\leq K^{n-1}\sum_{x_1 \in \partial S_0,x_{n+1}}f^{(n_\ast)}_{\ast}(x_1, x_{n+1}) \leq C |\partial S_0|\tilde \epsilon^{\alpha n_\ast(r_\ast+1)/2}
\end{align}
where we introduced a constant $C$ such that the factor of
$\mathsf d(x_1, x_{n+1})^{-\beta}$ in $f_\ast^{(n_\ast)}(x_1, x_{n+1})$
cancels with the sum over final points $x_{n+1}$.
Plugging this into the sum over $n$ from \eqref{eq:sumoverpaths}, we can bound
\begin{align}
\sum_{n=1}^{n_\ast+1}\qty(\frac{2K^3\widetilde h}{s})^{n}\qty(1+\frac{\widetilde h}{s})^{2(n-1)}
&\leq \sum_{n=1}^{3n_\ast+1}\qty(\frac{8K^3\widetilde h}{s})^{n}
\end{align}
where we used the fact that the sum of coefficients in the polynomial
$x^n(1+x)^{2n}$ are bounded by $4^n$, which follows from setting $x = 1$. Using
the assumption that $\tilde \epsilon < \e^{-1}$, we find
\begin{align}
  \epsilon^{\alpha (r_\ast+1) \zeta n_\ast/2}\sum_{n=1}^{3n_\ast+1}\qty(\frac{8K^3\widetilde h}{s})^{n}
  &\leq\sum_{n=1}^{3n_\ast+1}\tilde\epsilon^{\alpha (r_\ast+1) \zeta (3n_\ast-n)/8}\qty(\frac{8K^3\widetilde h \epsilon^{\alpha \zeta (r_\ast+1)/8}}{s})^{n} \notag \\
  &\leq \e^{-\alpha \zeta\mathsf d(S_0, S_1) /8}\sum_{n=1}^{3n_\ast+1}\qty(\frac{8K^3\widetilde h (\e\epsilon)^{\alpha \zeta (r_\ast+1)/8}}{s})^{n}
  \label{eq:before_Laplace_transform}
\end{align}
where we used $n_\ast(r_\ast+1) \geq \frac{1}{3}\mathsf d(S_1, S_0)$.
Multiplying by $1/s$ [from \eqref{eq:comm}] we take the upper bound of the sum to $\infty$ and perform the inverse Laplace transform to find
\begin{align}
\Vert[B_{S_0}, A_{S_1}(t)]\Vert \leq \frac{2C}{K}\Vert A_{S_1} \Vert\Vert B_{S_0} \Vert |\partial S_0|\qty[\exp(8K^3\widetilde h (\e\tilde \epsilon)^{\alpha \zeta(r_\ast+1)/8})-1]\exp(- \frac{\alpha \zeta}{8}\mathsf d(S_0, S_1))
\end{align}
Then to finish the claim, recall the definition \eqref{eq:widetilde_h} and the assumption that $\zeta r_\ast/8 \geq k$.  For any $\alpha' < \alpha$ we can choose a constant
$C'$ such that
\begin{equation}
\widetilde h(\e\tilde\epsilon)^{\alpha \zeta (r_{\ast}+1)/8} < C'h(\e\tilde \epsilon)^{\alpha' (\zeta(r_\ast+1)/8-k)}
\end{equation}
To get the prefactor $\min(|\partial S_0|, |\partial S_1|)$, we only need to
note that $\Vert [B_{S_0}, A_{S_1}(t)]\Vert = \Vert[A_{S_1}, B_{S_0}(-t)]\Vert$,
and the same argument can be applied. To write the bound in terms of $r =
\min(r_\ast + 1, \mathsf d(S_0, S_1))$, note that we do not have to use the largest possible $r_\ast$
given to us by the non-resonance condition, and in fact if $\mathsf d(S_0, S_1)< r_*$ we may take any alternative value $\tilde r_\ast <
\mathsf d(S_0, S_1)-1$ freely in the bound.
\end{proof}

\subsection{Proof of Lemma \ref{lem:jump_paths}}\label{sec:proof_jumppaths}

The relatively straightforward proof follows from the iterated
Duhamel expansion:
\begin{proof}[Proof of Lem.~\ref{lem:jump_paths}]
  First, we will show that
  \begin{align}
  \label{eq:jump_paths}
    \mathcal B_{S_0}\e(\mathcal H) A_{S_1} = \sum_{\Gamma \in \Gamma(S_0\to S_1)}(\ii)^{|\Gamma|}\mathcal B_{S_0}\e(\mathcal H)\ast \mathcal H_{\Gamma_1}\e(\overlinebold P_{\Sint{0}(\Gamma)}\mathcal H)
    \ast\mathcal H_{\Gamma_2}\e(\overlinebold P_{\Sint{1}(\Gamma)}\mathcal H) \ast \dots  \ast \mathcal H_{\Gamma_l}\e(\overlinebold P_{\Sint{l-1}(\Gamma)}\mathcal H)A_{S_1}
  \end{align}
This is essentially identical to the expansion in \eqref{eq:Schwinger_Karplus} with the paths constructed in Def.~\ref{def:generate_paths}.

We will prove this inductively. Since $n_\ast \geq 1$ by assumption, our base case is $n_\ast = 1$, where we have
\begin{equation}\Gamma^{(1)}(S_0) = \{\Gamma \subseteq \Lambda: \Gamma \cap S_0 \neq \emptyset \text{ and  } \Gamma \cap S_0^{\mathrm{c}} \neq \emptyset\}
\end{equation}
Since $S_0 \cap S_1 = \emptyset$, a single Duhamel expansion with respect to the decomposition \begin{equation}H = \sum_{\Gamma \in \Gamma^{(1)}(S_0)}H_\Gamma + \sum_{\Gamma \notin \Gamma^{(1)}(S_0)}H_{\Gamma} = \mathbb P_{S_0}\mathbb P_{S_0^{\mathrm{c}}}H + (\overlinebold P_{S_0}\mathbb P_{S_0^{\mathrm{c}}} + \mathbb P_{S_0} \overlinebold P_{S_0^{\mathrm{c}}} + \overlinebold P_{S_0} \overlinebold P_{S_0^{\mathrm{c}}})H \end{equation} yields
\begin{align}
 \mathcal B_{S_0}\e(\mathcal H) A_{S_1} &= \mathcal B_{S_0}\e([\overlinebold P_{S_0}\mathbb P_{S_0^{\mathrm{c}}} + \mathbb P_{S_0} \overlinebold P_{S_0^{\mathrm{c}}} + \overlinebold P_{S_0} \overlinebold P_{S_0^{\mathrm{c}}}]\mathcal H)A_{S_1} + \ii\sum_{\Gamma \in \Gamma^{(1)}}\mathcal B_{S_0} \e(\mathcal H) \mathcal H_{\Gamma} \e([\overlinebold P_{S_0}\mathbb P_{S_0^{\mathrm{c}}} + \mathbb P_{S_0} \overlinebold P_{S_0^{\mathrm{c}}} + \overlinebold P_{S_0} \overlinebold P_{S_0^{\mathrm{c}}}]\mathcal H)A_{S_1} \notag \\
 &= \ii\sum_{\Gamma \in \Gamma^{(1)}}\mathcal B_{S_0} \e(\mathcal H) \mathcal H_{\Gamma} \e(\overlinebold P_{S_0}\mathcal H)A_{S_1}
\end{align}
Since $\Sint{0} = S_0$, this completes the base case. 

For the inductive step, let $X = \Sint{l}\backslash\Sint{l-1}(\Gamma)$ be the region we wish to decouple from the rest of the
system, as depicted in Fig.~\ref{fig:examplepaths}.
Consider the decomposition
\begin{align}
H \equiv \sum_{S:\substack{S \cap X \neq \emptyset\\ S \cap X^{\mathrm{c}}\neq \emptyset}}H_{S} + \sum_{S:S \subseteq X}H_S + \sum_{S:S \subseteq X^{\mathrm{c}}}H_S = H_{\mathrm{bdy}} + H_{\mathrm{int}} + H_{\mathrm{ext}}
\end{align}
Then performing the Duhamel expansion on $\exp(\overlinebold P_{\Sint{i-1}(\Gamma)}\mathcal H)$ gives
\begin{align}
  &\mathcal H_{\Gamma_l}\e(\overlinebold P_{\Sint{l-1}(\Gamma)}\mathcal H)A_{S_1} \notag\\
  &= \mathcal H_{\Gamma_l}[\e(\overlinebold P_{\Sint{l-1}(\Gamma)}[\mathcal H_{\mathrm{ext}}+\mathcal H_{\mathrm{int}}]) + \ii\e(\overlinebold P_{\Sint{l-1}(\Gamma)}\mathcal H) \overlinebold P_{\Sint{l-1}(\Gamma)}\mathcal H_{\mathrm{bdy}} \ast\e(\overlinebold P_{\Sint{l-1}(\Gamma)}[\mathcal H_{\mathrm{int}} + \mathcal H_{\mathrm{ext}}])]A_{S_1} \notag\\
  &= \ii\sum_{S:\substack{S \cap \Sint{l-1}(\Gamma) = \emptyset \\ S\cap  \partial X^{\mathrm{c}} \neq \emptyset}}\mathcal H_{\Gamma_l}\e(\overlinebold P_{\Sint{l-1}(\Gamma)}\mathcal H)\mathcal H_{S} \ast \e(\overlinebold P_{\Sint{l}(\Gamma)}\mathcal H_{\mathrm{ext}})A_{S_1}
\end{align}
The first term in the second line vanishes because $\overlinebold
P_{\Sint{l-1}(\Gamma)}H_{\mathrm{ext}}$ and $\overlinebold P_{\Sint{l-1}(\Gamma)}H_{\mathrm{int}}$
act on disjoint regions containing $S_1$ and $\Gamma_l$ by construction.
Similarly, $\e(\overlinebold P_{\Sint{l-1}(\Gamma)}\mathcal H_{\mathrm{int}})A_{S_1} = A_{S_1}$ and $[\mathcal H_{\mathrm{int}}, \mathcal H_{\mathrm{ext}}] = \ad_{[H_{\mathrm{int}}, H_{\mathrm{ext}}]} = 0$.
Thus, put $X = \Sint{l}(\Gamma)$.
We can write the above compactly as
\begin{align}
  \mathcal H_{\Gamma_l}\e(\overlinebold P_{\Sint{l}(\Gamma)}\mathcal H)A_{S_1}
  &=
  \ii\sum_{(\Gamma, \Gamma_{l+1}) \in \Gamma^{(l+1)}(S_0)}\mathcal H_{\Gamma_l}\e(\overlinebold P_{\Sint{l-1}(\Gamma)}\mathcal H)\mathcal H_{\Gamma_{l+1}} \ast \e(\overlinebold P_{\Sint{l}(\Gamma)}\mathcal H)A_{S_1}
\end{align}
This finishes the induction. 

By construction, $\e(\overlinebold P_{\Sint{i-1}}\mathcal H)\mathcal H_{\Gamma_{i+1}}\e(\overlinebold P_{\Sint{i}}\mathcal H)\dots A_{S_1}$ is supported within $\Sint{i-1}(\Gamma)^{\mathrm{c}}$, which means that after applying $\mathcal H_{\Gamma_i}$, the operator $\mathcal H_{\Gamma_i}\e(\overlinebold P_{\Sint{i-1}}\mathcal H)\mathcal H_{\Gamma_{i+1}}\e(\overlinebold P_{\Sint{i}}\mathcal H)\dots A_{S_1}$ is supported within $\Sint{i-1}(\Gamma)^{\mathrm{c}} \cup \Gamma_i$. Therefore, we can insert projectors into \eqref{eq:jump_paths} in the following way:
\begin{align}
  \mathcal B_{S_0}\e(\mathcal H)A_{S_1} &= \ii\sum_{\Gamma \in \Gamma^{(n_\ast)}(S_0)}
  \qty[\mathcal B_{S_0}\e(\mathcal H)] \ast \qty[\mathcal H_{\Gamma_1} \e(\overlinebold P_{\Sint{0}(\Gamma)}\mathcal H) \overlinebold P_{\Sint{1}(\Gamma)\backslash \Gamma_2}] \ast \dots \\
  &\ast \qty[\mathcal H_{\Gamma_{l-1}}\e(\overlinebold P_{\Sint{l-2}}\mathcal H)\overlinebold P_{\Sint{l-1}(\Gamma)\backslash \Gamma_l}] \ast \qty[\mathcal H_{\Gamma_l}\e(\overlinebold P_{\Sint{l-1}(\Gamma)}\mathcal H)\overlinebold P_{\Sint{l}\backslash S_1}]A_{S_1}
\end{align}
Using the submultiplicativity of the $\infty-\infty$ norm (Def.~\ref{defn:inftyinftynorm}) gives
\begin{align}
\Vert \mathcal B_{S_0} \e(\mathcal H) A_{S_1}\Vert &\leq 2\Vert B_{S_0}\Vert \Vert A_{S_1}\Vert\sum_{\Gamma \in \Gamma(S_0 \to S_1)} 1 \ast \Vert  \mathcal H_{\Gamma_1}\e(\overlinebold P_{\Sint{0}(\Gamma)}\mathcal H)\overlinebold P_{\Sint{1}\backslash \Gamma_{2}}\Vert_{\infty-\infty} \ast \notag\\
&\dots 
    \ast \Vert \mathcal H_{\Gamma_l}\e(\overlinebold P_{\Sint{l-1}(\Gamma)}\mathcal H)\overlinebold P_{\Sint{l}(\Gamma)\backslash S_1}\Vert_{\infty-\infty}
\end{align}
This completes the proof.
\end{proof}

\subsection{Proof of Lemma \ref{lem:small_terms}}\label{sec:proof_C4}

\begin{proof}[Proof of Lem.~\ref{lem:small_terms}]
  First, we may write
  \begin{align}
  C_{i}^{(\Gamma)}(s) = \sup_{\Vert O\Vert = 1}\Vert \mathcal H_{\Gamma_i}\e(\overlinebold P_{\Sint{i-1}(\Gamma)}\mathcal H)\overlinebold P_{\Sint{i}(\Gamma)\backslash \Gamma_{i+1}}O \Vert = \sup_{\Vert O_{\mathrm{bdy}}\Vert = 1}\Vert \mathcal H_{\Gamma_i}\e(\overlinebold P_{\Sint{i-1}(\Gamma)}\mathcal H)O_{\mathrm{bdy}}\Vert 
\end{align}
where $O_{\mathrm{bdy}}$ is an operator supported strictly within
$\Sext{i}(\Gamma) \cup \Gamma_{i+1}$. For shorthand, we will drop the $\Gamma$
argument of $\Sint{i}$ and $\Sext{i}$, and for any subset $S$ we will write $\overlinebold P_{S}H = H|_{S^{\mathrm{c}}}$.
Now by hypothesis $\Gamma_i \cap \Sext{i-1} \subseteq B_{r_{\ast}}(x_i)$, so
consider $U^\SW$ constructed on $B_{r_\ast}(x_i) \cap \Sext{i-1}$. It is important that $U^\SW$ is
supported within $\Sext{i-1} \cap \Sint{i}$ so that no two $U^\SW$ constructed with the same
$\Gamma$ for different choices of $i$ ever share support.
Using the invariance of the operator norm, we have 
\begin{align}
  \Vert \mathcal H_{\Gamma_i}\e(\mathcal H \resext)O_{\mathrm{bdy}}\Vert &= 
  \Vert \mathcal H^\SW_{\Gamma_i} \e(\mathcal H^{\SW} \resext)O_{\mathrm{bdy}}^{\SW} \Vert \ ,
\label{eq:preexpansion}
\end{align}
where $O^{\SW} \equiv U_{\SW}^\dagger O U_{\SW}$ for an operator $O$. Then take local decompositions $\mathcal H^{\SW}_{\Gamma_i} = \sum_{S_2}\tildemathcal
H_{S_2}$ and $O_{\mathrm{bdy}}^\SW = \sum_{S_3}\widetilde O_{S_3}$. We can decompose
\begin{align}
  \mathcal H_{\Gamma_i}^\SW \e(\mathcal H^{\SW}\resext)O^\SW_{\mathrm{bdy}} &=
  \underbrace{\sum_{\mathcal S_2 \cap S_3 \neq \emptyset}\tildemathcal H_{S_2}\e(\mathcal H^\SW\resext)\widetilde O_{S_3}}_{(\rm{IA})}
  + \sum_{S_2 \cap S_3 = \emptyset}\tildemathcal H_{S_2}\e(\mathcal H^\SW\resext)\widetilde O_{S_3}
  \label{eq:exp1}
\end{align}
The first term is already optimal, because $S_2 \cup S_3$ contains a path from
$\Gamma_i$ to $\Sext{i}\cup \Gamma_{i+1}$, so we have labeled it (IA).

Now we expand the second term. First consider the decomposition
\begin{align}
H = \mathbb P_{\Sint{i}} \overlinebold P_{\Sext{i}}H + \mathbb P_{\Sint{i}}\mathbb P_{\Sext{i}}H + \mathbb P_{\Sext{i}} \overlinebold P_{\Sint{i}}H  = H_{\mathrm{int}} + H_{\mathrm{bdy}} + H_{\mathrm{ext}}
\end{align}
We will denote $H_{\mathrm{int}}^\SW|_{\Sext{i-1}} = Z_{\mathrm{int}}$ to emphasize that it is
commuting. Since $\Sext{i}$ does not intersect $B_{r_\ast}(x_i)$ by
construction, this means that every term in $\mathbb P_{\Sext{i}}H$ has support
on at least one site at a distance $r_\ast+1$ from $x_i$.

Now for each $S_2$ with $S_2 \cap S_3 = \emptyset$, we decompose $H^{\text{SW}}\resext = H' + \mathbb P_{S_2}H_{\mathrm{bdy}}^\SW$ where 
$H' = Z_{\mathrm{int}} + \overline{\mathbb P}_{S_2}H^\SW_{\mathrm{bdy}} + H_{\mathrm{ext}}$.
Using the Duhamel expansion on $\e(\mathcal H^{\SW}\resext)$ in the second term of Eq.~\eqref{eq:exp1}, we have
\begin{align}
&\sum_{S_2 \cap S_3 = \emptyset}\tildemathcal H_{S_2}\e(\mathcal H^{\SW})\widetilde O_{S_3} = \sum_{S_2 \cap S_3 = \emptyset}\tildemathcal H_{S_2}\e( \mathcal H')\widetilde O_{S_3} +\underbrace{\ii\sum_{S_2 \cap S_3 = \emptyset}\sum_{S_3' \cap S_2 \neq \emptyset}\tildemathcal H_{S_2}\e(\mathcal H^\SW\resext )\ast (\mathcal H_{\mathrm{bdy}}^\SW)_{S_3'}\e(\mathcal H')\widetilde O_{\mathrm{bdy}})_{S_3}}_{(\rm{IB})}
\label{eq:exp2}
\end{align}
The second term is again optimal, because $S_2 \cup S_3'$ contains a
path connecting $\Gamma_i$ and $\Sext{i} \cup \Gamma_{i+1}$, so we continue to expand the first term.

Consider the decomposition $H' = H'' + \mathbb P_{S_2}Z_{\mathrm{int}}$, where $H'' = \overline {\mathbb P}_{S_2}Z_{\mathrm{int}} + \overline {\mathbb P}_{S_2}H_{\mathrm{bdy}}^\SW+H_{\mathrm{ext}}$. Applying the Duhamel identity,
\begin{align}
  &\sum_{S_2 \cap S_3 = \emptyset}\mathcal H_{S_2}\e(\mathcal H')\widetilde O_{S_3}\notag\\
  &=\ii\underbrace{\sum_{S_2 \cap S_3 = \emptyset}\sum_{S_2':\substack{S_2' \cap S_2 \neq \emptyset \\ S_2' \cap S_3 \neq \emptyset}}\tildemathcal H_{S_2}\e(\mathcal{H}')\ast (\mathcal Z_{\mathrm{int}})_{S_2'}\e(\mathcal{H}'')\widetilde O_{S_3}}_{(\rm{IIA})}
+\ii\sum_{S_2 \cap S_3 = \emptyset}\sum_{S_2':\substack{S_2' \cap S_2 \neq \emptyset \\ S_2' \cap S_3 = \emptyset}}\tildemathcal H_{S_2}\e(\mathcal{H}')\ast (\mathcal Z_{\mathrm{int}})_{S_2'}\e(\mathcal{H}'')\widetilde O_{S_3}
\label{eq:exp3}
\end{align}
The first term in the Duhamel expansion vanishes:
\begin{align}
\ii\sum_{S_3 \cap S_2 = \emptyset}\tildemathcal H_{S_2}\e(\mathcal H'')\widetilde O_{S_3} = 0
\end{align}
because $H''$ is not supported on $S_2$. The first term has a path from $\Gamma_i$ to
$\Sext{i} \cup \Gamma_{i+1}$ contained in $S_2 \cup S_2' \cup S_3$.

Lastly,
we apply the Duhamel identity to the last part of the second term in Eq.~\eqref{eq:exp3}.
Expanding $H'' = H''' + \mathbb P_{S_2'}\overline{\mathbb P}_{S_2}\mathcal
H_{\mathrm{bdy}}^{\text{SW}}$, where $H''' = \overline{\mathbb P}_{S_2'}\overline{\mathbb
  P}_{S_2}H_{\mathrm{bdy}}^\SW+\overline{\mathbb P}_{S_2}Z_{\mathrm{int}} +H_{\mathrm{ext}}$, we have
\begin{align}
(\mathcal Z_{\mathrm{int}})_{S_2'}\e(\mathcal H'')\widetilde O_{S_3} &= 
\ii\sum_{S_3' \cap S_2' \neq \emptyset}(\mathcal Z_{\mathrm{int}})_{S_2'}\e(\mathcal H'')\ast (\mathcal H_{\mathrm{bdy}}^\SW)_{S_3'}\e(\mathcal H''')\widetilde O_{S_3}
\label{eq:exp3prime}
\end{align}
Dropping the first term from the Duhamel expansion in the equation above is where we use the fact that $Z_{\mathrm{int}}$ is commuting. This term is the following:
\begin{align}
 (\mathcal Z_{\mathrm{int}})_{S_2'}\e(\mathcal H''')\widetilde O_{S_3} = 0
\end{align}
The only term in $H'''$ with support overlapping  $S_2'$ is
$\overline{\mathbb P}_{S_2}Z_{\mathrm{int}}$. Since $S_3 \cap S_2' = \emptyset$, any
operator in the Taylor expansion of $\e(\mathcal H''')\widetilde O_{S_3}$ supported on $S_2'$ must be a Pauli-Z operator, and so commutes with $(Z_{\mathrm{int}})_{S_2'}$.

Now inserting Eq.~\eqref{eq:exp3prime} into Eq.~\eqref{eq:exp3}, we get the last term
\begin{align}
&\ii\sum_{S_2 \cap S_3 = \emptyset}\sum_{S_2':\substack{S_2' \cap S_2 \neq \emptyset \\ S_2' \cap S_3 = \emptyset}}\tildemathcal H_{S_2}\e(\mathcal{H}')\ast (\mathcal Z_{\mathrm{int}})_{S_2'}\e(\mathcal{H}'')\widetilde O_{S_3} \notag\\
&= \underbrace{-\sum_{S_3 \cap S_2 = \emptyset}\sum_{S_2':\substack{S_2' \cap S_2 \neq \emptyset \\ S_2' \cap S_3 = \emptyset}}\sum_{\substack{S_3' \cap S_2' \neq \emptyset \\ S_3' \cap S_2 = \emptyset}}\tildemathcal H_{S_2}\e(\mathcal H') \ast (\mathcal Z_{\mathrm{int}})_{S_2'}\e(\mathcal H'')\ast (\mathcal H_{\mathrm{bdy}}^\SW)_{S_3'}\e(\mathcal H''')\widetilde O_{S_3}}_{(\rm{IIB})}
\label{eq:exp4}
\end{align}
Since $S_2 \cup S_2' \cup S_3'$ contains a path from $\Gamma_i$ to $\Sext{i} \cup \Gamma_{i+1}$, we
have now completed the expansion. While the computation is tedious, the interpretation in terms of paths to the boundary is quite simple; This is illustrated pictorially in Fig.~\ref{fig:legos_placeholder}.
\begin{figure}[bt!]
    \centering
    \def\svgwidth{0.95\linewidth} 
    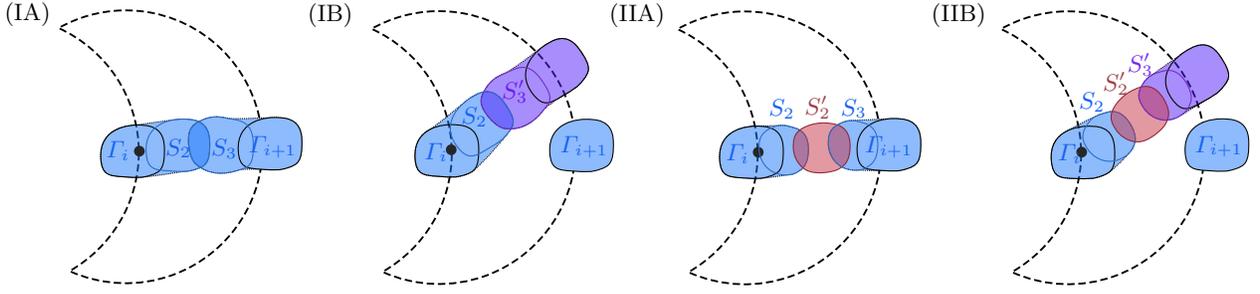
    \caption{Figure illustrating the interpretation of the conclusion of the lemma as a sum over irreducible paths, sorted into four different types of paths. The sets $S_2$ and $S_3$ result from applying the quasilocal SW transform within $B_{r_\ast}(x_i)$ to $H_{\Gamma_i}$ and $H_{\Gamma_{i+1}}$. In (IB) and (IIB), the path ends at the boundary, but not at $\Gamma_{i+1}$, so we sum over endpoints $v$, accruing a factor of $A_{\ast}$ which will be absorbed into $\tilde \epsilon^{\alpha r_\ast/2}$. The red sets illustrate where a term from $Z_{\mathrm{int}}$ occurs in the path, and because it is commuting, we can design the paths such that only one term from $Z_{\mathrm{int}}$ becomes an irreducible path component.}
    \label{fig:legos_placeholder}
\end{figure}

In terms of $\widetilde C_{i}^{(\Gamma)}(s)$, this expansion leads to the following bound:
\begin{align}
  \widetilde C_{i}^{(\Gamma)}(s) \leq
  \Vert \mathrm{(IA)} \Vert + \Vert \mathrm{(IB)} \Vert + \Vert \mathrm{(IIA)} \Vert + \Vert \mathrm{(IIB)} \Vert
  \label{eq:recursion}
\end{align}
where the terms labeled (IA), (IB), (IIA), and (IIB) come from equations
\eqref{eq:exp1},\eqref{eq:exp2}, \eqref{eq:exp3}, \eqref{eq:exp4} respectively. 

Let
$f(x,y) \equiv c\tilde\epsilon^{\alpha \mathsf d(x,y)/2}/\mathsf d(x,y)^\beta$
be reproducing with parameter $K$.  Starting with (IA),
\begin{align}
  \Vert \mathrm{(IA)} \Vert &\leq 
 \frac{2}{s}\sum_{\mathcal S_2 \cap S_3 \neq \emptyset}\Vert \widetilde H_{S_2}\Vert\Vert \widetilde O_{S_3}\Vert
\end{align}
Then the sum over subsets is bounded using the same set-intersection-counting
argument as explain in Sec.~\ref{sec:lrbound_duhamel}. Importantly, since
the support of $U^\SW$ only intersects the support of $O_{\mathrm{bdy}}$ at
$\Gamma_{i+1}$, when we apply Lem.~\ref{lem:SW_transformed_LRbound} we can use
$S = \Gamma_{i+1}$. A way to visualize this is that $U^\SW$ can only `grow'
$O_{\mathrm{bdy}}$ into the interior of $B_{r_\ast+1}(x_i)$.
This results in the following:
\begin{align}
  \sum_{\mathcal S_2 \cap S_3 \neq \emptyset}\Vert \widetilde H_{S_2}\Vert\Vert \widetilde O_{S_3}\Vert
  &\leq \sum_{v \in \Lambda}\Vert H^{\SW}_{\Gamma_i}\Vert_{S_2,v}\Vert \widetilde O_{\mathrm{bdy}}\Vert_{v, S_3}
 \leq \Vert H_{\Gamma_i}\Vert K\sum_{\substack{u \in S_2 \\ w \in S_3}}f(u, w)
\end{align}
Then (IB),
\begin{align}
\Vert \mathrm{(IB)} \Vert &\leq 
  \frac{4}{s^2}\sum_{S_2 \cap S_3 = \emptyset}\sum_{S_3' \cap S_2 \neq \emptyset}\Vert \widetilde H_{S_2}\Vert\Vert(H_{\mathrm{bdy}}^\SW)_{S_3'}\Vert\Vert \widetilde O_{S_3}\Vert
\end{align}
Counting set intersections and applying
Lem.~\ref{lem:SW_transformed_LRbound} gives us
\begin{align}
\sum_{S_2 \cap \Gamma_i \neq \emptyset}\sum_{S_3' \cap S_2 \neq \emptyset}\Vert \widetilde H_{S_2}\Vert\Vert(H_{\mathrm{bdy}}^\SW)_{S_3'}\Vert &\leq
\Vert H_{\Gamma_i}\Vert\sum_{S_3}\Vert (H_{\mathrm{bdy}})_{S_3}\Vert K\sum_{\substack{u \in \Gamma_i \\ w \in S_3}}f(u, w) \notag \\
&\leq \Vert H_{\Gamma_i}\Vert K\sum_{u \in \Gamma_i}\sum_{v \in \partial B_{r_\ast+1}(x_i)}\sum_{w \in \Lambda}f(u, w)\Vert H_{\mathrm{bdy}}\Vert_{w, v} \notag\\
  &\leq 2\Vert H_{\Gamma_i}\Vert K^2hA_\ast(1+k)^{\beta}\tilde \epsilon^{-\alpha k}\sum_{u \in S_i}f(u, \partial B_{r_\ast+1}(x)) \notag\\
  &\leq \Vert H_{\Gamma_i}\Vert 2K\widetilde h\sum_{u \in \Gamma_i}\sum_{v \in \Gamma_{i+1}}f(u, v)
\end{align}
where $A_{\ast} = |\partial B_{r_\ast+1}(x)|$. In the second step, we used the
fact that every term in $H_{\mathrm{bdy}}$ is supported on $\Sext{i}$. In the third step, we observed that
\begin{align}
  \Vert H_{\mathrm{bdy}} \Vert_{v, w} &\leq \Vert H_0 \Vert_{v,w} + \sum_{r=1}^{\infty}\epsilon^r \Vert V_{r} \Vert _{v, w}
\leq \mathds 1[\mathsf d(v, w) \leq k]\sum_{S \ni w}\Vert (H_0)_S \Vert + \sum_{r = \mathsf d(v, w)}^{\infty}\epsilon^r \sum_{S \ni w}\Vert (V_r)_S \Vert \notag\\
  &\leq 2h\epsilon^{\mathsf d(v,w)-k} \leq 2h (1+k)^\beta \tilde \epsilon^{-\alpha k/2} f(v, w) \leq \frac{\widetilde h}{K} f(v,w)
\end{align}
where we have also used the definition for $\widetilde h$ in \eqref{eq:widetilde_h}.  Then we summed over the endpoints $w$, and for convenience, over-bounded this
with a sum over endpoints in $\Gamma_{i+1}$ in the last step.

The remaining terms (IIA) and (IIB) exactly mirror the manipulations for (IA)
and (IB). To account for the additional term from $Z_{\mathrm{int}}$, we apply the
estimate from Cor.~\ref{cor:HSW_locality}, which gives $\Vert Z_{\mathrm{int}}\Vert_{u,v}
\leq \frac{\widetilde h}{K} f(u,v)$:
\begin{align}
\Vert \mathrm{(IIA)} \Vert &\leq 
 \frac{4}{s^2}\sum_{S_2 \cap S_3 = \emptyset}\sum_{S_2':\substack{S_2' \cap S_w \neq \emptyset \\ S_2' \cap S_3 \neq \emptyset}}\Vert \widetilde H_{S_2}\Vert \Vert (Z_{\mathrm{int}})_{S_2'}\Vert\Vert\widetilde O_{S_3}\Vert\leq \Vert H_{\Gamma_i}\Vert\frac{2K\widetilde h}{s^2}\sum_{\substack{u \in \Gamma_i \\ v \in \Gamma_{i+1}}}f(u, v)
\end{align}
Lastly:
\begin{align}
\Vert \mathrm{(IIB)} \Vert &\leq 
  \frac{8}{s^3} \sum_{S_3 \cap S_2 = \emptyset}\sum_{S_2':\substack{S_2' \cap S_2 \neq \emptyset \\ S_2' \cap S_3 = \emptyset}}\sum_{\substack{S_3' \cap S_2' \neq \emptyset \\ S_3' \cap S_2 = \emptyset}}\Vert\widetilde H_{S_2}\Vert\Vert (Z_{\mathrm{int}})_{S_2'}\Vert \Vert (H_{\mathrm{bdy}}^\SW)_{S_3'}\Vert  \Vert \widetilde O_{S_3}\Vert \leq \Vert H_{\Gamma_i}\Vert\frac{2K\tilde h^2}{s^3}\sum_{u \in \Gamma_i, v \in \Gamma_{i+1}}f(u, v)
\end{align}
Inserting this into \eqref{eq:recursion},
\begin{align}
  \widetilde C_{i}^{(\Gamma)}(s) \leq \frac{2K\Vert H_{\Gamma_i}\Vert}{s}\qty(1+\frac{\widetilde h}{s})^2\sum_{u \in \Gamma_i,v \in \Gamma_{i+1}}f(u,v)
  \label{eq:smalliter}
\end{align}
which completes the proof.
\end{proof}

\subsection{Simulation of non-resonant systems}

As an immediate corollary, we find that non-resonant dynamics are well-approximated by local dynamics up to much later times than a generic Hamiltonian, simply because non-resonance makes the Lieb-Robinson velocity so small:

\begin{cor}
Let $H$ satisfy the hypothesis of Thm.~\ref{thm:main_thm}.  Given $S \subseteq \Lambda$, let $S_r = \{x : \mathsf d(S, x) \leq r\}$ be the ``thickening" of $S$ by $r$ sites. Given an operator $A_{S}(t)$ where $S$ is connected, let \begin{equation}A^{(r)}_S(t) \equiv \e(\overlinebold P_{S_r^{\mathrm{c}}}H)(t) A_S \end{equation}
be a local approximant for $A_S(t)$. Then for any $\mu < \frac{\alpha \zeta}{8}$ there exists a constant $C$ such that the approximation error is bounded by 
\begin{align}
\Vert A_S(t) - A^{(r)}_S(t) \Vert \leq C|S|\Vert A_{S} \Vert \qty[\exp(C'h(\e \tilde \epsilon)^{\alpha[\zeta R/8-k]}t)-1]\exp(-\mu r)
\end{align}
where $R = \min(r_\ast +1, r)$, so long as $R > \frac{8k}{\zeta}$.
\end{cor}

\begin{proof}
First, we Duhamel-expand the evolution under $H$:
\begin{align}
\e(\mathcal H ) - \e(\overlinebold P_{S_r^{\mathrm{c}}}\mathcal H) = \sum_{\substack{X \cap S_r^{\mathrm{c}} \neq \emptyset \\ X \cap S_r \neq \emptyset}}\e(\mathcal H)\ast\mathcal \mathcal H_{X} \e(\overlinebold P_{S_r^{\mathrm{c}}}\mathcal H)
\end{align}
Then we may treat $X  = \Gamma_1$ as the first term in $\Gamma(S_r^{\mathrm{c}} \to S)$ and continue the expansion, arriving at
\begin{align}
\Vert [\e(\mathcal H t) - \e(\overlinebold P_{S_r^{\mathrm{c}}}\mathcal H)]A_S \Vert
\leq 2 \Vert A_{S}\Vert\sum_{\Gamma \in \Gamma(S_r^{\mathrm{c}} \to S)}1 \ast C_{1}^{(\Gamma)} \ast \dots \ast C_{|\Gamma|}^{(\Gamma)}(t)
\end{align}
This is, of course, the sum which we just bounded in our main theorem. Applying the main theorem with $|\partial S_r| \leq B|S|r^d$ for appropriate constant $B$, we find
\begin{align}
\Vert A_S(t) - A^{(r)}_S(t) \Vert \leq C\Vert A_S \Vert|S|\qty[\exp(C'h\tilde \epsilon^{\alpha (\zeta R/8-k)}t)-1]\e^{-\mu r}
\end{align}
where we chose $C$ for any $\mu < \frac{\alpha \zeta}{8}$ to cancel the factor of $r^d$.
\end{proof}

This gives us a tool for simulating non-resonant systems by truncating the dynamics to within $S^r$. With it, we can show that classically, non-resonant systems are especially easy to simulate to late times:
\begin{cor}
Let $\rho_0$ be an initial state such that given any $S \subseteq \Lambda$, the marginal $\Tr_{S^{\mathrm{c}}}(\rho_0)$ can be obtained at cost $\e^{\mathrm O(|S|)}$. Consider the problem of computing the correlation function $\Tr[\rho_0 A(t) B]$. For any $\gamma > 0$, there exists an algorithm outputting this correlation function up to error 
\begin{align}
\qty|\Tr[A_S(t)B\rho] - \Tr[A_S^{(r)}(t)B\rho]| \leq \gamma
\end{align}
using $N \leq \exp[\mathrm O(|S|r^d)]$ classical resources,
where
\begin{align}
r = \log(\frac{|S|}{\gamma}) + \max\qty(\frac{\log(t)}{\log(1/\tilde \epsilon)}, \tilde \epsilon^{\alpha (\zeta[r_\ast+1]/8 - k)}t) \label{eq:simulation_r}
\end{align}
where the first term is due to the logarithmic lightcone for $r \leq r_\ast$ and the second term is due to the ballistic lightcone when $r \geq r_\ast$. Note that in this bound, we treat $\alpha$, $\zeta$, $k$ as constants.
\end{cor}

\begin{proof}
First, by the previous corollary, we have the error estimate
\begin{align}
\gamma = \qty|\Tr[A_S(t)B \rho] - \Tr[A_S^{(r)}(t)B\rho]| \leq \Vert A_S(t) - A_S^{(r)}(t) \Vert \Vert B \Vert \leq \mathrm O(|S|\exp(vt - \mu r))
\end{align}
where $v(r) = C'h\tilde \epsilon^{\alpha (\zeta \min(r_\ast+1, r)/8-k)}$. In the region where $r < r_{\ast} + 1$, we express the solution to $\gamma = |S|\exp(v(r)t - \mu r)$ in terms of the Lambert $W$ function:
\begin{align}
r &= \frac{8}{\alpha \zeta\log(1/\tilde\epsilon)}W\qty(\frac{\alpha \zeta}{8\mu}\log(1/\tilde \epsilon)\tilde\epsilon^{-\alpha (k + \frac{\zeta}{8\mu}\log(|S|/\gamma))}t) + \frac{1}{\mu}\log(\frac{|S|}{\gamma}) \notag \\
&\leq 
\frac{8}{\alpha \zeta\log(1/\tilde\epsilon)}\qty[\log(\frac{\alpha \zeta c \e}{8}) + \alpha'(k + \frac{\zeta}{8\mu}\log(|S|/\gamma))\log(1/\tilde\epsilon) + \log(t)] + \frac{1}{\mu}\log(|S|/\gamma) \notag\\
&= \mathrm O\qty(\frac{\log(t)}{\log(\frac{1}{\tilde \epsilon})} + \log(|S|/\gamma))
\end{align}
where we applied the bound $W(x) \leq \ln(1+x)$ for $x > -\frac{1}{\e}$ \cite{Wfunction}, and chose $\alpha' < \alpha$ and $c$ so that $\log(1/\tilde \epsilon)\tilde\epsilon^{-\alpha (k + \frac{\zeta}{8\mu}\log(|S|/\gamma))} \leq c\tilde\epsilon^{-\alpha' (k + \frac{\zeta}{8\mu}\log(|S|/\gamma))}$.

When $r > r_\ast$, we can solve the equation directly:
\begin{align}
 r = \frac{1}{\mu}\log(|S|/\gamma) + \frac{1}{\mu}v(r_\ast+1)t = \mathrm O\qty(\log(\frac{|S|}{\gamma})+ h\tilde \epsilon^{\alpha (\zeta [r_\ast+1]/8-k)}t)
\end{align}

Since there are $\mathrm O(|S|r^d)$ sites in $S_r$ and the onsite Hilbert space dimension is uniformly bounded, such an evolution can be simulated at cost $\e^{\mathrm O(|S|r^d)}$.
\end{proof}

\begin{cor}
A discrete-time quantum algorithm can compute $\Tr[\rho_0 A(t)B]$ up to error $\gamma$ with $\mathrm O(tr^d\polylog(tr^d/\gamma))$ gates, 
where $r$ is given in \eqref{eq:simulation_r}.
\end{cor}
\begin{proof}
This follows by applying the HHKL algorithm \cite{haah2021quantum} which has complexity $\mathrm O(nt\polylog(nt/\gamma))$, where $n$ is the system size, to $S_r$.
\end{proof}

\section{Non-resonance at all scales}
We have shown that imposing a non-resonance condition on finite regions of radius $r_\ast$ leads to a slow spread of correlations, with a LR velocity suppressed up to order $n=\mathrm{\Theta}(r_*)$. This begs the question of whether it is possible that a \emph{fixed} potential $H_0$ can obey a non-resonance condition at all scales, with $\Delta$ as large as possible at each scale.  This appendix answers this question affirmatively and demonstrates a number of models where such non-resonant conditions are obeyed at all scales. We begin by  exploring the consequences of non-resonance at all scales.

\subsection{Non-perturbatively small Lieb-Robinson velocity}
We begin by showing that if a non-resonance condition holds at every scale, then the Lieb-Robinson velocity is non-perturbatively small in $\epsilon$:

\begin{cor}
\label{cor:gobal_non_resonance}
Let $S_1, S_0 \subseteq \Lambda$ be fixed, and suppose that for $n_\ast = \lfloor (\mathsf d(S_0, S_1) - 1)/(r_\ast+1)\rfloor$ and each $r_{\ast}$ we have a $(h, r_\ast, \Delta, n_{\ast},
 \zeta)$ partial non-resonance condition at every $x \in \partial S_0$ of the form
\begin{equation}
\Delta(r_\ast) \geq h\times a\exp(-cr_\ast^\xi) \label{eq:xi_form}
\end{equation}
with $a,c>0$.  Then for any $\alpha < 1/(5k+1)$, there exist constants $b,C,C'$ and a
constant $b$ for any $c' > c$ such that 
\begin{align} \label{eq:eqcorD1}
\Vert [B_{S_0}, A_{S_1}(t)] \Vert \leq C\min(|\partial S_0|, |\partial S_1|)\Vert B_{S_0} \Vert \Vert A_{S_1} \Vert \qty[\exp(C'h(b\epsilon)^{\frac{\alpha}{2}\qty(\lambda\log^{\frac{1}{\xi}}([b\epsilon]^{-1})-k)}t)-1]\exp(- \mu\mathsf d(S_0, S_1))
\end{align}   
where $\lambda^{-1}= 8(4c')^{1/\xi}/\zeta$ and $\mu = \alpha \zeta/8$.
\end{cor}
\begin{proof}
Since our bounds hold for any $r_\ast$, it suffices to optimize over $r_\ast$.
As we have shown, our bounds are valid for $\tilde \epsilon \leq \e^{-1}$. We
choose constants $c', b$ such that
\begin{align}
\tilde \epsilon \leq a^{-1}\epsilon 2\qty(4c\pi V_{\ast})^2\exp(2cr_\ast^\xi) < \exp(-2c'r_\ast^\xi + \log(\frac{1}{b\epsilon})) \leq \frac{1}{\e}
\end{align}
for any $c' > c$ and some constant $b$ dependent on $c'$ and $\xi$. Inverting this inequality, we see that we may choose
\begin{equation}
r_\ast^\xi = \frac{1}{4c'}\log \frac{1}{b\mathrm{e}\epsilon},
\end{equation}
implying $\tilde \epsilon = \sqrt{b \e\epsilon}$.  \eqref{eq:eqcorD1} follows
from plugging this $r_\ast$ and $\tilde\epsilon$ into Thm.~\ref{thm:main_thm}.
\end{proof}

\subsection{Constraints on non-resonance at all scales}
To explain why we chose the form \eqref{eq:xi_form} for the non-resonance condition, let us first show that this (stretched) exponential scaling is necessary, and that models with this scaling do exist.  In this appendix, we will focus on single-site potentials, where non-resonance conditions are a bit easier to prove: \begin{defn}[Non-interacting $H_0$]
    A classical Hamiltonian $H_0$ is non-interacting if it takes the form \begin{equation}
        H_0 = \sum_{n\in \Lambda}h_n Z_n.
    \end{equation}
\end{defn}

\begin{prop}
Suppose that $\Lambda$ is $d$-dimensional (Def. \ref{defn:d-dim}).  Then there exist non-interacting $H_0$ which are globally non-resonant, obeying \eqref{eq:xi_form} for any $c > M\log(3)+\e^{-1}$, $\xi = d$, and \begin{equation}
    a < \frac{d}{M}\frac{[c-M\log(3)-\e^{-1}]^{1/d}}{\mathrm{\Gamma}(d^{-1})}, \label{eq:aineq}
\end{equation} where $M$ is defined in \eqref{eq:Mdimension}.
\end{prop}

\begin{proof}
Pick an ordering of the sites in $\Lambda$: $x_1,x_2,\ldots$.  We will provide an nonconstructive existence proof by showing that, given a non-resonant potential $\{h_i\}$ on $n$ sites $x_1, \dots, x_n$, we may choose $h_{n+1}$ at $x_{n+1}$ so that the potential still obeys a non-resonance condition at all scales.

Consider a ball of radius $r$ surrounding $x_{i}$ for some fixed $r$ and $x_i$ such that $x_{n+1} \in B_{r_\ast}(x_{i})$. Since $|B_{r}(x_{i})| \leq Mr^d$, the number of possible resonances in this region is smaller than $ 3^{Mr^d}$: the energy level differences of $H_0$ are given by $(\vec \sigma_1 - \vec \sigma_2) \cdot \vec h$, where $\vec \sigma_{1,2} \in \{-1, 1\}^{n}$ and $\vec \sigma_1-\vec \sigma_2\in \lbrace -2,0,2\rbrace^n$.  After the addition of site $x_{n+1}$, each new energy level difference $\Delta E_i^{(n+1)}$ of  $H_0$ restricted to the first $n+1$ sites in $\Lambda$ can be written as $\Delta E_i^{(n+1)} = \Delta E_i^{(n)} \pm 2h_{n+1}$. Therefore, our choices for $h_{n+1}$ are restricted to those such that $|h_{n+1}-\Delta E_i^{(n)}| \geq \epsilon/2$ for some $\epsilon > 0$, which guarantees that $\Delta_{\rm{min}} \geq \epsilon$ within $B_{r}(x_i)$ after the addition of $h_{n+1}$.  Each $x_i \in B_{r}(x_{n+1})$ and $r > 0$ thus eliminates a portion of the interval where we are forbidden to place $h_{n+1}$. Summing over $r$ and $x_i$ and invoking \eqref{eq:xi_form}, the measure of this portion of the interval is less than
\begin{align}
\sum_{r = 1}^{\infty} aMr^d3^{Mr^d}\e^{-cr^d} = 
aM\sum_{r = 1}^{\infty}r^d\exp([M\log(3)-c]r^d) \leq aM\sum_{r = 1}^{\infty}\exp([M\log(3)-c+\e^{-1}]r^d)
\end{align}
Clearly we have chosen $c$ so that this sum will converge.  More concretely, we will bound this sum with an integral using an upper-sum:
\begin{align}
\sum_{r = 1}^{\infty}\exp([M\log(3)-c+\e^{-1}]r^d)  &< \int_0^{\infty}\dd r \ \exp(-r^d [c-M\log(3)-\e^{-1}]) \notag\\
&= \frac{1}{d}\int_0^\infty \dd u \  u^{\frac{1}{d}-1} \e^{-u[c-M\log(3)-\e^{-1}]} \notag \\
& = \frac{\Gamma(d^{-1})}{d (c-M\log(3)-\e^{-1})^{1/d}} < \frac{1}{Ma},
\end{align}
where we used \eqref{eq:aineq}.  Thus we see that
\begin{align}
\sum_{r = 1}^{\infty} aMr^d3^{Mr^d}\e^{-cr^d} < 1
\end{align}
which means that we are able to pick $x_{n+1}$ to satisfy our non-resonance condition.
\end{proof}

One might hope for a better global non-resonance condition, such as $\Delta(r) \sim \frac{1}{r^{\alpha}}$ for some $\alpha$, but this is not possible.
\begin{prop}
Suppose that $\{h_n\}_{n=1}^V$ is a set of numbers $h_n \in [0,1]$ describing an onsite potential $H_0 = \sum_{n=1}^V h_nZ_n$. Then for any $\alpha < 1$ there exists a constant $a$ such that the gap of $H_0$ is bounded above by $\Delta_{\rm{min}} \leq a2^{-\alpha V}$.
\end{prop}

\begin{proof}
The energy levels of $H_0$ are given by $E_{\sigma} = \sum_{n=1}^V h_n \sigma_n$ with $\sigma_n = \pm 1$. Since $h_n \in [0,1]$, we have $|E_{\sigma}| \leq V$ for all $\sigma$. Since there are two choices for each $\sigma$, the number of distinct energy levels is $2^V$. Dividing the interval $[-V, V]$ into $2^V - 1$ sections of size $\frac{2V}{2^V-1}$, by the pigeonhole principle, at least one pair $E_{\sigma}, E_{\sigma'}$ lie in the same section, so $|E_{\sigma} - E_{\sigma'}| \leq 2V(2^V-1)^{-1} \leq a2^{-\alpha V}$ where $a = 4[\e(1-\alpha)\ln(2)]^{-1}$.
\end{proof}

\subsection{Explicit construction of a globally non-resonant potential} 
This section gives an explicit, deterministic construction of a globally non-resonant potential, which we call the ``dyadic-triadic scheme''. We work on boxes in $\mathbb Z^d$ which we may shift without loss of generality to be of the form 
\begin{equation}
\Lambda_L=\{0,\ldots,L-1\}^d\subset \mathbb Z^d. \label{eq:LambdaL}
\end{equation}

We first give the explicit construction formula of the dyadic-triadic scheme and we then explain the underlying rationale afterwards.
Given a collection of non-negative integers $\vec{n}=(n_1,\ldots,n_d)$, we iteratively define 
\begin{subequations}\begin{align}
  h^{(1)}_{\vec{n}} :=& 3^{-\sum_{j=1}^d 2^{j-1}[n_j \text{ (mod 2)}]},
  \label{eq:h1}\\
        h^{(m+1)}_{\vec{n}} :=& 3^{-2^{dm}\left(1+\sum_{j=1}^d2^{j-1}\left[\left\lfloor 2^{-m}n_j\right\rfloor \text{ (mod 2)} \right]\right)} h^{(m)}_{\vec{n}},\qquad m\geq 1\label{eq:hm+1recursion}.\end{align} \end{subequations}
Then we define the on-site potential by summing over all $m$, i.e.,
  \begin{equation}
  \label{eq:hsumhm}
        h_{\vec{n}} := \sum_{m=1}^\infty h^{(m)}_{\vec{n}}. 
    \end{equation}
We remark that the iteration can be solved exactly and leads to a potential of the form $h_{\vec{n}} =\sum_{m\geq 1}3^{-p_m(\vec{n})}$ for suitable powers $p_m(\vec{n})$, cf.\ \eqref{eq:pmdefn} and \eqref{eq:hm_powersof3}.\\

We come to the main result of this subsection which says that \eqref{eq:hsumhm} manages to avoid resonances at all scales, i.e., it is globally resonant for every $r_*>0$.

    \begin{thm}[Globally non-resonant model]\label{prop:dyadic-triadic}
        The non-interacting Hamiltonian $H_0$ on $\Lambda=\mathbb{Z}^d$ with $h_{\vec{n}}$ defined by \eqref{eq:hsumhm}       
        satisfies the $(2,r_*,C_d 3^{-8^dr_*^d})$ non-resonance condition at every $x \in \Lambda$ for every $r_*>0$ and $C_d = 2\times 3^{-2^{d-1}}$.
    \end{thm}

The proof of Theorem \ref{prop:dyadic-triadic} is given below. Before we give it, we  explain the heuristic idea behind the dyadic-triadic scheme. First, since the minimal spectral gap of $H_0$ can be expressed through differences of local $h_{\vec{n}}$ (see \eqref{eq:propreduction} below for the precise statement), the main assertion to be proven is that  sums of the local fields $h_{\vec{n}}$ within balls $B_{r_*}$ cannot cancel exactly and, moreover, are bounded away from zero by $\gtrsim 3^{-r_*^d}$. The dyadic-triadic constructing achieves this non-resonance through an iterative definition, where lack of resonances is ensured in a multiscale way by iteratively ``activating'' higher digits in the base-3 (``triadic'') representation on each dyadic $2^m$-scale. Roughly speaking, the  $\sim {2^{md}}$th triadic digits of the collection of $h_{\vec{n}}$ on the $2^m$-scale form the sequence $(10\ldots0), (010\ldots0), (0010\ldots0),\ldots (0\ldots01)$ and thus avoid cancellation completely at that scale. Higher order digits cannot spoil this cancellation because the construction is such that they can be controlled by a geometric series.  Lower order digits may as well cancel, as we only aim for a gap of size $\sim 3^{-{2^{md}}}\sim 3^{-r_*^d}$ when working at scale $r_*=2^m$.

Let us see explicitly how it works in $d=1$.  In the first $m=1$ step of the iterative definition, we ensure that the smallest dyadic scale $r_*\leq 2^1=2$ (nearest-neighbors) is off resonance, by repeating the block $(*)=(1,3^{-1})$, i.e.,
\begin{equation}
h^{(1)}=(\underbrace{1,3^{-1}}_{(*)},\underbrace{1,3^{-1}}_{(*)},\underbrace{1,3^{-1}}_{\mathrm{etc.}},\ldots,1,3^{-1})
\end{equation}
Of course, this leaves the possibility of perfect cancellation on the next dyadic scale, e.g., between the $1$ in the first and third position of the above vector. This is prevented at the $m=2$ step by setting
\begin{equation}
h^{(2)}=(\underbrace{3^{-2},3^{-3},3^{-4},3^{-5}}_{3^{-2}((*),3^{-2}(*))},\underbrace{3^{-2},3^{-3},3^{-4},3^{-5}}_{\mathrm{etc.}},\ldots,3^{-2},3^{-3},3^{-4},3^{-5})
\end{equation}
where the underbrace displays the emerging iterative structure that eventually takes the general form \eqref{eq:hm+1recursion}. Taking $h^{(1)}+h^{(2)}$, we have achieved the off-resonance property up to the next dyadic scale $r_*\leq 2^{2}=4$, since differences are always at least $3^{-4}$ in size. Notice also that adding $h^{(2)}$ did not spoil the degree of being off-resonance on the previous $m=1$ scale because $3^{-2}+3^{-3}+3^{-4}+3^{-5}<\frac{1}{2}3^{-1}$ and so the minimal gap on the previous scale is only changed mildly from $3^{-1}$ to $\frac{1}{2}3^{-1}$. The fact that higher scales do not spoil prior ones clearly generalizes to higher scales by a geometric series bound. This iterative construction now continues on to higher dyadic scales $2^m$ as expressed by the recursive definition  \eqref{eq:hm+1recursion}. In higher dimensions, the construction is similar, but the basic pattern that is repeated on increasing scales is box-like instead of  interval-like. As the proof of Theorem \ref{prop:dyadic-triadic} confirms, this heuristic indeed rigorously achieves the desired global off-resonance condition in all dimensions.

\begin{rmk}
    Comparing ideas in the dyadic-triadic construction to ones in computer science, one finds that the triadic part ensuring lack of cancellation is related to the fact that the (ordinarily NP-complete) partition problem \cite{garey1990guide} becomes easy when the set to be partitioned is super-increasing, which is the case for the sequence $3^j$. The dyadic part is used for the spatial rearrangement of the triadic powers which is a version of the construction of  the Lebesgue space filling curve (also known as Morton's order in computer science). Roughly speaking, the full argument below combines these two constructions up to errors which can be controlled.
\end{rmk}

\begin{proof}[Proof of Theorem \ref{prop:dyadic-triadic}]
\textit{Step 1. Reductions.}
Fix a dimension $d$, a length $L\geq 1$ and consider the lattice $\Lambda_L$ given in \eqref{eq:LambdaL}.
Recall Definition \ref{defn:noresonance}. It is easy to check  that each $h_{\vec{n}}\leq 2$ by summing a geometric series. The main task in verifying Definition \ref{defn:noresonance} is thus  to bound the  minimal  gap over balls from below. We write $\mathsf{d}(x,y)$ for the graph distance on $\Lambda_L$, as before. Let $r_*>0$, $a\in \Lambda_L$ and recall that $ B_{r_*}(a)$ defines an open ball 
of radius $r_*$ centered at $a$.  
For any $r_*>0$, the relevant minimal gap to bound is \begin{equation}
    \Delta_{r_*}=\inf_{a\in\Lambda_L}\inf_{A\subset B_{r_*}(a)} \Delta_A
\end{equation}where \begin{equation}
    \Delta_A=\inf_{\substack{E,E'\in \mathrm{spec} H_0^A:\\ E\neq E'}} |E-E'|
\end{equation} and \begin{equation}
    H_0^A=\sum_{\vec{n}\in A} h_{\vec{n}} Z_{\vec{n}}
\end{equation}
is the Hamiltonian restricted to terms in subset $A$.  
Since all Pauli-$Z$'s commute, the  eigenvalues of $H_0^A$ are given by 
\begin{equation}
S^A_{\vec{\sigma}}=\sum_{\vec{n}\in A}\sigma_{\vec{n}} h_{\vec{n}}
\end{equation}
with $\vec{\sigma}\in\{\pm 1\}^{|A|}$. Taking differences, we obtain 
\begin{equation}
\Delta_{A}=\inf_{\substack{\vec{\sigma},\vec{\sigma}' \in\{\pm 1\}^{|A|}:\\ \vec{\sigma}\neq \vec{\sigma}'}}\left|S^A_{\vec{\sigma}}-S^A_{\vec{\sigma}'}\right|.
\end{equation}
By linearity, $S^A_{\vec{\sigma}}-S^A_{\vec{\sigma}'}=S^A_{\vec{\sigma}-\vec{\sigma}'}$ and $\vec{\sigma}-\vec{\sigma}'\in\{0,\pm 2\}^{|A|}$, although we note that not all $\vec \sigma - \vec \sigma^\prime$ are in $\lbrace 0,\pm2\rbrace^A$.  Next we take the infimum over $A\subset B_{r_*}$ and notice that we can extend any $\vec{\sigma}-\vec{\sigma}'$ from $A$ to $B_{r_*}$ by setting it equal to zero on $B_{r_*}\setminus A$. Factoring out $2$, we arrive at 
\begin{equation}
\Delta_{r_*}=2\inf_{a\in\Lambda_L}\inf_{\substack{\sigma \in\{0,\pm 1\}^{|B_{r_*}(a)|}:\\ \sigma\neq 0}}\left|S^{B_{r_*}(a)}_{\vec{\sigma}}\right|.
\end{equation}
At this stage, it is easy to verify the claim for $r_*<1$ when $|{B}_{r_*}(a)|=1$, since $h_{\vec{n}}\geq h^{(1)}_{\vec n} \geq 3^{-2^{d-1}}$.
Hence, we assume $r_*\geq 1$ in the following. 
 It is convenient to relax the problem.
For every $r_*\geq 1$ and $a\in \Lambda_L$, there exists $\tilde a\in \mathbb Z^d$ such that $
{B}_{r_*}(a)|_{\Lambda_L}= B_{r_*}(\tilde a)|_{\mathbb{Z}^d}\cap \Lambda_L$. Hence,
\begin{equation}
\Delta_{r_*}\geq 2\inf_{a\in\mathbb Z^d}\inf_{\substack{\sigma \in\{0,\pm 1\}^{| B_{r_*}(a)|}:\\ \sigma\neq 0}}\left|S^{ B_{r_*}(a)}_{\vec{\sigma}}\right|.
\end{equation}
Next, we note that for every ball $ B_{r_*}(\tilde a)\subset \mathbb{Z}^d$ there exists $a'\in \mathbb Z^d$ such that the ball is contained in a shifted square
\begin{equation}
 B_{r_*}(a)\subset a'+\Lambda_{2r_*-1}=a'+\{0,\ldots,2r_*-2\}^d.
\end{equation}
Hence, we may consider the relaxed problem
\begin{equation}\label{eq:propreduction}
\Delta_{r_*}\geq 2\inf_{a\in\mathbb Z^d}\inf_{\substack{\sigma \in\{0,\pm 1\}^{|a+\Lambda_{2r_*-1}|}:\\ \sigma\neq 0}}\left|S^{a+\Lambda_{2r_*-1}}_{\vec{\sigma}}\right|.
\end{equation}

Thanks to \eqref{eq:propreduction}, it suffices to bound $|S^{a+\Lambda_{2r_*-1}}_{\vec{\sigma}}|$ from below for all choices of $a\in\mathbb Z^d$ and of $\vec{\sigma} \in\{0,\pm 1\}^{|{a+\Lambda_{2r_*-1}}|}$ that are not identically zero.

We will achieve this through the following bound on dyadic scales 
\begin{equation}\label{eq:23dyadicclaim}
\inf_{a\in\Lambda_L}\inf_{\substack{\vec{\sigma} \in\{0,\pm 1\}^{|a+\Lambda_{2^M}|}:\\ \vec{\sigma}\neq 0}} |S^{a+\Lambda_{2^M}}_{\vec{\sigma}}|\geq  3^{- 2^{dM+1}},\qquad \forall M\geq 0.
\end{equation}
Assuming \eqref{eq:23dyadicclaim}, we can directly conclude  the claim for any $r_*\geq 1$ as follows. Given any $r_*\geq 1$, we find the unique $M_*\geq 1$ such that $2^{M_*-1}\leq 2r_*-1< 2^{M_*}$  and we apply \eqref{eq:23dyadicclaim} for $M=M_*$ to obtain
\begin{equation}
\inf_{\substack{\vec{\sigma} \in\{0,\pm 1\}^{|a+\Lambda_{2r_*-1}|}:\\ \vec{\sigma}\neq 0}}\left|S^{a+\Lambda_{2r_*-1}}_{\vec{\sigma}}\right| \geq \inf_{\substack{\vec{\sigma} \in\{0,\pm 1\}^{|a+\Lambda_{2^{M_*}}|}:\\ \vec{\sigma}\neq 0}} |S^{\mathcal B_{a+\Lambda_{2^{M_*}}}}_{\vec{\sigma}}|
\geq  3^{-2^{dM_*+1}}
\geq   3^{-8^d r_*^d},
\end{equation}
where the last step used $2^{dM_*+1}\leq 2^{d(M_*+1)}\leq  8^dr_*^d$. This yields the claim of Proposition \ref{prop:dyadic-triadic}. It thus suffices to prove the bound on dyadic scales \eqref{eq:23dyadicclaim}. 

We  fix $M\geq 0$, $a\in \mathbb Z^d$, and $\vec{\sigma} \in\{0,\pm 1\}^{|a+\Lambda_{2^M}|}$ not identically zero. We decompose and estimate
\begin{equation}\label{eq:23decompose}
    \begin{aligned}
\left|S^{ a+\Lambda_{2^M}}_{\vec{\sigma}}\right|
=&\left|\sum_{\vec{n}\in a+\Lambda_{2^M}}\vec{\sigma}_{\vec{n}} \sum_{m\geq 1} h^{(m)}_{\vec{n}}\right|
=\left| \sum_{\vec{n}\in a+\Lambda_{2^M}}\vec{\sigma}_{\vec{n}}\sum_{m=1}^M h^{(m)}_{\vec{n}}
+ \sum_{\vec{n}\in a+\Lambda_{2^M}}\vec{\sigma}_{\vec{n}} \sum_{m=M+1}^\infty h^{(m)}_{\vec{n}}\right|\\
\geq &\underbrace{\left|\sum_{\vec{n}\in a+\Lambda_{2^M}}\vec{\sigma}_{\vec{n}} \sum_{m=1}^M h^{(m)}_{\vec{n}}\right|}_{\mathrm{(I)}}
-\underbrace{ \sum_{\vec{n}\in a+\Lambda_{2^M}} \sum_{m=M+1}^\infty h^{(m)}_{\vec{n}}}_{\mathrm{(II)}}.
\end{aligned}
\end{equation}
Following the heuristic described before the proof, $\mathrm{(I)}$ is the main term to be analyzed in which we harness cancellations at the matching scale $m=M$. By contrast, $\mathrm{(II)}$ is an error term that we control afterwards by a geometric series.

\textit{Step 2: Main term  $\mathrm{(I)}$.} 
We perform another convenient simplification. Notice that each $h^{(m)}$ is periodic in the coordinate directions,
\begin{equation}\label{eq:23periodic}
    h^{(m)}_{\vec{n}}=h^{(m)}_{\vec{n}-2^m e_j},\qquad j=1,\ldots,d.
\end{equation}
Hence, the set of fields entering in term $\mathrm(I)$ is equal to the set of fields appearing on the  box $\Lambda_{2^M}=\{0,\ldots,2^M-1\}^d$ touching the origin, i.e.,
\begin{equation}
\{h^{(m)}_{\vec{n}}\,:\, 1\leq m\leq M,\; \vec{n}\in a+\Lambda_{2^M}\}=\{h^{(m)}_{\vec{n}}\,:\, 1\leq m\leq M,\; \vec{n}\in \Lambda_{2^M}\}.
\end{equation}
Since relabeling only changes the $\vec{\sigma}$ which is arbitrary anyway, it suffices to consider the case $a=0$ i.e.,
\begin{equation}\label{eq:23maintermrewrite}
\inf_{a\in\Lambda_L}\inf_{\substack{\vec{\sigma} \in\{0,\pm 1\}^{|a+\Lambda_{2^M}|}:\\ \vec{\sigma}\neq 0}} \mathrm{(I)}
=\inf_{\substack{\vec{\sigma} \in\{0,\pm 1\}^{|\Lambda_{{2^M}}|}:\\ \vec{\sigma}\neq 0}} \left| \sum_{\vec{n}\in  \Lambda_{2^M}}\vec{\sigma}_{\vec{n}} \sum_{m=1}^Mh^{(m)}_{\vec{n}}\right|.
\end{equation}
In order to bound this term from below, the key idea is to identify the smallest negative power of $3$ that is included in  $\sum_{m=0}^M \sum_{\vec{n}\in  \Lambda_R}\vec{\sigma}_{\vec{n}} h^{(m)}_{\vec{n}}$ and to prove, using elementary modular arithmetic, that it cannot be canceled by any signed combination of the other local fields. 
To this end, we now make explicit how the negative powers of $3$ enter in each $h^{(m)}_{\vec{n}}$'s; see \eqref{eq:hm_powersof3} below. Given $\vec{n}=(n_1,\ldots,n_d)\in\mathbb Z^d$, we define the associated binary digits
\begin{equation}
b_{k}(n_j):=\left\lfloor 2^{-k} n_j\right\rfloor\text{ (mod 2)},\qquad 
k\geq 0,
\end{equation}
so that we have the binary expansion
\begin{equation}
 n_j=\sum_{k=0}^{\infty} b_{k}(n_j) 2^k,
\end{equation}
where the sum is actually finite for each $n_j$.
We define 
\begin{equation}\label{eq:pmdefn}
    p_m(\vec{n})
\;:=\;
\sum_{k=0}^{m-1}
\;2^{dk} \left(1+\sum_{j=1}^d2^{\,j-1}\,b_k(n_j)\right)-1,\qquad (m\geq 1).
\end{equation}
The interpretation of the function $p_M$ is that it associates to each $\vec{n}\in \Lambda_{2^M}$ the highest (i.e., most negative) triadic digit that is present in the fields $\{h_{\vec{n}}^{(m)}\}_{m\leq M}$ that enter into term $\mathrm{(I)}$.

\begin{lem}\label{lem:triadicrep}
 \begin{itemize}
     \item[(i)]
 With the notation above, we have for every $\vec{n}\in\mathbb Z^d$ and every $m\geq 1$ that
\begin{equation}
\label{eq:hm_powersof3}
   h^{(m)}_{ \vec{n}}
=
3^{-\,p_m( \vec{n})}.
\end{equation}
\item[(ii)]
For every $M\geq 1$, the function $p_M:\Lambda_{2^M}\to\mathbb Z_+$ is injective. 
 \end{itemize}  
\end{lem}

\begin{proof}[Proof of Lemma \ref{lem:triadicrep}]
      For property (i), we notice that the iterative definitions
    \eqref{eq:h1} and \eqref{eq:hm+1recursion} can be rewritten as
    \begin{subequations}\begin{align}
  h^{(1)}_{\vec{n}} =& =3^{-\sum_{j=1}^d 2^{j-1}b_0(n_j)}= 3^{-p_1(\vec{n})},
  \label{eq:h1new}\\
        h^{(m+1)}_{\vec{n}} :=& 3^{-2^{dm}} 3^{-2^{dm}\sum_{j=1}^d2^{j-1} b_m(n_j)} h^{(m)}_{\vec{n}},\qquad m\geq 1\label{eq:hm+1recursionnew}.\end{align} \end{subequations}
 Conversely, we observe that
 \begin{equation}
 p_{m+1}(\vec{n})-p_m(\vec{n})=2^{dm}\left(1+\sum_{j=1}^d2^{j-1} b_m(n_j)\right).
 \end{equation}
 Property (i) now follows, because both sides satisfy the same recursion relation.

 For property (ii), we first note that it is equivalent to injectivity of $Q_M:\Lambda_{2^M}\to \mathbb Z_+$ given by
 \begin{equation}
Q_M(\vec{n})= \sum_{k=0}^{M-1}\sum_{j=1}^d2^{dk+j-1}\,b_k(n_j).
 \end{equation}
since $p_M$ and $Q_M$ only differ by a constant.  On $\Lambda_{2^M}$, we can realize $Q_M$ as a composition of the two injective maps $\vec{n}\mapsto
\{b_k(n_j)\}_{\substack{0\leq k\leq M-1\\ 1\leq j\leq d}}
\mapsto  Q_M(\vec n)$.
\end{proof}

With Lemma \ref{lem:triadicrep} at hand, we return to the task of lower-bounding \eqref{eq:23maintermrewrite}. We fix an arbitrary $\vec{\sigma}\in\{0,\pm 1\}^{|\Lambda_{2^M}|}$ with $\vec{\sigma}\neq 0$   and consider
\begin{equation}\label{eq:23afterlemma}
     \sum_{\vec{n}\in  \Lambda_{2^M}}\vec{\sigma}_{\vec{n}} \sum_{m=1}^M h^{(m)}_{\vec{n}}
= \sum_{\vec{n}\in  \Lambda_{2^M}}\vec{\sigma}_{\vec{n}} \sum_{m=1}^M 3^{-\,p_m( \vec{n})}.
\end{equation}
Among the $\vec{n}\in\Lambda_{2^M}$'s which  satisfy $\vec{\sigma}_{\vec{n}}\neq 0$, we consider the one which maximizes $p_M(\vec{n})$ (the highest digit), i.e.,
\begin{equation}
\vec{n}^*
=
\operatorname*{arg\,max}_{ \vec{n}\,:\, \vec{\sigma}_{\vec{n}}\neq 0}\,p_M( \vec{n}),
\end{equation}
and we write \(p^*=p_M( \vec{n}^*)\).  By injectivity of $p_M$, this maximizer is unique. Moreover,  \(\vec{\sigma}_{ \vec{n}^*}=\pm1\). Next, we multiply \eqref{eq:23afterlemma} by \(3^{p^*}\) and reduce modulo $3$.
Notice that $p_m<p_{M}$ for every $1\leq m<M$.  
Moreover, for the summands with \( \vec{n}\neq \vec{n}^*\), we have that \(p_m( \vec{n})< p_M( \vec{n})<p^*\) by maximality of $p_*$. Hence,
  \begin{equation}
    3^{p^*}\sum_{\vec{n}\in  \Lambda_{2^M}}\vec{\sigma}_{\vec{n}} \sum_{m=1}^M 3^{-\,p_m( \vec{n})}
    =3^{p^*}\left( \vec{\sigma}_{ \vec{n}^*}3^{-p^*}+\sum_{k<p^*} c_k 3^{-k}\right)
    =\vec{\sigma}_{ \vec{n}^*}
    \pmod3.
  \end{equation}
  Notice that $|x|\geq |x\pmod 3| $ for any $x\in\mathbb R$.
Recalling \eqref{eq:23maintermrewrite}, \eqref{eq:hm_powersof3}, and $|\vec{\sigma}_{ \vec{n}^*}|=1$,  this implies
\begin{equation}
\inf_{a\in\Lambda_L}\inf_{\substack{\vec{\sigma} \in\{0,\pm 1\}^{|a+\Lambda_{2^M}|}:\\ \vec{\sigma}\neq 0}} \mathrm{(I)}
=\inf_{\substack{\vec{\sigma} \in\{0,\pm 1\}^{|\Lambda_{{2^M}}|}:\\ \vec{\sigma}\neq 0}} \left| \sum_{\vec{n}\in  \Lambda_{2^M}}\vec{\sigma}_{\vec{n}} \sum_{m=1}^Mh^{(m)}_{\vec{n}}\right|\geq 3^{-p^*}.
\end{equation}
From the definition of $p_m$ in \eqref{eq:pmdefn}, we have
\begin{equation}
\label{eq:p*estimate}
p_*\leq \sum_{k=0}^{M-1}
\;2^{dk} \left(1+\sum_{j=1}^d2^{\,j-1}\right)-1
= \frac{2^{d(M+1)}-2^d}{2^d-1}-1
\end{equation}
and so 
\begin{equation}\label{eq:23Ifinal}
    \inf_{a\in\Lambda_L}\inf_{\substack{\vec{\sigma} \in\{0,\pm 1\}^{|a+\Lambda_{2^M}|}:\\ \vec{\sigma}\neq 0}} \mathrm{(I)}\geq  3^{1-\frac{2^{d(M+1)}-2^d}{2^d-1}}.
\end{equation}
This completes Step 2.

\textit{Step 3: Error term $\mathrm{(II)}$.}
We apply Lemma \ref{lem:triadicrep} to write
\begin{equation}
 \mathrm{(II)}=\sum_{\vec{n}\in a+\Lambda_{2^M}} \sum_{m=M+1}^\infty h^{(m)}_{\vec{n}}
 =\sum_{\vec{n}\in a+\Lambda_{2^M}} \sum_{m=M+1}^\infty 3^{-p_m(\vec{n})}.
\end{equation}
We may again reduce to the case $a=0$ with loss of generality. Indeed, recall the periodicity \eqref{eq:23periodic} and notice that within the period, the coordinate-wise maps $k \mapsto h^{(m)}_{\vec{n}_ke_j}$ defined on $\{0,\ldots,2^m-1\}$ are monotonically increasing for every $j=1,\ldots,d$. Together, these facts imply that we may estimate term $\mathrm{(II)}$ by the case $a=0$, i.e.,
\begin{equation}\label{eq:IIestimatea=0}
     \mathrm{(II)}\geq \sum_{\vec{n}\in \Lambda_{2^M}} \sum_{m=M+1}^\infty 3^{-p_m(\vec{n})}.
\end{equation}

We note the following  generalization of the injectivity property from Lemma \ref{lem:triadicrep}.

\begin{lem}\label{lem:23injective}
   Fix $M\geq 1$. Consider the family of functions $p_{m}:\Lambda_{2^M}\to\mathbb Z_+$ indexed by $m\geq M+1$. This is an injective family in the sense that
   \[
   p_{m}(\vec{n})=p_{m'}(\vec{n}') \;\Rightarrow\; m=m'\text{ and } \vec{n}=\vec{n}'.
   \]
\end{lem}

\begin{proof}
   Starting from \eqref{eq:pmdefn}, elementary estimates show that   $\min p_{m+1}-\max p_{m}\geq 1$. This implies $m=m'$. For a fixed $m\geq M+1$, the injectivity of $p_{m}:\Lambda_{2^M}\to\mathbb Z_+$ follows by the same argument as the injectivity of $p_{M}:\Lambda_{2^M}\to\mathbb Z_+$ shown in Lemma \ref{lem:triadicrep}. 
\end{proof}

By Lemma \ref{lem:23injective}, the sequence of powers $p_m(\vec{n})$ appearing in \eqref{eq:IIestimatea=0} is a disjoint subset of $\mathbb Z_+$.
Moreover, we have for  every $m\geq M+1$ and $\vec{n}\in\Lambda_{2^M}$,
  \begin{equation}
   p_m(\vec{n})\geq \sum_{k=0}^{m-1}
\;2^{dk}-1
\geq \frac{2^{dm}-1}{2^d-1}-1
\geq  \frac{2^{d(M+1)}-2^d}{2^d-1}.
   \end{equation}

Hence, we can compare to a geometric series: 
\begin{equation}\label{eq:23IIfinal}
   \mathrm{(II)}\leq \sum_{\vec{n}\in \Lambda_{2^M}} \sum_{m=M+1}^\infty 3^{-p_m(\vec{n})}
 \leq \sum_{j= \frac{2^{d(M+1)}-2^d}{2^d-1}}^\infty 3^{-j}
 =\frac{3}{2} 3^{-\frac{2^{d(M+1)}-2^d}{2^d-1}}. 
\end{equation}
This completes Step 3.

We are now ready to conclude \eqref{eq:23dyadicclaim}. Indeed, combining \eqref{eq:23decompose}, \eqref{eq:23Ifinal}, and \eqref{eq:23IIfinal}, we obtain 
\begin{equation}
\left|S^{ a+\Lambda_{2^M}}_{\vec{\sigma}}\right|
\geq \mathrm{(I)}-\mathrm{(II)}\geq 3^{1-\frac{2^{d(M+1)}-2^d}{2^d-1}}-\frac{3}{2} 3^{-\frac{2^{d(M+1)}-2^d}{2^d-1}}=
\frac{3}{2} 3^{-\frac{2^{d(M+1)}-2^d}{2^d-1}}
\geq 
 3^{- \frac{2^{d(M+1)}}{2^d-1}}
 \geq 3^{-2^{dM+1}}.
\end{equation}
where the last step used $\frac{3^{\frac{2^d}{2^d-1}}}{2}\geq 1$ and $2^d-1\geq 2^{d-1}$. This establishes \eqref{eq:23dyadicclaim}.
\end{proof}

\subsection{Disordered potentials}
The dyadic-triadic scheme has global non-resonance at all scales, but is finely tuned.  We now show that disordered models also obey the non-resonance condition with high probability.  For simplicity we focus on Gaussian disorder below, which simplifies the discussion.   Similar facts are very important in previous efforts to establish slow dynamics and/or MBL in disordered one-dimensional spin chains \cite{imbrie2016many,absenceofconduction}.

\begin{prop}
\label{prop:probabilistic_noresonance}
Let $h_i$ be i.i.d. Gaussian zero-mean random variables with  variance $\sigma$, and let $V_\ast = \sup_{x \in \Lambda}|B_{r_\ast}(x)|$ (as before). Let $x \in \Lambda$ be fixed.
Then for any $p < 1$,$\delta > 0$, and $\delta' > \delta$ and any $r_\ast$ there exists a constant $c$ such that $\{h_i\}$ satisfies the $(h, r_\ast, \Delta)$-no resonance condition at $x$ with probability at least $1-p_{\mathrm{viol}}$, where: \begin{subequations}
    \begin{align}
    h &= \sigma\sqrt{2\log(4V_\ast 3^{\delta V_{\ast}}/p)}, \\
        \Delta &= p^{1+\delta}c3^{-(1+\delta')V_\ast}, \\
       p_{\mathrm{viol}} &\le  p3^{-\delta V_\ast}. \label{eq:pviol}
    \end{align}
\end{subequations}
\end{prop}

The need to take $h = \mathrm O\qty(\sqrt{V_\ast})$ is a consequence of controlling the maximum value of all the random Gaussian potentials within the ball. However, the velocity $v$ scales with (roughly) $v \sim h\epsilon^{\lambda r_\ast}$, so we will absorb the factor of $\sqrt{V_\ast}$ into a reduction in the exponent $\lambda$.
\begin{proof}
Consider $\{h_i\}_{i \in B_{r_\ast}(x)}$, and let $V \equiv |B_{r_\ast}(x)|$. The resonances are $(\vec \sigma - \vec \sigma') \cdot \vec h$ where $\vec \sigma, \vec \sigma' \in \{-1, 1\}^{V}$, and so $\vec \sigma - \vec \sigma' \in \{-2, 0, 2\}^{V}$. Note that this includes any subset $A \subseteq B_{r_\ast}(x)$ by simply setting $\sigma_i = \sigma_i' = 0$ for $i \in B_{r_\ast}(x) \backslash A$. Let $S_{\vec\sigma} \equiv \vec \sigma \cdot \vec h$, for $\vec \sigma \in \{-2, 0, 2\}^{V}$. Then the minimum resonance is
\begin{align}
\Delta_{\rm{min}} = \min_{\vec \sigma  \in \{-2,0,2\}^{|A|},\;  \vec\sigma \ne \vec 0}\qty|S_{\vec \sigma}^A|
\end{align}
Let $n(\vec \sigma)$ denote the number of nonzero entries in $\vec \sigma$. We then bound the minimum with a union-bound argument:
\begin{align}\label{eq:anti-concentration}
\mathbb P\qty[\min_{\vec \sigma \in \{-2, 0,2\}^{V}}\qty|S_{\vec \sigma}| < \Delta] &\leq \sum_{\vec \sigma \in \{-2, 0, 2\}}\mathbb P\qty[\qty|S_{\vec \sigma}| < \Delta] \notag\\
&= \sum_{\vec \sigma \in \{-2, 0, 2\}^{V}}\int\limits_{-\Delta}^{\Delta}\frac{\mathrm{d}x}{\sqrt{8\pi \sigma^2 n(\vec \sigma)}} \e^{-x^2/(8\sigma^2n(\vec \sigma))} &\leq
\sum_{\vec \sigma \in \{-2, 0, 2\}^{V}}\frac{2\Delta}{\sqrt{8\pi \sigma^2 n(\vec \sigma)}} 
\leq \frac{3^{V}\Delta}{\sigma}
\end{align}
Since $V_\ast \geq V$, taking $\Delta = \frac{p\sigma}{2}3^{-(1+\delta)V_\ast}$ for some $p < 1$ and any $\delta > 0$, the potential constructed is non-resonant with probability less than $\frac{p}{2}3^{-\delta V_\ast}$. However, we may not indeed have $|h_i| \leq 1$ for all $i \in S$. By a Chernoff bound for $h_i$, 
\begin{align}
p' = \mathbb P[\max |h_i| > h] \leq V\mathbb P[|h_i| > h] \leq 2V \exp(-\frac{h^2}{2\sigma^2})
\end{align}
for any $h$. Solving for $h$, we find that with probability less than $p'$ we have $\max |h_i| > h$ where $h = \sigma\sqrt{2\log(2V_\ast/p')}$. Setting $p' = \frac{p}{2}3^{-\delta V_{\ast}}$, we have $\max |h_{i}|\leq \sigma \sqrt{2\log(4V_\ast 3^{\delta V_{\ast}}/p)}$ \textit{and} non-resonance with 
\begin{equation}
\Delta/h \geq \frac{p}{2}3^{-(1+\delta)V_\ast}(2\log(4V_\ast 3^{\delta V_{\ast}}))^{-1/2} \geq p^{1+\delta}c3^{-(1+\delta')V_{\ast}}
\end{equation}
for an appropriate constant $c$ depending on $\delta, \delta'$, with probability $p_{\rm{viol}} \leq p+p' = p3^{-\delta V_\ast}$.
\end{proof}

While the union bound might appear to sacrifice too much, numerics indicate that the estimated scaling of the smallest resonance is qualitatively realistic with a slightly sub-optimal exponent. This is shown in Fig.~\ref{fig:gaussian-numerics}. Now we can crucially apply the above estimate to bound the probability for Gaussian disorder to fail the $(h, \Delta, r_\ast, n_\ast, \zeta)$-no resonance condition at a point $x$. Since a union-bound argument allows us to neglect correlations between different regions in each path, the problem is reduced to the combinatorics of estimating the failure probability within a single path.

\begin{figure}[t]
    \centering
    \def\svgwidth{0.45\linewidth} 
    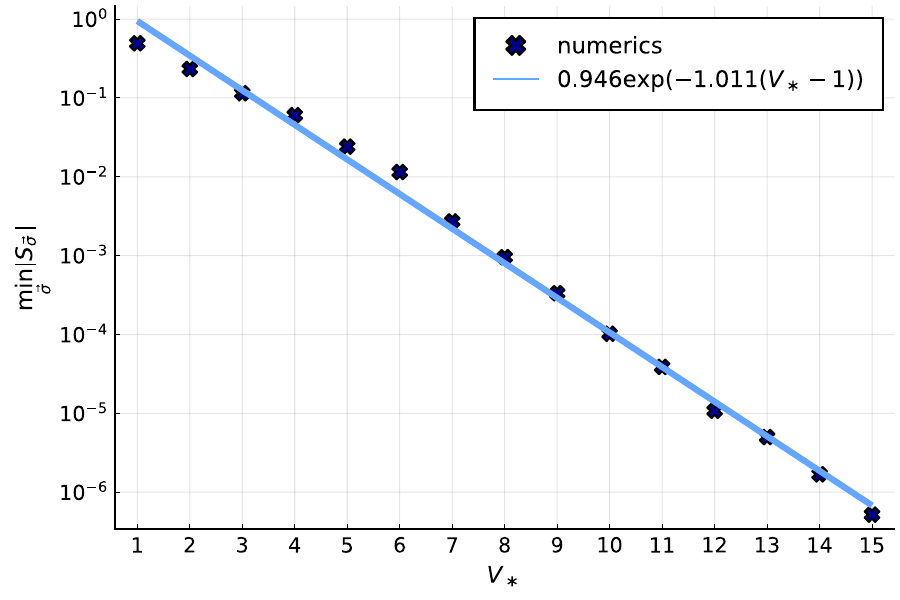
    \def\svgwidth{0.45\linewidth} 
    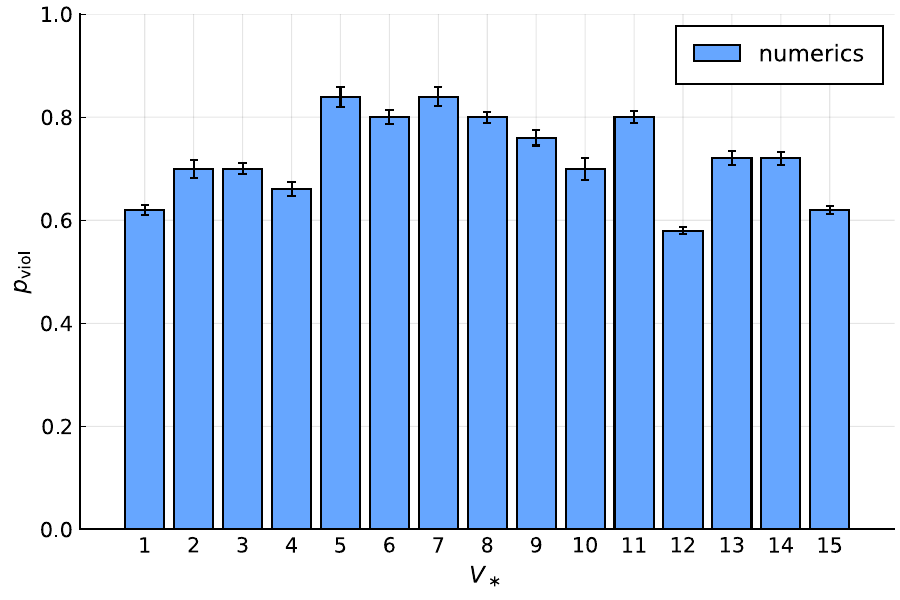

    \caption{\textbf{Left:} An average of the minimum gap over disorder of
      strength $\sigma = 1$. The average-case non-resonance condition appears to
      be $\Delta \sim \e^{-(V_\ast-1)}$ for random Gaussian potentials. The data
      in the plot represent a single region of size $V_\ast$, and the minimum
      resonance is averaged over 50 realizations of disorder. This suggests the
      union bound is actually close to optimal. \textbf{Right:} Probability of
      failing the non-resonance condition $\Delta_{\rm{min}} \geq
      \frac{1}{4}\exp(-(V_\ast - 1))$. The probability remains finite as $V_\ast$ is
      increased. Error bars are estimated using resampling.}
    \label{fig:gaussian-numerics}
\end{figure}

\begin{cor}
\label{lem:no_percolating_paths}
Suppose that on a $d-$dimensional graph $\Gamma$, an on-site potential $H_0$ fails the $(h, r_\ast, \Delta)$ non-resonance condition at arbitrary $x^\prime \in \Lambda$ with probability no larger than $p_{\rm{viol}}$. 
Then for any $\zeta < 1$
the $(h, \Delta, r_\ast, n_\ast, \zeta)$  partial non-resonance condition (Def.~\ref{def:partialnonresonance}) is failed at any (other) point $x\in \Lambda$ with probability
\begin{align}
\mathbb P[\neg \operatorname{NR}(x)] \leq 2p_{\rm{viol}}^{1-\zeta}(2V_\ast p_{\rm{viol}}^{1-\zeta})^{n_\ast-1}
\end{align}
\end{cor}

Note that for Gaussian disorder,  we have $p_{\rm{viol}} \to 0$ with $\epsilon$ much faster than $V_\ast \sim \log^d(\epsilon^{-1}) \to \infty$, so the above bound improves as either $\epsilon \to 0$ or $\mathsf d(S_0, S_1) \to \infty$. 

\begin{proof}
Consider the set of tuples $(x_1, \dots, x_{n_\ast})$ with $x_1 = x$ and $\mathsf d(x_i, x_{i+1}) \leq r_\ast + 1$. Each path $\Gamma \in \Gamma^{(n_\ast)}(x)$ corresponds to a choice of $\vec \ell = (y_1, \dots, y_k)$ as described in Def.~\ref{def:generate_paths} which indicate where a SW is applied if $B_{r_\ast}(x_i)$ is non-resonant. By construction, $k \leq n_\ast$.
$\Gamma$ will be associated to $\vec x = (x_1, \dots, x_{n_\ast})$ if there is an ordering of $\vec \ell$ for which it appears as a subsequence of $\vec x$. Clearly, if non-resonance holds in least $\zeta n_\ast$ choices of $x_i \in x$, then each $\Gamma$ associated with $\vec x$ is non-resonant. Since the regions $\Sint{i}(\Gamma)\backslash \Sint{i-1}(\Gamma)$ for a fixed choice of $\Gamma$ and each $i$ within which we construct the SW transformations are disjoint, we can treat these probabilities as independent:
\begin{align}
\mathbb P[\neg\operatorname{NR}(\Gamma)] \leq \sum_{q = 0}^{(1-\zeta)n_\ast}\binom{n_\ast}{q}p_{\rm{viol}}^{q}(1-p_{\rm{viol}})^{1-q} \leq 2^{n_\ast}p_{\rm{viol}}^{(1-\zeta)n_\ast}
\end{align}
Applying the union bound over the possible choices of $(x_1, \dots, x_{n_\ast})$, of which there are less than $V_\ast^{n_\ast-1}$, the probability of failing the $(h, \Delta, r_\ast, n_\ast, \zeta)$-non-resonance condition is bounded by
\begin{align}
\mathbb P[\neg \operatorname{NR}(x)] \leq  2p_{\rm{viol}}^{1-\zeta}(2V_\ast p_{\rm{viol}}^{1-\zeta})^{n_\ast-1}.
\end{align}
We remark that the correlations between the failure of the non-resonance condition between different choices of $(x_1,\ldots , x_{n_*})$ does not affect the simple union bound argument above.
\end{proof}

As an immediate corollary, we obtain a probabilistic constraint on when we can apply the non-perturbative Lieb-Robinson bound in Cor.~\ref{cor:gobal_non_resonance}. We will see that this constraint is strong in that the probability becomes unity as $\mathsf d(S_0, S_1)$ becomes large, showing that Gaussian disorder generically produces this behavior up to fluctuations at small distances.

\begin{cor}
Let $S_0,S_1\subset\Lambda$ and \begin{equation}
    H_0 = \sum_{x_n \in \Lambda}h_nZ_n
\end{equation} 
where each $h_i$ is an i.i.d. zero-mean Gaussian variable with variance $\sigma$. Then for any $\delta > 0$, $\zeta < 1$, and $\delta' < \delta(1-\zeta)$, there exist constants $a, b, C, C'$ such that
\begin{align}
\label{eq:prob_nores_viol}
\mathbb P\qty[\Vert [A_{S_0}(t), B_{S_1}]\Vert \geq C\qty[\exp(C'\sigma h\qty[\frac{b\epsilon}{\sigma}]^{\lambda r}t)-1]\exp(-\mu \mathsf d(S_0, S_1))] \leq a\min(|\partial S_0|, |\partial S_1|)3^{-2\delta'\mathsf d(S_0, S_1)/3}
\end{align}
for all $b\epsilon \leq 1$, where $\lambda < \zeta [2^{2+3d}(1+\delta)\log 3]^{1/d}$  is arbitrary and $r = \max(r_\ast+1, \mathsf d(S_0, S_1))$.
\end{cor}

\begin{proof}
Applying the estimates from Prop.~\ref{prop:probabilistic_noresonance} and 
Cor.~\ref{lem:no_percolating_paths}, along with \eqref{eq:pviol}, the non-resonance condition is failed at $x \in \Lambda$ with probability less than
\begin{align}
(2V_\ast p_{\rm{viol}}^{1-\zeta})^{n_\ast} \leq(2pV_\ast 3^{-\delta V_\ast})^{n_\ast} \leq 3^{-2\delta'r_\ast n_\ast}
\end{align}
where for any $\delta'' < (1-\zeta)\delta$, we chose $p^{1-\zeta} = [2([1-\zeta]\delta - \delta'')\log(3)]^{-1}3^{-1/\log(3)}$ so that $2pV_\ast e^{-\delta V_\ast} \leq 3^{-\delta'' V_{\ast}}$, and for convenience, we used $V_\ast > 2r_\ast$.  Next we use that $r_\ast n_\ast \geq \mathsf d(S_0, S_1)/3$ if $\mathsf d(S_0, S_1) \leq r_\ast$ and $r_\ast n_\ast = r_\ast$ if $\mathsf d(S_0, S_1) > r_\ast$. Summing over these possibilities, we find
\begin{align}
\sum_{r_\ast \leq \mathsf d(S_0, S_1)} 3^{-2\delta''\mathsf d(S_0, S_1)/3} + \sum_{r_\ast > \mathsf d(S_0, S_1)} 3^{-2\delta'' r_{\ast}}
&= \mathsf d(S_0, S_1)3^{-2\delta''\mathsf d(S_0, S_1)/3} + \frac{3^{-\delta''[\mathsf d(S_0, S_1)+1]}}{1-3^{-2\delta'}} \notag \\
&\leq a3^{-\delta'\mathsf d(S_0, S_1)}
\end{align}
where $\delta' < \delta''$ is arbitrary and the constant $a$ is chosen for the above to be true: \begin{equation}
    a > \left(\frac{2}{3}(\delta''-\delta')\right)^{-2} + \frac{3^{-\delta''}}{1-3^{-2\delta''}}.
\end{equation}

Then we can apply Cor.~\ref{cor:gobal_non_resonance} by recognizing that $\Vert V_{\epsilon = 1} \Vert_{\kappa = 0} \leq h$ and $\Vert H_0 \Vert_{\kappa = 0} \leq \sigma h\mathrm O(\sqrt{V_\ast})$. We can choose $C'$ for any $\lambda < \zeta[2^{2+3d}(1+\delta)\log(3)]^{1/d}$ to absorb the factor of $\sqrt{V_{\ast}}$. Next, assuming that $|\partial S_0| \leq |\partial S_1|$ without loss of generality, we union bound the probability of failing the partial non-resonance condition at each $x \in \partial S_0$, arriving at the prefactor of $\min(|\partial S_0|, |\partial S_1|)$. This proves the claim.
\end{proof}

\subsection{Incommensurate lattices}
We now prove that generic incommensurate lattices, which have been studied previously in the context of MBL \cite{sarang_quasiperiodic,mastropietro2015localization,huse_quasiperiodic},  also obey a non-resonance condition.  We note that previous authors sometimes used multiple incommensurate potentials, but only one is needed to obtain a non-resonance condition.

\begin{thm}
\label{prop:incommensurate_potentials}
Let $\Lambda \subset \mathbb{Z}^d$ be a finite subset.   Consider the non-interacting potential
\begin{align} \label{eq:incommensuratelattice}
H_0 = \sum_{\vec k \in \Lambda}\cos \left(2\pi \vec \alpha \cdot \vec k\right)Z_{\vec k} \ .
\end{align}
For any $r > 0$, there is a set $S \subseteq [0,1]^d$ with Lebesgue measure $\mu(S) \geq 1-\frac{1}{2}3^{-2r}$ such that for any $\vec \alpha \in S$ and $\delta > 0$ there is a constant $C$ for which $H_0$ satisfies a $(r_\ast, \Delta, n_\ast, \frac{1}{4})$ partial non-resonance condition for every $r_\ast \geq r$ and $n_\ast \geq 4\times 3^{V_\ast}$ with
\begin{align}
\Delta \geq C3^{-(2+\delta)Mr_{\ast}^{d+1}} \ .
\end{align}
where $M = \min[3\times 2^{d-1}, (d+1)(d+2)]$.
\end{thm}

\begin{proof}
Our proof strategy will be to show the claim at $x = \vec 0$, and then to use the discrete translation of the incommensurate lattices to show that resonances imply a rational approximation to $\vec \alpha$, which will only occur infrequently for most $\vec \alpha$.

First, consider $r_\ast$ to be fixed. The resonances within a ball $B_{r_\ast}$ of volume $V_\ast$ arise when the real part of
\begin{align}
z_{\vec \sigma} = \sum_{\vec k \in B_{r_*}}\sigma_{\vec{k}} \e^{2\pi \ii \vec \alpha \vdot \vec k}
\end{align}
for a suitable choice of $\vec{\sigma}_{\vec{k}}\in\{0,\pm 1\}^{V_\ast} $ is sufficiently small. Notice that 
where $z_{\vec \sigma}$ is a multivariate polynomial in the variables $\xi_1 =
\e^{2\pi \ii \alpha_1}, \dots, \xi_d = \e^{2\pi \ii \alpha_d}$ of degree less than or equal to $n_1,
\dots, n_d$, with $n_1 + \dots + n_d \leq r_\ast$, precisely because this is the
Manhattan distance on the square lattice. Then by the fundamental theorem of algebra,  
\begin{align}
z_{\vec\sigma}(\xi_1, \dots, \xi_d) = p(\xi_2, \dots, \xi_d)\prod_{k=1}^{n_1}(\xi_1 - f_k^{(1)}(\xi_2, \dots, \xi_d))
\end{align}
such that $p(\xi_2, \dots, \xi_d)$ is another multivariate polynomial of degrees less than $n_2, \dots, n_d$ and $f_k^{(1)}$ is a (likely non-analytic) function. Since $p$ is the coefficient of $\xi_1^{n_1}$ (with $n_1$ the largest power of $\xi_1$) in $z_{\vec\sigma}$, $p$ has integer coefficients as well.
Carrying on in this manner, we arrive at the factorized form
\begin{equation}
z_{\vec \sigma}(\xi_1, \dots, \xi_d) = N\prod_{j=1}^{d}\prod_{k=1}^{n_j}(\xi_j - f_k^{(j)}(\xi_{j+1}, \dots, \xi_d))
\end{equation}
where $N$ is an integer.
We fix a parameter $\epsilon >0$ to be chosen later. 
We may then choose $\xi_d$ such that $|\xi_d - f_k^{(d)}| \geq 3^{-V_\ast}\epsilon/r_\ast^3$ for each $k$, and then iteratively choose $\xi_j$ for $j < d$ such that $|\xi_j - f_j^{(k)}(\xi_{j+1}, \dots, \xi_d)| \geq 3^{-V_\ast}\epsilon/r_\ast^3$. In this way, we find that there is a set
\begin{equation}
    A_{\sigma,r_*}:=\left\lbrace \vec \xi \in \mathrm{U}(1)^d :  \min_k |\xi_1 - f_k^{(1)}(\xi_2,\ldots \xi_d)| \ge \frac{3^{-V_*}\epsilon}{r_*^3}, \min_k |\xi_2 - f_k^{(2)}(\xi_3,\ldots \xi_d)| \ge \frac{3^{-V_*}\epsilon}{r_*^3}, \ldots \right\rbrace 
\end{equation}
where $\mathrm{U}(1)$ denotes the unit circle in the complex plane, and for every $\vec \xi \notin A_{\sigma, r_\ast}$,
\begin{align}
|z_{\vec\sigma}| \geq \left(\frac{3^{-V_*}\epsilon}{r_*^3}\right)^{n_1}\left(\frac{3^{-V_*}\epsilon}{r_*^3}\right)^{n_2}\dots\left(\frac{3^{-V_*}\epsilon}{r_*^3}\right)^{n_d} \geq \left(\frac{3^{-V_*}\epsilon}{r_*^3}\right)^{r_*}
\end{align}
with
\begin{equation}
    \mu(A_{\sigma, r_\ast}) \leq
\qty(\frac{2\times 3^{-V_\ast} \epsilon}{r_\ast^2})^d \leq \frac{2 \times 3^{-V_\ast} \epsilon}{r_\ast^2}.
\end{equation}

Since there are $V_\ast$ possible monomials corresponding to each of the
sites in $B_{r_\ast}$, there are $3^{V_\ast}-1$ nonzero choices for $\vec\sigma$, so summing over these
eliminates a set of measure \begin{equation}\mu(A_{r_\ast}) = \mu\qty(\bigcup_{\vec \sigma \in \{-2, 0, 2\}^{V_\ast}}A_{\sigma, r_\ast}) \leq 2\sum_{\vec \sigma \in \{-2, 0, 2\}, \vec \sigma \ne 0}\frac{3^{-V_\ast}\epsilon}{r_\ast^2} < \frac{2\epsilon}{r_\ast^2} \ .\end{equation}
Summing over each $r_\ast$, we then find
\begin{align}
\mu(A) = \mu\qty(\bigcup_{r_\ast \geq 1}A_{r_\ast}) \leq \frac{\pi^2 \epsilon}{3}
\end{align}
Then we take $\epsilon = \frac{3}{2\pi^2}3^{-2r}$ to get $\mu(A) \leq \frac{1}{2}3^{-2r}$, and put $S = [0,1]^d \backslash A$. 

Second, translating by $\tau$ multiplies the complex potential, i.e. $h_{\vec n + \vec \tau} = \e^{2\pi \ii \vec \alpha \cdot \vec \tau}h_{\vec n}$.  The resonances in a ball of radius $r_\ast$ around $\tau$ are given by $\Re(\e^{2\pi \ii \vec \alpha \cdot \vec \tau}z_{\vec \sigma})$. 
We then estimate
\begin{align}
\qty|\Re(\e^{2\pi \ii \vec \alpha \cdot \vec \tau}z_{\vec \sigma})| \geq \frac{\pi}{2}\qty|2(\theta_{\vec\sigma} + \vec \alpha \cdot \vec \tau )\bmod 1 - \frac{1}{2}||z_{\vec \sigma}| \label{eq:cos2linear}
\end{align}
where $\theta_{\vec \sigma} = \arg(z_{\vec \sigma})/2\pi$. Suppose that
\begin{align}
\qty|2(\theta_{\vec \sigma} + \vec \alpha \cdot \vec \tau )\bmod 1 - \frac{1}{2}| \leq \frac{\gamma}{2}
\label{eq:adjacent_resonances}
\end{align}
for $\vec \tau = \vec \tau_0$ and $\vec \tau = \vec \tau_0 + \vec\tau'$, where $\gamma$ is a parameter which will be taken small. This implies that
\begin{align}
\Vert 2\vec \alpha \cdot \vec \tau' \Vert \leq \gamma \ .
\end{align}
We now invoke the following small-denominator result, which is used in the proof of the KAM Theorem \cite{kolmogorov2005preservation}:
\begin{lem}
\label{lem:rational_approximation}
  For Lebesgue-a.e. $\vec \alpha \in [0,1]^d$ and $\delta > 1$, there exists a constant
  $C_{\alpha, \delta}$ such that
  \begin{align}
    |\vec q \cdot \vec \alpha - p| > \frac{C_{\alpha,\delta}}{\Vert \vec q \Vert_1^{\delta d}}.
  \end{align}
\end{lem}
By Lem.~\ref{lem:rational_approximation}, we have $\Vert 2\vec \alpha \cdot \vec \tau'\Vert \geq C_{\alpha, \delta}\Vert \vec \tau'\Vert_1^{-d\delta}$ for Lebesgue-almost-all $\alpha$, so we remove a measure-zero set from $S$. We choose $\gamma = C_{\alpha, \delta}(2r_\ast)^{-d\delta}3^{-d\delta V_\ast}$ so that $\vec \tau '$ must satisfy $\Vert \vec \tau' \Vert \geq 2r_\ast 3^{V_\ast}$. Plugging in the resulting bound to \eqref{eq:cos2linear}, we set
\begin{align}
\Delta = \frac{\pi}{4}C_{\alpha, \delta}(2r_\ast 3^{V_\ast})^{-d\delta}\qty(\frac{3}{2\pi^2} \times \frac{3^{-V_\ast-2r}}{r^3})^{r_\ast}
\end{align}

As in the last section, we identify $\Gamma \in \Gamma^{(n_\ast)}(x)$ with a tuple $x_1, \dots, x_{n_\ast}$ where $\mathsf d(x_i, x_{i+1}) = r_\ast+1$, and if non-resonance fails in $B_{r_\ast}(x_i)$ for at most $3n_\ast/4$ choices of $i$, then this establishes a $\zeta = \frac{1}{4}$ non-resonance condition.
Dividing the lattice into boxes of side length $r_\ast \times 3^{V_\ast}$, for $n_\ast \geq 4\times 3^{V_\ast}$, the path $x_1, \dots, x_{n_\ast}$ lies within a connected set of at least $4$ boxes. Of the $n \leq \lceil n_\ast 3^{-V_\ast}\rceil$ boxes in this set, any subset consisting entirely of boxes which are non-adjacent contains at most $3n/5$ boxes. By our choice of $\gamma$ and Lem.~\ref{lem:rational_approximation}, we see that $|\Re(\e^{2\pi \ii \vec \alpha \cdot \vec \tau}z_{\vec\sigma})| < \Delta$ may not hold for a single choice of $\vec \sigma$ within two adjacent boxes, because then $\Vert \vec \tau'\Vert_1 \leq 2r_\ast 3^{V_\ast}$. This means that each choice of $\vec \sigma$ can lead to at most $\frac{3}{5}\lceil n_\ast 3^{-V_\ast}\rceil$ resonances. Since there are $3^{V_\ast}$ choices of $\vec \sigma$, the maximal number of choices of $i$ in which the non-resonance condition is violated is
\begin{align}
\frac{3}{5}\left\lceil \frac{n_\ast}{3^{V_\ast}}\right\rceil 3^{V_\ast} \leq \frac{3n_\ast}{4}
\end{align}
This shows that the $\zeta = \frac{1}{4}$ non-resonance condition is obeyed with
\begin{equation}
\Delta > C3^{-(2+\delta')Mr_\ast^{d+1}}
\end{equation}
for any $\delta' > 0$, where $m = \min((d+1)(d+2), 3\times 2^{d-1})$ and we used Proposition \ref{prop:Msquarelattice} to relate $V_*$ to $r_*$ for the lattice $\mathbb{Z}^d$ or any subset thereof.
\end{proof}

\begin{rmk}Numerical evidence suggests that the additional power of $r$ in the exponent and the requirement that $|x-y| \gtrsim \frac{1}{\epsilon^d}$ is sub-optimal. By the Lindeman-Weierstrass theorem, $\e^{2\pi \ii \alpha}$ is transcendental whenever $\pi\alpha$ is algebraic. We pick $\alpha = 1/\pi$, and find that the potential satisfies $\Delta(r) \geq \frac{1}{4}\e^{-2.5r}$ on at least one of every two adjacent segments of length $r$ for each $r \leq 9$: see Figure \ref{fig:irrationalnumerics}.
\end{rmk}
\begin{figure}[t]
    \centering
    \def\svgwidth{0.45\linewidth} 
    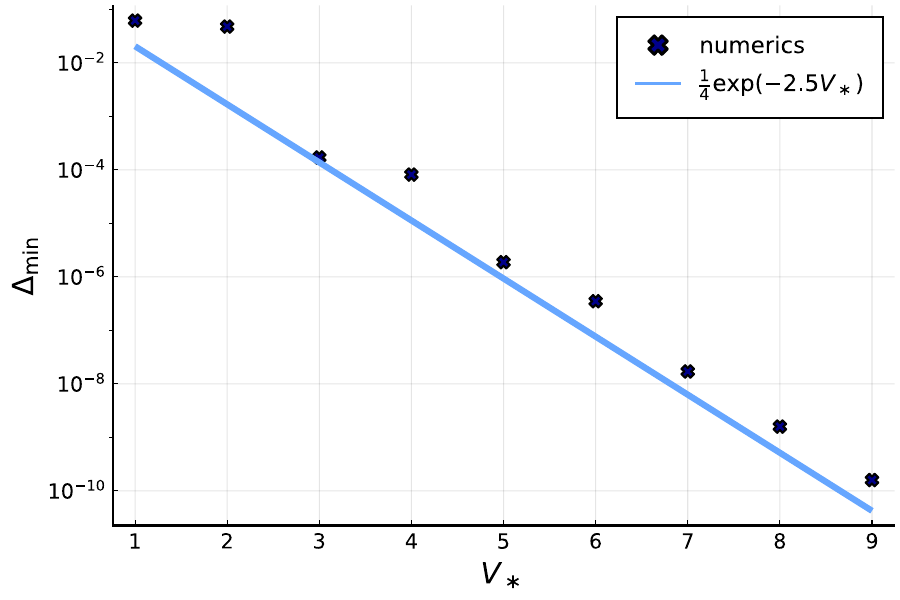
    \caption{Figure showing a numerical non-resonance condition established for a 1-dimensional quasiperiodic potential. On the y-axis is plotted $\min_{\vec\sigma}|z_{\vec\sigma}| \times \frac{\pi}{4}\min_{\vec\sigma, \vec\sigma'}\Vert 2(\theta_{\vec\sigma} - \theta_{\vec\sigma'}-r_\ast \alpha)\Vert$. Following the argument above, this establishes a non-resonance condition of the form $\Delta_{\rm{min}} > \frac{1}{4}\e^{-2.5 V_{\ast}}$ for at least $\frac{1}{3}$ of the regions of radius $r_\ast$ separating points $x$ and $y$ in $d=1$ for $V_\ast \leq 9$.}
    \label{fig:irrationalnumerics}
\end{figure}

   Previous literature \cite{hur2025stabilitymanybodylocalizationdimensions} has replaced \eqref{eq:incommensuratelattice} with \begin{equation}
        H_0 = \sum_{\mathbf{n}\in\Lambda}\left[\sum_{j=1}^d W_j \cos\left(\varphi_j + 2\pi \alpha_j n_j\right)\right]Z_{\mathbf{n}} \label{eq:incommensuratelatticebad}
    \end{equation}
    where, as before, one picks a typical $\vec\alpha$ which is linearly independent over $\mathbb{Q}$.  Unfortunately, this Hamiltonian does not obey the non-resonance condition for any choice of $\vec \alpha$:
    \begin{prop}\label{prop:incommensuratebad}
    Fix spatial dimension $d\ge 2$.  Define \begin{equation}
        Z_{n_1} := \sum_{n_2, \dots, n_{d}\in\mathbb{Z}^{d-1}}Z_{(n_1,n_2,\ldots,n_d)}, \label{eq:Zn2nd}
    \end{equation}
    and define $Z_{n_j}$ etc. similarly for the remaining $d-1$ dimensions. 
    Let $A$ be an arbitrary local operator that obeys \begin{equation}
       \left[A, Z_{n_j}\right] = 0\label{eq:Asubsystem}
    \end{equation}
    for every choice of $n_j\in\mathbb{Z}$ and $1 \leq j \leq d$.
        Then $H_0$ defined in \eqref{eq:incommensuratelatticebad} obeys $[H_0, A] = 0$. In \eqref{eq:Zn2nd} the sum over sites can be restricted to the sites in $\Lambda\subset \mathbb{Z}^d$.
    \end{prop}

\begin{proof}
    Rewriting \eqref{eq:incommensuratelatticebad} in the form \begin{equation}
        H_0 = W_1 \left[\sum_{n_1\in\mathbb{Z}} \sum_{(n_2,\ldots, n_d)\in\mathbb{Z}^{d-1}} \cos (\varphi_1 + 2\pi\alpha_1n_1)Z_{\mathbf{n}}\right]+ \cdots  = W_1 \sum_{n_1\in\mathbb{Z}}\cos (\varphi_1 + 2\pi\alpha_1n_1)Z_{n_1} + \cdots,
    \end{equation}
    we see that $[H_0,A]=0$ follows from \eqref{eq:Asubsystem}, so long as $A$ only acts within $\Lambda\subset\mathbb{Z}^d$, defined as in Theorem \ref{prop:incommensurate_potentials}.
\end{proof}

 Proposition \ref{prop:incommensuratebad} implies that in $d\ge 2$, there are a very large number of resonant clusters starting at a finite size in the classical Hamiltonian $H_0$.  Indeed, a minimal operator $A$ that flips a resonant cluster corresponds to the spin flip operator \begin{equation}
     A = \prod_{\vec n \in\lbrace 0,1\rbrace^d} X_{\vec n}\frac{1+(-1)^{n_1+\cdots + n_d} Z_{\vec n}}{2}.
 \end{equation}This pattern of spin flips is related to having a linear subsystem symmetry on a square lattice, see e.g. \cite{Iaconis:2019hab,Gromov:2020yoc}.  Our theorems do not prove slow dynamics when $H_0$ is given by \eqref{eq:incommensuratelatticebad} beyond a fixed value of $r_*=d$, when $d>1$.

\end{appendix}

\bibliography{thebib}

\begin{thebibliography}{77}%
\makeatletter
\providecommand \@ifxundefined [1]{%
 \@ifx{#1\undefined}
}%
\providecommand \@ifnum [1]{%
 \ifnum #1\expandafter \@firstoftwo
 \else \expandafter \@secondoftwo
 \fi
}%
\providecommand \@ifx [1]{%
 \ifx #1\expandafter \@firstoftwo
 \else \expandafter \@secondoftwo
 \fi
}%
\providecommand \natexlab [1]{#1}%
\providecommand \enquote  [1]{``#1''}%
\providecommand \bibnamefont  [1]{#1}%
\providecommand \bibfnamefont [1]{#1}%
\providecommand \citenamefont [1]{#1}%
\providecommand \href@noop [0]{\@secondoftwo}%
\providecommand \href [0]{\begingroup \@sanitize@url \@href}%
\providecommand \@href[1]{\@@startlink{#1}\@@href}%
\providecommand \@@href[1]{\endgroup#1\@@endlink}%
\providecommand \@sanitize@url [0]{\catcode `\\12\catcode `\$12\catcode
  `\&12\catcode `\#12\catcode `\^12\catcode `\_12\catcode `\%12\relax}%
\providecommand \@@startlink[1]{}%
\providecommand \@@endlink[0]{}%
\providecommand \url  [0]{\begingroup\@sanitize@url \@url }%
\providecommand \@url [1]{\endgroup\@href {#1}{\urlprefix }}%
\providecommand \urlprefix  [0]{URL }%
\providecommand \Eprint [0]{\href }%
\providecommand \doibase [0]{http://dx.doi.org/}%
\providecommand \selectlanguage [0]{\@gobble}%
\providecommand \bibinfo  [0]{\@secondoftwo}%
\providecommand \bibfield  [0]{\@secondoftwo}%
\providecommand \translation [1]{[#1]}%
\providecommand \BibitemOpen [0]{}%
\providecommand \bibitemStop [0]{}%
\providecommand \bibitemNoStop [0]{.\EOS\space}%
\providecommand \EOS [0]{\spacefactor3000\relax}%
\providecommand \BibitemShut  [1]{\csname bibitem#1\endcsname}%
\let\auto@bib@innerbib\@empty
\bibitem [{\citenamefont {Anderson}(1958)}]{anderson}%
  \BibitemOpen
  \bibfield  {author} {\bibinfo {author} {\bibfnamefont {Philip~W.}\
  \bibnamefont {Anderson}},\ }\bibfield  {title} {\enquote {\bibinfo {title}
  {Absence of diffusion in certain random lattices},}\ }\href {\doibase
  10.1103/PhysRev.109.1492} {\bibfield  {journal} {\bibinfo  {journal} {Phys.
  Rev.}\ }\textbf {\bibinfo {volume} {109}},\ \bibinfo {pages} {1492--1505}
  (\bibinfo {year} {1958})}\BibitemShut {NoStop}%
\bibitem [{\citenamefont {Fr{\"o}hlich}\ and\ \citenamefont
  {Spencer}(1983)}]{frohlich1983absence}%
  \BibitemOpen
  \bibfield  {author} {\bibinfo {author} {\bibfnamefont {J{\"u}rg}\
  \bibnamefont {Fr{\"o}hlich}}\ and\ \bibinfo {author} {\bibfnamefont {Thomas}\
  \bibnamefont {Spencer}},\ }\bibfield  {title} {\enquote {\bibinfo {title}
  {{Absence of diffusion in the Anderson tight binding model for large disorder
  or low energy}},}\ }\href {\doibase https://doi.org/10.1007/BF01209475}
  {\bibfield  {journal} {\bibinfo  {journal} {Communications in Mathematical
  Physics}\ }\textbf {\bibinfo {volume} {88}},\ \bibinfo {pages} {151--184}
  (\bibinfo {year} {1983})}\BibitemShut {NoStop}%
\bibitem [{\citenamefont {Aizenman}\ and\ \citenamefont
  {Molchanov}(1993)}]{aizenman1993localization}%
  \BibitemOpen
  \bibfield  {author} {\bibinfo {author} {\bibfnamefont {Michael}\ \bibnamefont
  {Aizenman}}\ and\ \bibinfo {author} {\bibfnamefont {Stanislav}\ \bibnamefont
  {Molchanov}},\ }\bibfield  {title} {\enquote {\bibinfo {title} {{Localization
  at large disorder and at extreme energies: An elementary derivations}},}\
  }\href {\doibase https://doi.org/10.1007/BF02099760} {\bibfield  {journal}
  {\bibinfo  {journal} {Communications in Mathematical Physics}\ }\textbf
  {\bibinfo {volume} {157}},\ \bibinfo {pages} {245--278} (\bibinfo {year}
  {1993})}\BibitemShut {NoStop}%
\bibitem [{\citenamefont {Aizenman}\ and\ \citenamefont
  {Warzel}(2015)}]{aizenman2015random}%
  \BibitemOpen
  \bibfield  {author} {\bibinfo {author} {\bibfnamefont {Michael}\ \bibnamefont
  {Aizenman}}\ and\ \bibinfo {author} {\bibfnamefont {Simone}\ \bibnamefont
  {Warzel}},\ }\href {\doibase 10.1090/gsm/168} {\emph {\bibinfo {title}
  {{Random Operators: Disorder Effects on Quantum Spectra and Dynamics}}}},\
  Vol.\ \bibinfo {volume} {168}\ (\bibinfo  {publisher} {American Mathematical
  Soc.},\ \bibinfo {year} {2015})\BibitemShut {NoStop}%
\bibitem [{\citenamefont {Basko}\ \emph {et~al.}(2006)\citenamefont {Basko},
  \citenamefont {Aleiner},\ and\ \citenamefont
  {Altshuler}}]{Basko_MBL_foundations}%
  \BibitemOpen
  \bibfield  {author} {\bibinfo {author} {\bibfnamefont {D.M.}\ \bibnamefont
  {Basko}}, \bibinfo {author} {\bibfnamefont {I.L.}\ \bibnamefont {Aleiner}}, \
  and\ \bibinfo {author} {\bibfnamefont {B.L.}\ \bibnamefont {Altshuler}},\
  }\bibfield  {title} {\enquote {\bibinfo {title} {Metal–insulator transition
  in a weakly interacting many-electron system with localized single-particle
  states},}\ }\href {\doibase https://doi.org/10.1016/j.aop.2005.11.014}
  {\bibfield  {journal} {\bibinfo  {journal} {Annals of Physics}\ }\textbf
  {\bibinfo {volume} {321}},\ \bibinfo {pages} {1126--1205} (\bibinfo {year}
  {2006})}\BibitemShut {NoStop}%
\bibitem [{\citenamefont {Gornyi}\ \emph {et~al.}(2005)\citenamefont {Gornyi},
  \citenamefont {Mirlin},\ and\ \citenamefont
  {Polyakov}}]{Gornyi_MBL_foundations}%
  \BibitemOpen
  \bibfield  {author} {\bibinfo {author} {\bibfnamefont {I.~V.}\ \bibnamefont
  {Gornyi}}, \bibinfo {author} {\bibfnamefont {A.~D.}\ \bibnamefont {Mirlin}},
  \ and\ \bibinfo {author} {\bibfnamefont {D.~G.}\ \bibnamefont {Polyakov}},\
  }\bibfield  {title} {\enquote {\bibinfo {title} {Interacting electrons in
  disordered wires: Anderson localization and low-$t$ transport},}\ }\href
  {\doibase 10.1103/PhysRevLett.95.206603} {\bibfield  {journal} {\bibinfo
  {journal} {Phys. Rev. Lett.}\ }\textbf {\bibinfo {volume} {95}},\ \bibinfo
  {pages} {206603} (\bibinfo {year} {2005})}\BibitemShut {NoStop}%
\bibitem [{\citenamefont {Yin}\ \emph {et~al.}(2024)\citenamefont {Yin},
  \citenamefont {Nandkishore},\ and\ \citenamefont {Lucas}}]{Yin:2024jad}%
  \BibitemOpen
  \bibfield  {author} {\bibinfo {author} {\bibfnamefont {Chao}\ \bibnamefont
  {Yin}}, \bibinfo {author} {\bibfnamefont {Rahul}\ \bibnamefont
  {Nandkishore}}, \ and\ \bibinfo {author} {\bibfnamefont {Andrew}\
  \bibnamefont {Lucas}},\ }\bibfield  {title} {\enquote {\bibinfo {title}
  {{Eigenstate Localization in a Many-Body Quantum System}},}\ }\href {\doibase
  10.1103/PhysRevLett.133.137101} {\bibfield  {journal} {\bibinfo  {journal}
  {Phys. Rev. Lett.}\ }\textbf {\bibinfo {volume} {133}},\ \bibinfo {pages}
  {137101} (\bibinfo {year} {2024})},\ \Eprint
  {http://arxiv.org/abs/2405.12279} {arXiv:2405.12279 [cond-mat.stat-mech]}
  \BibitemShut {NoStop}%
\bibitem [{\citenamefont {Abanin}\ \emph {et~al.}(2019)\citenamefont {Abanin},
  \citenamefont {Altman}, \citenamefont {Bloch},\ and\ \citenamefont
  {Serbyn}}]{MBLColloquium}%
  \BibitemOpen
  \bibfield  {author} {\bibinfo {author} {\bibfnamefont {Dmitry~A.}\
  \bibnamefont {Abanin}}, \bibinfo {author} {\bibfnamefont {Ehud}\ \bibnamefont
  {Altman}}, \bibinfo {author} {\bibfnamefont {Immanuel}\ \bibnamefont
  {Bloch}}, \ and\ \bibinfo {author} {\bibfnamefont {Maksym}\ \bibnamefont
  {Serbyn}},\ }\bibfield  {title} {\enquote {\bibinfo {title} {Colloquium:
  Many-body localization, thermalization, and entanglement},}\ }\href {\doibase
  10.1103/RevModPhys.91.021001} {\bibfield  {journal} {\bibinfo  {journal}
  {Rev. Mod. Phys.}\ }\textbf {\bibinfo {volume} {91}},\ \bibinfo {pages}
  {021001} (\bibinfo {year} {2019})}\BibitemShut {NoStop}%
\bibitem [{\citenamefont {De~Roeck}\ and\ \citenamefont
  {Huveneers}(2017)}]{deroeckhuveneers}%
  \BibitemOpen
  \bibfield  {author} {\bibinfo {author} {\bibfnamefont {Wojciech}\
  \bibnamefont {De~Roeck}}\ and\ \bibinfo {author} {\bibfnamefont {Fran\ifmmode
  \mbox{\c{c}}\else~\c{c}\fi{}ois}\ \bibnamefont {Huveneers}},\ }\bibfield
  {title} {\enquote {\bibinfo {title} {Stability and instability towards
  delocalization in many-body localization systems},}\ }\href {\doibase
  10.1103/PhysRevB.95.155129} {\bibfield  {journal} {\bibinfo  {journal} {Phys.
  Rev. B}\ }\textbf {\bibinfo {volume} {95}},\ \bibinfo {pages} {155129}
  (\bibinfo {year} {2017})}\BibitemShut {NoStop}%
\bibitem [{\citenamefont {Thiery}\ \emph {et~al.}(2018)\citenamefont {Thiery},
  \citenamefont {Huveneers}, \citenamefont {M{\"u}ller},\ and\ \citenamefont
  {De~Roeck}}]{thiery2018many}%
  \BibitemOpen
  \bibfield  {author} {\bibinfo {author} {\bibfnamefont {Thimoth{\'e}e}\
  \bibnamefont {Thiery}}, \bibinfo {author} {\bibfnamefont {Fran{\c c}ois}\
  \bibnamefont {Huveneers}}, \bibinfo {author} {\bibfnamefont {Markus}\
  \bibnamefont {M{\"u}ller}}, \ and\ \bibinfo {author} {\bibfnamefont
  {Wojciech}\ \bibnamefont {De~Roeck}},\ }\bibfield  {title} {\enquote
  {\bibinfo {title} {Many-body delocalization as a quantum avalanche},}\ }\href
  {\doibase 10.1103/PhysRevLett.121.140601} {\bibfield  {journal} {\bibinfo
  {journal} {Physical Review Letters}\ }\textbf {\bibinfo {volume} {121}},\
  \bibinfo {pages} {140601} (\bibinfo {year} {2018})}\BibitemShut {NoStop}%
\bibitem [{\citenamefont {Morningstar}\ \emph {et~al.}(2020)\citenamefont
  {Morningstar}, \citenamefont {Huse},\ and\ \citenamefont
  {Imbrie}}]{morningstar2020many}%
  \BibitemOpen
  \bibfield  {author} {\bibinfo {author} {\bibfnamefont {Alan}\ \bibnamefont
  {Morningstar}}, \bibinfo {author} {\bibfnamefont {David~A.}\ \bibnamefont
  {Huse}}, \ and\ \bibinfo {author} {\bibfnamefont {John~Z.}\ \bibnamefont
  {Imbrie}},\ }\bibfield  {title} {\enquote {\bibinfo {title} {Many-body
  localization near the critical point},}\ }\href {\doibase
  10.1103/PhysRevB.102.125134} {\bibfield  {journal} {\bibinfo  {journal}
  {Physical Review B}\ }\textbf {\bibinfo {volume} {102}},\ \bibinfo {pages}
  {125134} (\bibinfo {year} {2020})}\BibitemShut {NoStop}%
\bibitem [{\citenamefont {Pal}\ and\ \citenamefont {Huse}(2010)}]{pal2010many}%
  \BibitemOpen
  \bibfield  {author} {\bibinfo {author} {\bibfnamefont {Arijeet}\ \bibnamefont
  {Pal}}\ and\ \bibinfo {author} {\bibfnamefont {David~A}\ \bibnamefont
  {Huse}},\ }\bibfield  {title} {\enquote {\bibinfo {title} {Many-body
  localization phase transition},}\ }\href@noop {} {\bibfield  {journal}
  {\bibinfo  {journal} {Physical Review B—Condensed Matter and Materials
  Physics}\ }\textbf {\bibinfo {volume} {82}},\ \bibinfo {pages} {174411}
  (\bibinfo {year} {2010})}\BibitemShut {NoStop}%
\bibitem [{\citenamefont {Lukin}\ \emph {et~al.}(2019)\citenamefont {Lukin},
  \citenamefont {Rispoli}, \citenamefont {Schittko}, \citenamefont {Tai},
  \citenamefont {Kaufman}, \citenamefont {Choi}, \citenamefont {Khemani},
  \citenamefont {Léonard},\ and\ \citenamefont
  {Greiner}}]{alukin_entanglement}%
  \BibitemOpen
  \bibfield  {author} {\bibinfo {author} {\bibfnamefont {Alexander}\
  \bibnamefont {Lukin}}, \bibinfo {author} {\bibfnamefont {Matthew}\
  \bibnamefont {Rispoli}}, \bibinfo {author} {\bibfnamefont {Robert}\
  \bibnamefont {Schittko}}, \bibinfo {author} {\bibfnamefont {M.~Eric}\
  \bibnamefont {Tai}}, \bibinfo {author} {\bibfnamefont {Adam~M.}\ \bibnamefont
  {Kaufman}}, \bibinfo {author} {\bibfnamefont {Soonwon}\ \bibnamefont {Choi}},
  \bibinfo {author} {\bibfnamefont {Vedika}\ \bibnamefont {Khemani}}, \bibinfo
  {author} {\bibfnamefont {Julian}\ \bibnamefont {Léonard}}, \ and\ \bibinfo
  {author} {\bibfnamefont {Markus}\ \bibnamefont {Greiner}},\ }\bibfield
  {title} {\enquote {\bibinfo {title} {Probing entanglement in a
  many-body–localized system},}\ }\href {\doibase 10.1126/science.aau0818}
  {\bibfield  {journal} {\bibinfo  {journal} {Science}\ }\textbf {\bibinfo
  {volume} {364}},\ \bibinfo {pages} {256--260} (\bibinfo {year} {2019})},\
  \Eprint
  {http://arxiv.org/abs/https://www.science.org/doi/pdf/10.1126/science.aau0818}
  {https://www.science.org/doi/pdf/10.1126/science.aau0818} \BibitemShut
  {NoStop}%
\bibitem [{\citenamefont {Choi}\ \emph {et~al.}(2016)\citenamefont {Choi},
  \citenamefont {Hild}, \citenamefont {Zeiher}, \citenamefont {Schauß},
  \citenamefont {Rubio-Abadal}, \citenamefont {Yefsah}, \citenamefont
  {Khemani}, \citenamefont {Huse}, \citenamefont {Bloch},\ and\ \citenamefont
  {Gross}}]{Choi_2016}%
  \BibitemOpen
  \bibfield  {author} {\bibinfo {author} {\bibfnamefont {Jae-yoon}\
  \bibnamefont {Choi}}, \bibinfo {author} {\bibfnamefont {Sebastian}\
  \bibnamefont {Hild}}, \bibinfo {author} {\bibfnamefont {Johannes}\
  \bibnamefont {Zeiher}}, \bibinfo {author} {\bibfnamefont {Peter}\
  \bibnamefont {Schauß}}, \bibinfo {author} {\bibfnamefont {Antonio}\
  \bibnamefont {Rubio-Abadal}}, \bibinfo {author} {\bibfnamefont {Tarik}\
  \bibnamefont {Yefsah}}, \bibinfo {author} {\bibfnamefont {Vedika}\
  \bibnamefont {Khemani}}, \bibinfo {author} {\bibfnamefont {David~A.}\
  \bibnamefont {Huse}}, \bibinfo {author} {\bibfnamefont {Immanuel}\
  \bibnamefont {Bloch}}, \ and\ \bibinfo {author} {\bibfnamefont {Christian}\
  \bibnamefont {Gross}},\ }\bibfield  {title} {\enquote {\bibinfo {title}
  {Exploring the many-body localization transition in two dimensions},}\ }\href
  {\doibase 10.1126/science.aaf8834} {\bibfield  {journal} {\bibinfo  {journal}
  {Science}\ }\textbf {\bibinfo {volume} {352}},\ \bibinfo {pages}
  {1547–1552} (\bibinfo {year} {2016})}\BibitemShut {NoStop}%
\bibitem [{\citenamefont {Li}\ \emph {et~al.}(2025)\citenamefont {Li} \emph
  {et~al.}}]{Li:2025kje}%
  \BibitemOpen
  \bibfield  {author} {\bibinfo {author} {\bibfnamefont {Tian-Ming}\
  \bibnamefont {Li}} \emph {et~al.},\ }\bibfield  {title} {\enquote {\bibinfo
  {title} {{Many-body delocalization with a two-dimensional 70-qubit
  superconducting quantum simulator}},}\ }\href {\doibase
  10.48550/arXiv.2507.16882} {\  (\bibinfo {year} {2025}),\
  10.48550/arXiv.2507.16882},\ \Eprint {http://arxiv.org/abs/2507.16882}
  {arXiv:2507.16882 [quant-ph]} \BibitemShut {NoStop}%
\bibitem [{\citenamefont {Hur}\ \emph {et~al.}(2025)\citenamefont {Hur},
  \citenamefont {Li}, \citenamefont {Lee}, \citenamefont {Kwon}, \citenamefont
  {Kim}, \citenamefont {Hwang}, \citenamefont {Kim}, \citenamefont {Yu},
  \citenamefont {Chan}, \citenamefont {Wahl},\ and\ \citenamefont {yoon
  Choi}}]{hur2025stabilitymanybodylocalizationdimensions}%
  \BibitemOpen
  \bibfield  {author} {\bibinfo {author} {\bibfnamefont {Junhyeok}\
  \bibnamefont {Hur}}, \bibinfo {author} {\bibfnamefont {Joey}\ \bibnamefont
  {Li}}, \bibinfo {author} {\bibfnamefont {Byungjin}\ \bibnamefont {Lee}},
  \bibinfo {author} {\bibfnamefont {Kiryang}\ \bibnamefont {Kwon}}, \bibinfo
  {author} {\bibfnamefont {Minseok}\ \bibnamefont {Kim}}, \bibinfo {author}
  {\bibfnamefont {Samgyu}\ \bibnamefont {Hwang}}, \bibinfo {author}
  {\bibfnamefont {Sumin}\ \bibnamefont {Kim}}, \bibinfo {author} {\bibfnamefont
  {Yong~Soo}\ \bibnamefont {Yu}}, \bibinfo {author} {\bibfnamefont {Amos}\
  \bibnamefont {Chan}}, \bibinfo {author} {\bibfnamefont {Thorsten}\
  \bibnamefont {Wahl}}, \ and\ \bibinfo {author} {\bibfnamefont {Jae}\
  \bibnamefont {yoon Choi}},\ }\href {https://arxiv.org/abs/2508.20699}
  {\enquote {\bibinfo {title} {Stability of many-body localization in two
  dimensions},}\ } (\bibinfo {year} {2025}),\ \Eprint
  {http://arxiv.org/abs/2508.20699} {arXiv:2508.20699 [cond-mat.quant-gas]}
  \BibitemShut {NoStop}%
\bibitem [{\citenamefont {Schreiber}\ \emph {et~al.}(2015)\citenamefont
  {Schreiber}, \citenamefont {Hodgman}, \citenamefont {Bordia}, \citenamefont
  {Lüschen}, \citenamefont {Fischer}, \citenamefont {Vosk}, \citenamefont
  {Altman}, \citenamefont {Schneider},\ and\ \citenamefont
  {Bloch}}]{Schreiber_2015}%
  \BibitemOpen
  \bibfield  {author} {\bibinfo {author} {\bibfnamefont {Michael}\ \bibnamefont
  {Schreiber}}, \bibinfo {author} {\bibfnamefont {Sean~S.}\ \bibnamefont
  {Hodgman}}, \bibinfo {author} {\bibfnamefont {Pranjal}\ \bibnamefont
  {Bordia}}, \bibinfo {author} {\bibfnamefont {Henrik~P.}\ \bibnamefont
  {Lüschen}}, \bibinfo {author} {\bibfnamefont {Mark~H.}\ \bibnamefont
  {Fischer}}, \bibinfo {author} {\bibfnamefont {Ronen}\ \bibnamefont {Vosk}},
  \bibinfo {author} {\bibfnamefont {Ehud}\ \bibnamefont {Altman}}, \bibinfo
  {author} {\bibfnamefont {Ulrich}\ \bibnamefont {Schneider}}, \ and\ \bibinfo
  {author} {\bibfnamefont {Immanuel}\ \bibnamefont {Bloch}},\ }\bibfield
  {title} {\enquote {\bibinfo {title} {Observation of many-body localization of
  interacting fermions in a quasirandom optical lattice},}\ }\href {\doibase
  10.1126/science.aaa7432} {\bibfield  {journal} {\bibinfo  {journal}
  {Science}\ }\textbf {\bibinfo {volume} {349}},\ \bibinfo {pages} {842–845}
  (\bibinfo {year} {2015})}\BibitemShut {NoStop}%
\bibitem [{\citenamefont {Smith}\ \emph {et~al.}(2016)\citenamefont {Smith},
  \citenamefont {Lee}, \citenamefont {Richerme}, \citenamefont {Neyenhuis},
  \citenamefont {Hess}, \citenamefont {Hauke}, \citenamefont {Heyl},
  \citenamefont {Huse},\ and\ \citenamefont {Monroe}}]{Smith2016}%
  \BibitemOpen
  \bibfield  {author} {\bibinfo {author} {\bibfnamefont {J.}~\bibnamefont
  {Smith}}, \bibinfo {author} {\bibfnamefont {A.}~\bibnamefont {Lee}}, \bibinfo
  {author} {\bibfnamefont {P.}~\bibnamefont {Richerme}}, \bibinfo {author}
  {\bibfnamefont {B.}~\bibnamefont {Neyenhuis}}, \bibinfo {author}
  {\bibfnamefont {P.~W.}\ \bibnamefont {Hess}}, \bibinfo {author}
  {\bibfnamefont {P.}~\bibnamefont {Hauke}}, \bibinfo {author} {\bibfnamefont
  {M.}~\bibnamefont {Heyl}}, \bibinfo {author} {\bibfnamefont {D.~A.}\
  \bibnamefont {Huse}}, \ and\ \bibinfo {author} {\bibfnamefont
  {C.}~\bibnamefont {Monroe}},\ }\bibfield  {title} {\enquote {\bibinfo {title}
  {Many-body localization in a quantum simulator with programmable random
  disorder},}\ }\href {\doibase 10.1038/nphys3783} {\bibfield  {journal}
  {\bibinfo  {journal} {Nature Physics}\ }\textbf {\bibinfo {volume} {12}},\
  \bibinfo {pages} {907--911} (\bibinfo {year} {2016})}\BibitemShut {NoStop}%
\bibitem [{\citenamefont {Bordia}\ \emph {et~al.}(2017)\citenamefont {Bordia},
  \citenamefont {L\"uschen}, \citenamefont {Scherg}, \citenamefont
  {Gopalakrishnan}, \citenamefont {Knap}, \citenamefont {Schneider},\ and\
  \citenamefont {Bloch}}]{localization_2d_quasiperiodic}%
  \BibitemOpen
  \bibfield  {author} {\bibinfo {author} {\bibfnamefont {Pranjal}\ \bibnamefont
  {Bordia}}, \bibinfo {author} {\bibfnamefont {Henrik}\ \bibnamefont
  {L\"uschen}}, \bibinfo {author} {\bibfnamefont {Sebastian}\ \bibnamefont
  {Scherg}}, \bibinfo {author} {\bibfnamefont {Sarang}\ \bibnamefont
  {Gopalakrishnan}}, \bibinfo {author} {\bibfnamefont {Michael}\ \bibnamefont
  {Knap}}, \bibinfo {author} {\bibfnamefont {Ulrich}\ \bibnamefont
  {Schneider}}, \ and\ \bibinfo {author} {\bibfnamefont {Immanuel}\
  \bibnamefont {Bloch}},\ }\bibfield  {title} {\enquote {\bibinfo {title}
  {Probing slow relaxation and many-body localization in two-dimensional
  quasiperiodic systems},}\ }\href {\doibase 10.1103/PhysRevX.7.041047}
  {\bibfield  {journal} {\bibinfo  {journal} {Phys. Rev. X}\ }\textbf {\bibinfo
  {volume} {7}},\ \bibinfo {pages} {041047} (\bibinfo {year}
  {2017})}\BibitemShut {NoStop}%
\bibitem [{\citenamefont {Burin}(2006)}]{MBL_powerlaw}%
  \BibitemOpen
  \bibfield  {author} {\bibinfo {author} {\bibfnamefont {Alexander~L.}\
  \bibnamefont {Burin}},\ }\href {https://arxiv.org/abs/cond-mat/0611387}
  {\enquote {\bibinfo {title} {Energy delocalization in strongly disordered
  systems induced by the long-range many-body interaction},}\ } (\bibinfo
  {year} {2006}),\ \Eprint {http://arxiv.org/abs/cond-mat/0611387}
  {arXiv:cond-mat/0611387 [cond-mat.dis-nn]} \BibitemShut {NoStop}%
\bibitem [{\citenamefont {Nandkishore}\ and\ \citenamefont
  {Sondhi}(2017)}]{MBL_long_range}%
  \BibitemOpen
  \bibfield  {author} {\bibinfo {author} {\bibfnamefont {Rahul~M.}\
  \bibnamefont {Nandkishore}}\ and\ \bibinfo {author} {\bibfnamefont {S.~L.}\
  \bibnamefont {Sondhi}},\ }\bibfield  {title} {\enquote {\bibinfo {title}
  {Many-body localization with long-range interactions},}\ }\href {\doibase
  10.1103/PhysRevX.7.041021} {\bibfield  {journal} {\bibinfo  {journal} {Phys.
  Rev. X}\ }\textbf {\bibinfo {volume} {7}},\ \bibinfo {pages} {041021}
  (\bibinfo {year} {2017})}\BibitemShut {NoStop}%
\bibitem [{\citenamefont {Imbrie}(2016)}]{imbrie2016many}%
  \BibitemOpen
  \bibfield  {author} {\bibinfo {author} {\bibfnamefont {John~Z}\ \bibnamefont
  {Imbrie}},\ }\bibfield  {title} {\enquote {\bibinfo {title} {On many-body
  localization for quantum spin chains},}\ }\href {\doibase
  10.1007/s10955-016-1508-x} {\bibfield  {journal} {\bibinfo  {journal}
  {Journal of Statistical Physics}\ }\textbf {\bibinfo {volume} {163}},\
  \bibinfo {pages} {998--1048} (\bibinfo {year} {2016})}\BibitemShut {NoStop}%
\bibitem [{\citenamefont {Roeck}\ \emph {et~al.}(2025)\citenamefont {Roeck},
  \citenamefont {Giacomin}, \citenamefont {Huveneers},\ and\ \citenamefont
  {Prosniak}}]{absenceofconduction}%
  \BibitemOpen
  \bibfield  {author} {\bibinfo {author} {\bibfnamefont {Wojciech~De}\
  \bibnamefont {Roeck}}, \bibinfo {author} {\bibfnamefont {Lydia}\ \bibnamefont
  {Giacomin}}, \bibinfo {author} {\bibfnamefont {Francois}\ \bibnamefont
  {Huveneers}}, \ and\ \bibinfo {author} {\bibfnamefont {Oskar}\ \bibnamefont
  {Prosniak}},\ }\href {https://arxiv.org/abs/2408.04338} {\enquote {\bibinfo
  {title} {Absence of normal heat conduction in strongly disordered interacting
  quantum chains},}\ } (\bibinfo {year} {2025}),\ \Eprint
  {http://arxiv.org/abs/2408.04338} {arXiv:2408.04338 [math-ph]} \BibitemShut
  {NoStop}%
\bibitem [{\citenamefont {Mastropietro}(2015)}]{mastropietro2015localization}%
  \BibitemOpen
  \bibfield  {author} {\bibinfo {author} {\bibfnamefont {Vieri}\ \bibnamefont
  {Mastropietro}},\ }\bibfield  {title} {\enquote {\bibinfo {title}
  {{Localization of interacting fermions in the Aubry-Andr{\'e} model}},}\
  }\href {\doibase https://doi.org/10.1103/PhysRevLett.115.180401} {\bibfield
  {journal} {\bibinfo  {journal} {Physical review letters}\ }\textbf {\bibinfo
  {volume} {115}},\ \bibinfo {pages} {180401} (\bibinfo {year}
  {2015})}\BibitemShut {NoStop}%
\bibitem [{\citenamefont {Baldwin}\ \emph {et~al.}(2023)\citenamefont
  {Baldwin}, \citenamefont {Ehrenberg}, \citenamefont {Guo},\ and\
  \citenamefont {Gorshkov}}]{baldwingorshkov}%
  \BibitemOpen
  \bibfield  {author} {\bibinfo {author} {\bibfnamefont {Christopher~L.}\
  \bibnamefont {Baldwin}}, \bibinfo {author} {\bibfnamefont {Adam}\
  \bibnamefont {Ehrenberg}}, \bibinfo {author} {\bibfnamefont {Andrew~Y.}\
  \bibnamefont {Guo}}, \ and\ \bibinfo {author} {\bibfnamefont {Alexey~V.}\
  \bibnamefont {Gorshkov}},\ }\bibfield  {title} {\enquote {\bibinfo {title}
  {Disordered lieb-robinson bounds in one dimension},}\ }\href {\doibase
  10.1103/PRXQuantum.4.020349} {\bibfield  {journal} {\bibinfo  {journal} {PRX
  Quantum}\ }\textbf {\bibinfo {volume} {4}},\ \bibinfo {pages} {020349}
  (\bibinfo {year} {2023})}\BibitemShut {NoStop}%
\bibitem [{\citenamefont {Gebert}\ \emph {et~al.}(2022)\citenamefont {Gebert},
  \citenamefont {Moon},\ and\ \citenamefont {Nachtergaele}}]{gebert2022lieb}%
  \BibitemOpen
  \bibfield  {author} {\bibinfo {author} {\bibfnamefont {Martin}\ \bibnamefont
  {Gebert}}, \bibinfo {author} {\bibfnamefont {Alvin}\ \bibnamefont {Moon}}, \
  and\ \bibinfo {author} {\bibfnamefont {Bruno}\ \bibnamefont {Nachtergaele}},\
  }\bibfield  {title} {\enquote {\bibinfo {title} {A lieb--robinson bound for
  quantum spin chains with strong on-site impurities},}\ }\href@noop {}
  {\bibfield  {journal} {\bibinfo  {journal} {Reviews in Mathematical Physics}\
  }\textbf {\bibinfo {volume} {34}},\ \bibinfo {pages} {2250007} (\bibinfo
  {year} {2022})}\BibitemShut {NoStop}%
\bibitem [{\citenamefont {Toniolo}\ and\ \citenamefont
  {Bose}(2025)}]{toniolo2025logarithmic}%
  \BibitemOpen
  \bibfield  {author} {\bibinfo {author} {\bibfnamefont {Daniele}\ \bibnamefont
  {Toniolo}}\ and\ \bibinfo {author} {\bibfnamefont {Sougato}\ \bibnamefont
  {Bose}},\ }\bibfield  {title} {\enquote {\bibinfo {title} {Logarithmic
  lightcones in the multiparticle anderson model with sparse interactions},}\
  }\href {\doibase 10.48550/arXiv.2509.02383} {\bibfield  {journal} {\bibinfo
  {journal} {arXiv preprint arXiv:2509.02383}\ } (\bibinfo {year} {2025}),\
  10.48550/arXiv.2509.02383}\BibitemShut {NoStop}%
\bibitem [{\citenamefont {Baldwin}(2025)}]{baldwin2025subballistic}%
  \BibitemOpen
  \bibfield  {author} {\bibinfo {author} {\bibfnamefont {Christopher~L}\
  \bibnamefont {Baldwin}},\ }\bibfield  {title} {\enquote {\bibinfo {title}
  {Subballistic operator growth in spin chains with heavy-tailed random
  fields},}\ }\href {\doibase https://doi.org/10.1103/PhysRevB.111.184204}
  {\bibfield  {journal} {\bibinfo  {journal} {Physical Review B}\ }\textbf
  {\bibinfo {volume} {111}},\ \bibinfo {pages} {184204} (\bibinfo {year}
  {2025})}\BibitemShut {NoStop}%
\bibitem [{\citenamefont {Serbyn}\ \emph {et~al.}(2013)\citenamefont {Serbyn},
  \citenamefont {Papić},\ and\ \citenamefont {Abanin}}]{Serbyn_2013}%
  \BibitemOpen
  \bibfield  {author} {\bibinfo {author} {\bibfnamefont {Maksym}\ \bibnamefont
  {Serbyn}}, \bibinfo {author} {\bibfnamefont {Z.}~\bibnamefont {Papić}}, \
  and\ \bibinfo {author} {\bibfnamefont {Dmitry~A.}\ \bibnamefont {Abanin}},\
  }\bibfield  {title} {\enquote {\bibinfo {title} {Local conservation laws and
  the structure of the many-body localized states},}\ }\href {\doibase
  10.1103/physrevlett.111.127201} {\bibfield  {journal} {\bibinfo  {journal}
  {Physical Review Letters}\ }\textbf {\bibinfo {volume} {111}} (\bibinfo
  {year} {2013}),\ 10.1103/physrevlett.111.127201}\BibitemShut {NoStop}%
\bibitem [{\citenamefont {Huse}\ \emph {et~al.}(2014)\citenamefont {Huse},
  \citenamefont {Nandkishore},\ and\ \citenamefont
  {Oganesyan}}]{huse2014phenomenology}%
  \BibitemOpen
  \bibfield  {author} {\bibinfo {author} {\bibfnamefont {David~A}\ \bibnamefont
  {Huse}}, \bibinfo {author} {\bibfnamefont {Rahul}\ \bibnamefont
  {Nandkishore}}, \ and\ \bibinfo {author} {\bibfnamefont {Vadim}\ \bibnamefont
  {Oganesyan}},\ }\bibfield  {title} {\enquote {\bibinfo {title} {Phenomenology
  of fully many-body-localized systems},}\ }\href@noop {} {\bibfield  {journal}
  {\bibinfo  {journal} {Physical Review B}\ }\textbf {\bibinfo {volume} {90}},\
  \bibinfo {pages} {174202} (\bibinfo {year} {2014})}\BibitemShut {NoStop}%
\bibitem [{\citenamefont {{\v{S}}untajs}\ \emph {et~al.}(2020)\citenamefont
  {{\v{S}}untajs}, \citenamefont {Bon{\v{c}}a}, \citenamefont {Prosen},\ and\
  \citenamefont {Vidmar}}]{Suntajs:2019lmb}%
  \BibitemOpen
  \bibfield  {author} {\bibinfo {author} {\bibfnamefont {J.}~\bibnamefont
  {{\v{S}}untajs}}, \bibinfo {author} {\bibfnamefont {J.}~\bibnamefont
  {Bon{\v{c}}a}}, \bibinfo {author} {\bibfnamefont {T.}~\bibnamefont {Prosen}},
  \ and\ \bibinfo {author} {\bibfnamefont {L.}~\bibnamefont {Vidmar}},\
  }\bibfield  {title} {\enquote {\bibinfo {title} {{Quantum chaos challenges
  many-body localization}},}\ }\href {\doibase 10.1103/PhysRevE.102.062144}
  {\bibfield  {journal} {\bibinfo  {journal} {Phys. Rev. E}\ }\textbf {\bibinfo
  {volume} {102}},\ \bibinfo {pages} {062144} (\bibinfo {year} {2020})},\
  \Eprint {http://arxiv.org/abs/1905.06345} {arXiv:1905.06345
  [cond-mat.str-el]} \BibitemShut {NoStop}%
\bibitem [{\citenamefont {Morningstar}\ \emph {et~al.}(2022)\citenamefont
  {Morningstar}, \citenamefont {Colmenarez}, \citenamefont {Khemani},
  \citenamefont {Luitz},\ and\ \citenamefont {Huse}}]{Morningstar:2021pcy}%
  \BibitemOpen
  \bibfield  {author} {\bibinfo {author} {\bibfnamefont {Alan}\ \bibnamefont
  {Morningstar}}, \bibinfo {author} {\bibfnamefont {Luis}\ \bibnamefont
  {Colmenarez}}, \bibinfo {author} {\bibfnamefont {Vedika}\ \bibnamefont
  {Khemani}}, \bibinfo {author} {\bibfnamefont {David~J.}\ \bibnamefont
  {Luitz}}, \ and\ \bibinfo {author} {\bibfnamefont {David~A.}\ \bibnamefont
  {Huse}},\ }\bibfield  {title} {\enquote {\bibinfo {title} {{Avalanches and
  many-body resonances in many-body localized systems}},}\ }\href {\doibase
  10.1103/PhysRevB.105.174205} {\bibfield  {journal} {\bibinfo  {journal}
  {Phys. Rev. B}\ }\textbf {\bibinfo {volume} {105}},\ \bibinfo {pages}
  {174205} (\bibinfo {year} {2022})},\ \Eprint
  {http://arxiv.org/abs/2107.05642} {arXiv:2107.05642 [cond-mat.dis-nn]}
  \BibitemShut {NoStop}%
\bibitem [{\citenamefont {Lieb}\ and\ \citenamefont
  {Robinson}(1972)}]{Lieb1972}%
  \BibitemOpen
  \bibfield  {author} {\bibinfo {author} {\bibfnamefont {Elliott~H.}\
  \bibnamefont {Lieb}}\ and\ \bibinfo {author} {\bibfnamefont {Derek~W.}\
  \bibnamefont {Robinson}},\ }\bibfield  {title} {\enquote {\bibinfo {title}
  {The finite group velocity of quantum spin systems},}\ }\href {\doibase
  10.1007/BF01645779} {\bibfield  {journal} {\bibinfo  {journal} {Commun. Math.
  Phys.}\ }\textbf {\bibinfo {volume} {28}},\ \bibinfo {pages} {251--257}
  (\bibinfo {year} {1972})}\BibitemShut {NoStop}%
\bibitem [{\citenamefont {(Anthony)~Chen}\ \emph {et~al.}(2023)\citenamefont
  {(Anthony)~Chen}, \citenamefont {Lucas},\ and\ \citenamefont
  {Yin}}]{AnthonyChen:2023bbe}%
  \BibitemOpen
  \bibfield  {author} {\bibinfo {author} {\bibfnamefont {Chi-Fang}\
  \bibnamefont {(Anthony)~Chen}}, \bibinfo {author} {\bibfnamefont {Andrew}\
  \bibnamefont {Lucas}}, \ and\ \bibinfo {author} {\bibfnamefont {Chao}\
  \bibnamefont {Yin}},\ }\bibfield  {title} {\enquote {\bibinfo {title} {{Speed
  limits and locality in many-body quantum dynamics}},}\ }\href {\doibase
  10.1088/1361-6633/acfaae} {\bibfield  {journal} {\bibinfo  {journal} {Rept.
  Prog. Phys.}\ }\textbf {\bibinfo {volume} {86}},\ \bibinfo {pages} {116001}
  (\bibinfo {year} {2023})},\ \Eprint {http://arxiv.org/abs/2303.07386}
  {arXiv:2303.07386 [quant-ph]} \BibitemShut {NoStop}%
\bibitem [{\citenamefont {Chen}\ and\ \citenamefont
  {Lucas}(2021)}]{chen2021operator}%
  \BibitemOpen
  \bibfield  {author} {\bibinfo {author} {\bibfnamefont {Chi-Fang}\
  \bibnamefont {Chen}}\ and\ \bibinfo {author} {\bibfnamefont {Andrew}\
  \bibnamefont {Lucas}},\ }\bibfield  {title} {\enquote {\bibinfo {title}
  {Operator growth bounds from graph theory},}\ }\href@noop {} {\bibfield
  {journal} {\bibinfo  {journal} {Communications in Mathematical Physics}\
  }\textbf {\bibinfo {volume} {385}},\ \bibinfo {pages} {1273--1323} (\bibinfo
  {year} {2021})}\BibitemShut {NoStop}%
\bibitem [{\citenamefont {Pr{\'e}mont-Schwarz}\ \emph
  {et~al.}(2010)\citenamefont {Pr{\'e}mont-Schwarz}, \citenamefont {Hamma},
  \citenamefont {Klich},\ and\ \citenamefont
  {Markopoulou-Kalamara}}]{premont2010lieb}%
  \BibitemOpen
  \bibfield  {author} {\bibinfo {author} {\bibfnamefont {Isabeau}\ \bibnamefont
  {Pr{\'e}mont-Schwarz}}, \bibinfo {author} {\bibfnamefont {Alioscia}\
  \bibnamefont {Hamma}}, \bibinfo {author} {\bibfnamefont {Israel}\
  \bibnamefont {Klich}}, \ and\ \bibinfo {author} {\bibfnamefont {Fotini}\
  \bibnamefont {Markopoulou-Kalamara}},\ }\bibfield  {title} {\enquote
  {\bibinfo {title} {{Lieb-Robinson bounds for commutator-bounded
  operators}},}\ }\href {\doibase 10.1103/PhysRevA.81.040102} {\bibfield
  {journal} {\bibinfo  {journal} {Physical Review A}\ }\textbf {\bibinfo
  {volume} {81}},\ \bibinfo {pages} {040102} (\bibinfo {year}
  {2010})}\BibitemShut {NoStop}%
\bibitem [{\citenamefont {Haah}\ \emph {et~al.}(2021)\citenamefont {Haah},
  \citenamefont {Hastings}, \citenamefont {Kothari},\ and\ \citenamefont
  {Low}}]{haah2021quantum}%
  \BibitemOpen
  \bibfield  {author} {\bibinfo {author} {\bibfnamefont {Jeongwan}\
  \bibnamefont {Haah}}, \bibinfo {author} {\bibfnamefont {Matthew~B}\
  \bibnamefont {Hastings}}, \bibinfo {author} {\bibfnamefont {Robin}\
  \bibnamefont {Kothari}}, \ and\ \bibinfo {author} {\bibfnamefont {Guang~Hao}\
  \bibnamefont {Low}},\ }\bibfield  {title} {\enquote {\bibinfo {title}
  {Quantum algorithm for simulating real time evolution of lattice
  hamiltonians},}\ }\href {\doibase 10.1137/18M1231511} {\bibfield  {journal}
  {\bibinfo  {journal} {SIAM Journal on Computing}\ }\textbf {\bibinfo {volume}
  {52}},\ \bibinfo {pages} {FOCS18--250} (\bibinfo {year} {2021})}\BibitemShut
  {NoStop}%
\bibitem [{\citenamefont {Wang}\ and\ \citenamefont
  {Hazzard}(2020)}]{wanghazzard}%
  \BibitemOpen
  \bibfield  {author} {\bibinfo {author} {\bibfnamefont {Zhiyuan}\ \bibnamefont
  {Wang}}\ and\ \bibinfo {author} {\bibfnamefont {Kaden~R.A.}\ \bibnamefont
  {Hazzard}},\ }\bibfield  {title} {\enquote {\bibinfo {title} {Tightening the
  lieb-robinson bound in locally interacting systems},}\ }\href {\doibase
  10.1103/PRXQuantum.1.010303} {\bibfield  {journal} {\bibinfo  {journal} {PRX
  Quantum}\ }\textbf {\bibinfo {volume} {1}},\ \bibinfo {pages} {010303}
  (\bibinfo {year} {2020})}\BibitemShut {NoStop}%
\bibitem [{\citenamefont {Lemm}\ and\ \citenamefont
  {Wessel}(2025)}]{lemm2025enhanced}%
  \BibitemOpen
  \bibfield  {author} {\bibinfo {author} {\bibfnamefont {Marius}\ \bibnamefont
  {Lemm}}\ and\ \bibinfo {author} {\bibfnamefont {Tom}\ \bibnamefont
  {Wessel}},\ }\bibfield  {title} {\enquote {\bibinfo {title} {{Enhanced
  Lieb-Robinson bounds for commuting long-range interactions}},}\ }\href
  {\doibase 10.1063/5.0252675} {\bibfield  {journal} {\bibinfo  {journal}
  {Journal of Mathematical Physics}\ }\textbf {\bibinfo {volume} {66}}
  (\bibinfo {year} {2025}),\ 10.1063/5.0252675}\BibitemShut {NoStop}%
\bibitem [{\citenamefont {Bravyi}\ \emph {et~al.}(2011)\citenamefont {Bravyi},
  \citenamefont {DiVincenzo},\ and\ \citenamefont
  {Loss}}]{bravyi2011schrieffer}%
  \BibitemOpen
  \bibfield  {author} {\bibinfo {author} {\bibfnamefont {Sergey}\ \bibnamefont
  {Bravyi}}, \bibinfo {author} {\bibfnamefont {David~P}\ \bibnamefont
  {DiVincenzo}}, \ and\ \bibinfo {author} {\bibfnamefont {Daniel}\ \bibnamefont
  {Loss}},\ }\bibfield  {title} {\enquote {\bibinfo {title} {{Schrieffer--Wolff
  transformation for quantum many-body systems}},}\ }\href {\doibase
  10.1016/j.aop.2011.06.004} {\bibfield  {journal} {\bibinfo  {journal} {Annals
  of Physics}\ }\textbf {\bibinfo {volume} {326}},\ \bibinfo {pages}
  {2793--2826} (\bibinfo {year} {2011})}\BibitemShut {NoStop}%
\bibitem [{\citenamefont {Roeck}\ \emph {et~al.}(2023)\citenamefont {Roeck},
  \citenamefont {Huveneers}, \citenamefont {Meeus},\ and\ \citenamefont
  {Prosniak}}]{deroeck2023}%
  \BibitemOpen
  \bibfield  {author} {\bibinfo {author} {\bibfnamefont {W.~De}\ \bibnamefont
  {Roeck}}, \bibinfo {author} {\bibfnamefont {F.}~\bibnamefont {Huveneers}},
  \bibinfo {author} {\bibfnamefont {B.}~\bibnamefont {Meeus}}, \ and\ \bibinfo
  {author} {\bibfnamefont {O.~A.}\ \bibnamefont {Prosniak}},\ }\bibfield
  {title} {\enquote {\bibinfo {title} {Rigorous and simple results on very slow
  thermalization, or quasi-localization, of the disordered quantum chain},}\
  }\href {\doibase https://doi.org/10.1016/j.physa.2023.129245} {\bibfield
  {journal} {\bibinfo  {journal} {Physica}\ }\textbf {\bibinfo {volume}
  {A631}},\ \bibinfo {pages} {129245} (\bibinfo {year} {2023})}\BibitemShut
  {NoStop}%
\bibitem [{\citenamefont {Bravyi}\ \emph {et~al.}(2010)\citenamefont {Bravyi},
  \citenamefont {Hastings},\ and\ \citenamefont {Michalakis}}]{bravyi2010}%
  \BibitemOpen
  \bibfield  {author} {\bibinfo {author} {\bibfnamefont {Sergey}\ \bibnamefont
  {Bravyi}}, \bibinfo {author} {\bibfnamefont {Matthew~B.}\ \bibnamefont
  {Hastings}}, \ and\ \bibinfo {author} {\bibfnamefont {Spyridon}\ \bibnamefont
  {Michalakis}},\ }\bibfield  {title} {\enquote {\bibinfo {title} {Topological
  quantum order: Stability under local perturbations},}\ }\href {\doibase
  10.1063/1.3490195} {\bibfield  {journal} {\bibinfo  {journal} {Journal of
  Mathematical Physics}\ }\textbf {\bibinfo {volume} {51}},\ \bibinfo {pages}
  {093512} (\bibinfo {year} {2010})}\BibitemShut {NoStop}%
\bibitem [{\citenamefont {Yin}\ and\ \citenamefont
  {Lucas}(2023)}]{yin2023prethermalization}%
  \BibitemOpen
  \bibfield  {author} {\bibinfo {author} {\bibfnamefont {Chao}\ \bibnamefont
  {Yin}}\ and\ \bibinfo {author} {\bibfnamefont {Andrew}\ \bibnamefont
  {Lucas}},\ }\bibfield  {title} {\enquote {\bibinfo {title} {Prethermalization
  and the local robustness of gapped systems},}\ }\href@noop {} {\bibfield
  {journal} {\bibinfo  {journal} {Physical Review Letters}\ }\textbf {\bibinfo
  {volume} {131}},\ \bibinfo {pages} {050402} (\bibinfo {year}
  {2023})}\BibitemShut {NoStop}%
\bibitem [{\citenamefont {Yin}\ \emph {et~al.}(2025)\citenamefont {Yin},
  \citenamefont {Surace},\ and\ \citenamefont {Lucas}}]{Yin:2024hjm}%
  \BibitemOpen
  \bibfield  {author} {\bibinfo {author} {\bibfnamefont {Chao}\ \bibnamefont
  {Yin}}, \bibinfo {author} {\bibfnamefont {Federica~M.}\ \bibnamefont
  {Surace}}, \ and\ \bibinfo {author} {\bibfnamefont {Andrew}\ \bibnamefont
  {Lucas}},\ }\bibfield  {title} {\enquote {\bibinfo {title} {{Theory of
  Metastable States in Many-Body Quantum Systems}},}\ }\href {\doibase
  10.1103/PhysRevX.15.011064} {\bibfield  {journal} {\bibinfo  {journal} {Phys.
  Rev. X}\ }\textbf {\bibinfo {volume} {15}},\ \bibinfo {pages} {011064}
  (\bibinfo {year} {2025})},\ \Eprint {http://arxiv.org/abs/2408.05261}
  {arXiv:2408.05261 [math-ph]} \BibitemShut {NoStop}%
\bibitem [{\citenamefont {Aubry}\ and\ \citenamefont
  {Andr{\'e}}(1980)}]{aubry1980analyticity}%
  \BibitemOpen
  \bibfield  {author} {\bibinfo {author} {\bibfnamefont {Serge}\ \bibnamefont
  {Aubry}}\ and\ \bibinfo {author} {\bibfnamefont {Gilles}\ \bibnamefont
  {Andr{\'e}}},\ }\bibfield  {title} {\enquote {\bibinfo {title} {Analyticity
  breaking and anderson localization in incommensurate lattices},}\ }\href@noop
  {} {\bibfield  {journal} {\bibinfo  {journal} {Ann. Israel Phys. Soc}\
  }\textbf {\bibinfo {volume} {3}},\ \bibinfo {pages} {18} (\bibinfo {year}
  {1980})}\BibitemShut {NoStop}%
\bibitem [{\citenamefont {Agrawal}\ \emph {et~al.}(2022)\citenamefont
  {Agrawal}, \citenamefont {Vasseur},\ and\ \citenamefont
  {Gopalakrishnan}}]{sarang_quasiperiodic}%
  \BibitemOpen
  \bibfield  {author} {\bibinfo {author} {\bibfnamefont {Utkarsh}\ \bibnamefont
  {Agrawal}}, \bibinfo {author} {\bibfnamefont {Romain}\ \bibnamefont
  {Vasseur}}, \ and\ \bibinfo {author} {\bibfnamefont {Sarang}\ \bibnamefont
  {Gopalakrishnan}},\ }\bibfield  {title} {\enquote {\bibinfo {title}
  {Quasiperiodic many-body localization transition in dimension $d \ge 1$},}\
  }\href {\doibase 10.1103/PhysRevB.106.094206} {\bibfield  {journal} {\bibinfo
   {journal} {Phys. Rev. B}\ }\textbf {\bibinfo {volume} {106}},\ \bibinfo
  {pages} {094206} (\bibinfo {year} {2022})}\BibitemShut {NoStop}%
\bibitem [{\citenamefont {Iyer}\ \emph {et~al.}(2013)\citenamefont {Iyer},
  \citenamefont {Oganesyan}, \citenamefont {Refael},\ and\ \citenamefont
  {Huse}}]{huse_quasiperiodic}%
  \BibitemOpen
  \bibfield  {author} {\bibinfo {author} {\bibfnamefont {Shankar}\ \bibnamefont
  {Iyer}}, \bibinfo {author} {\bibfnamefont {Vadim}\ \bibnamefont {Oganesyan}},
  \bibinfo {author} {\bibfnamefont {Gil}\ \bibnamefont {Refael}}, \ and\
  \bibinfo {author} {\bibfnamefont {David~A.}\ \bibnamefont {Huse}},\
  }\bibfield  {title} {\enquote {\bibinfo {title} {Many-body localization in a
  quasiperiodic system},}\ }\href {\doibase 10.1103/PhysRevB.87.134202}
  {\bibfield  {journal} {\bibinfo  {journal} {Phys. Rev. B}\ }\textbf {\bibinfo
  {volume} {87}},\ \bibinfo {pages} {134202} (\bibinfo {year}
  {2013})}\BibitemShut {NoStop}%
\bibitem [{\citenamefont {Abdul-Rahman}\ \emph {et~al.}(2025)\citenamefont
  {Abdul-Rahman}, \citenamefont {Fillman}, \citenamefont {Fischbacher},\ and\
  \citenamefont {Liu}}]{abdul2025sharp}%
  \BibitemOpen
  \bibfield  {author} {\bibinfo {author} {\bibfnamefont {Houssam}\ \bibnamefont
  {Abdul-Rahman}}, \bibinfo {author} {\bibfnamefont {Jake}\ \bibnamefont
  {Fillman}}, \bibinfo {author} {\bibfnamefont {Christoph}\ \bibnamefont
  {Fischbacher}}, \ and\ \bibinfo {author} {\bibfnamefont {Wencai}\
  \bibnamefont {Liu}},\ }\bibfield  {title} {\enquote {\bibinfo {title} {Sharp
  polynomial velocity decay bounds for multidimensional periodic schrödinger
  operators},}\ }\href {\doibase 10.48550/arXiv.2509.04381} {\bibfield
  {journal} {\bibinfo  {journal} {arXiv preprint arXiv:2509.04381}\ } (\bibinfo
  {year} {2025}),\ 10.48550/arXiv.2509.04381}\BibitemShut {NoStop}%
\bibitem [{\citenamefont {Bravyi}\ \emph {et~al.}(2006)\citenamefont {Bravyi},
  \citenamefont {Hastings},\ and\ \citenamefont {Verstraete}}]{bravyi2006lieb}%
  \BibitemOpen
  \bibfield  {author} {\bibinfo {author} {\bibfnamefont {Sergey}\ \bibnamefont
  {Bravyi}}, \bibinfo {author} {\bibfnamefont {Matthew~B}\ \bibnamefont
  {Hastings}}, \ and\ \bibinfo {author} {\bibfnamefont {Frank}\ \bibnamefont
  {Verstraete}},\ }\bibfield  {title} {\enquote {\bibinfo {title}
  {Lieb-robinson bounds and the generation of correlations and topological
  quantum order},}\ }\href {\doibase 10.1103/PhysRevLett.97.050401} {\bibfield
  {journal} {\bibinfo  {journal} {Physical Review Letters}\ }\textbf {\bibinfo
  {volume} {97}},\ \bibinfo {pages} {050401} (\bibinfo {year}
  {2006})}\BibitemShut {NoStop}%
\bibitem [{\citenamefont {Hastings}\ and\ \citenamefont
  {Koma}(2006)}]{hastings2006spectral}%
  \BibitemOpen
  \bibfield  {author} {\bibinfo {author} {\bibfnamefont {Matthew~B}\
  \bibnamefont {Hastings}}\ and\ \bibinfo {author} {\bibfnamefont {Tohru}\
  \bibnamefont {Koma}},\ }\bibfield  {title} {\enquote {\bibinfo {title}
  {Spectral gap and exponential decay of correlations},}\ }\href {\doibase
  10.1007/s00220-006-0030-4} {\bibfield  {journal} {\bibinfo  {journal}
  {Communications in Mathematical Physics}\ }\textbf {\bibinfo {volume}
  {265}},\ \bibinfo {pages} {781--804} (\bibinfo {year} {2006})}\BibitemShut
  {NoStop}%
\bibitem [{\citenamefont {Nachtergaele}\ and\ \citenamefont
  {Sims}(2006)}]{exponential_clustering}%
  \BibitemOpen
  \bibfield  {author} {\bibinfo {author} {\bibfnamefont {Bruno}\ \bibnamefont
  {Nachtergaele}}\ and\ \bibinfo {author} {\bibfnamefont {Robert}\ \bibnamefont
  {Sims}},\ }\bibfield  {title} {\enquote {\bibinfo {title} {Lieb-robinson
  bounds and the exponential clustering theorem},}\ }\href {\doibase
  10.1007/s00220-006-1556-1} {\bibfield  {journal} {\bibinfo  {journal}
  {Communications in Mathematical Physics}\ }\textbf {\bibinfo {volume}
  {265}},\ \bibinfo {pages} {119–130} (\bibinfo {year} {2006})}\BibitemShut
  {NoStop}%
\bibitem [{\citenamefont {McDonough}\ \emph {et~al.}(2025)\citenamefont
  {McDonough}, \citenamefont {Yin}, \citenamefont {Lucas},\ and\ \citenamefont
  {Zhang}}]{mcdonough2025}%
  \BibitemOpen
  \bibfield  {author} {\bibinfo {author} {\bibfnamefont {Ben~T.}\ \bibnamefont
  {McDonough}}, \bibinfo {author} {\bibfnamefont {Chao}\ \bibnamefont {Yin}},
  \bibinfo {author} {\bibfnamefont {Andrew}\ \bibnamefont {Lucas}}, \ and\
  \bibinfo {author} {\bibfnamefont {Carolyn}\ \bibnamefont {Zhang}},\ }\href
  {https://arxiv.org/abs/2502.02652} {\enquote {\bibinfo {title} {Lieb-robinson
  bounds with exponential-in-volume tails},}\ } (\bibinfo {year} {2025}),\
  \Eprint {http://arxiv.org/abs/2502.02652} {arXiv:2502.02652 [quant-ph]}
  \BibitemShut {NoStop}%
\bibitem [{\citenamefont {Osborne}(2007)}]{Osborne_2007}%
  \BibitemOpen
  \bibfield  {author} {\bibinfo {author} {\bibfnamefont {Tobias~J.}\
  \bibnamefont {Osborne}},\ }\bibfield  {title} {\enquote {\bibinfo {title}
  {Simulating adiabatic evolution of gapped spin systems},}\ }\href {\doibase
  10.1103/physreva.75.032321} {\bibfield  {journal} {\bibinfo  {journal}
  {Physical Review A}\ }\textbf {\bibinfo {volume} {75}} (\bibinfo {year}
  {2007}),\ 10.1103/physreva.75.032321}\BibitemShut {NoStop}%
\bibitem [{\citenamefont {Osborne}(2006)}]{osborne2006}%
  \BibitemOpen
  \bibfield  {author} {\bibinfo {author} {\bibfnamefont {Tobias~J.}\
  \bibnamefont {Osborne}},\ }\bibfield  {title} {\enquote {\bibinfo {title}
  {Efficient approximation of the dynamics of one-dimensional quantum spin
  systems},}\ }\href {\doibase 10.1103/PhysRevLett.97.157202} {\bibfield
  {journal} {\bibinfo  {journal} {Phys. Rev. Lett.}\ }\textbf {\bibinfo
  {volume} {97}},\ \bibinfo {pages} {157202} (\bibinfo {year}
  {2006})}\BibitemShut {NoStop}%
\bibitem [{\citenamefont {Chen}\ \emph {et~al.}(2023)\citenamefont {Chen},
  \citenamefont {Lucas},\ and\ \citenamefont {Yin}}]{chen2023speed}%
  \BibitemOpen
  \bibfield  {author} {\bibinfo {author} {\bibfnamefont {Chi-Fang~Anthony}\
  \bibnamefont {Chen}}, \bibinfo {author} {\bibfnamefont {Andrew}\ \bibnamefont
  {Lucas}}, \ and\ \bibinfo {author} {\bibfnamefont {Chao}\ \bibnamefont
  {Yin}},\ }\bibfield  {title} {\enquote {\bibinfo {title} {Speed limits and
  locality in many-body quantum dynamics},}\ }\href@noop {} {\bibfield
  {journal} {\bibinfo  {journal} {Reports on Progress in Physics}\ }\textbf
  {\bibinfo {volume} {86}},\ \bibinfo {pages} {116001} (\bibinfo {year}
  {2023})}\BibitemShut {NoStop}%
\bibitem [{\citenamefont {Elgart}\ and\ \citenamefont
  {Klein}(2024)}]{elgart2024slow}%
  \BibitemOpen
  \bibfield  {author} {\bibinfo {author} {\bibfnamefont {Alexander}\
  \bibnamefont {Elgart}}\ and\ \bibinfo {author} {\bibfnamefont {Abel}\
  \bibnamefont {Klein}},\ }\bibfield  {title} {\enquote {\bibinfo {title} {Slow
  propagation of information on the random xxz quantum spin chain},}\ }\href
  {\doibase 10.1007/s00220-024-05127-y} {\bibfield  {journal} {\bibinfo
  {journal} {Communications in Mathematical Physics}\ }\textbf {\bibinfo
  {volume} {405}},\ \bibinfo {pages} {239} (\bibinfo {year}
  {2024})}\BibitemShut {NoStop}%
\bibitem [{\citenamefont {Kolmogorov}(2005)}]{kolmogorov2005preservation}%
  \BibitemOpen
  \bibfield  {author} {\bibinfo {author} {\bibfnamefont {Academician~AN}\
  \bibnamefont {Kolmogorov}},\ }\bibfield  {title} {\enquote {\bibinfo {title}
  {Preservation of conditionally periodic movements with small change in the
  hamilton function},}\ }in\ \href {\doibase 10.1007/BFb0021737} {\emph
  {\bibinfo {booktitle} {Stochastic Behavior in Classical and Quantum
  Hamiltonian Systems: Volta Memorial Conference, Como, 1977}}}\ (\bibinfo
  {organization} {Springer},\ \bibinfo {year} {2005})\ pp.\ \bibinfo {pages}
  {51--56}\BibitemShut {NoStop}%
\bibitem [{\citenamefont {Rakovszky}\ \emph {et~al.}(2022)\citenamefont
  {Rakovszky}, \citenamefont {von Keyserlingk},\ and\ \citenamefont
  {Pollmann}}]{tiborhydro}%
  \BibitemOpen
  \bibfield  {author} {\bibinfo {author} {\bibfnamefont {Tibor}\ \bibnamefont
  {Rakovszky}}, \bibinfo {author} {\bibfnamefont {C.~W.}\ \bibnamefont {von
  Keyserlingk}}, \ and\ \bibinfo {author} {\bibfnamefont {Frank}\ \bibnamefont
  {Pollmann}},\ }\bibfield  {title} {\enquote {\bibinfo {title}
  {Dissipation-assisted operator evolution method for capturing hydrodynamic
  transport},}\ }\href {\doibase 10.1103/PhysRevB.105.075131} {\bibfield
  {journal} {\bibinfo  {journal} {Phys. Rev. B}\ }\textbf {\bibinfo {volume}
  {105}},\ \bibinfo {pages} {075131} (\bibinfo {year} {2022})}\BibitemShut
  {NoStop}%
\bibitem [{\citenamefont {Hartman}\ \emph {et~al.}(2017)\citenamefont
  {Hartman}, \citenamefont {Hartnoll},\ and\ \citenamefont
  {Mahajan}}]{hartman}%
  \BibitemOpen
  \bibfield  {author} {\bibinfo {author} {\bibfnamefont {Thomas}\ \bibnamefont
  {Hartman}}, \bibinfo {author} {\bibfnamefont {Sean~A.}\ \bibnamefont
  {Hartnoll}}, \ and\ \bibinfo {author} {\bibfnamefont {Raghu}\ \bibnamefont
  {Mahajan}},\ }\bibfield  {title} {\enquote {\bibinfo {title} {Upper bound on
  diffusivity},}\ }\href {\doibase 10.1103/PhysRevLett.119.141601} {\bibfield
  {journal} {\bibinfo  {journal} {Phys. Rev. Lett.}\ }\textbf {\bibinfo
  {volume} {119}},\ \bibinfo {pages} {141601} (\bibinfo {year}
  {2017})}\BibitemShut {NoStop}%
\bibitem [{\citenamefont {Huveneers}(2013)}]{Huveneers_2013}%
  \BibitemOpen
  \bibfield  {author} {\bibinfo {author} {\bibfnamefont {François}\
  \bibnamefont {Huveneers}},\ }\bibfield  {title} {\enquote {\bibinfo {title}
  {Drastic fall-off of the thermal conductivity for disordered lattices in the
  limit of weak anharmonic interactions},}\ }\href {\doibase
  10.1088/0951-7715/26/3/837} {\bibfield  {journal} {\bibinfo  {journal}
  {Nonlinearity}\ }\textbf {\bibinfo {volume} {26}},\ \bibinfo {pages}
  {837–854} (\bibinfo {year} {2013})}\BibitemShut {NoStop}%
\bibitem [{\citenamefont {De~Roeck}\ and\ \citenamefont
  {Huveneers}(2014)}]{De_Roeck_2014}%
  \BibitemOpen
  \bibfield  {author} {\bibinfo {author} {\bibfnamefont {Wojciech}\
  \bibnamefont {De~Roeck}}\ and\ \bibinfo {author} {\bibfnamefont {François}\
  \bibnamefont {Huveneers}},\ }\bibfield  {title} {\enquote {\bibinfo {title}
  {Asymptotic quantum many-body localization from thermal disorder},}\ }\href
  {\doibase 10.1007/s00220-014-2116-8} {\bibfield  {journal} {\bibinfo
  {journal} {Communications in Mathematical Physics}\ }\textbf {\bibinfo
  {volume} {332}},\ \bibinfo {pages} {1017–1082} (\bibinfo {year}
  {2014})}\BibitemShut {NoStop}%
\bibitem [{\citenamefont {Yin}\ and\ \citenamefont {Lucas}(2025)}]{chaocode}%
  \BibitemOpen
  \bibfield  {author} {\bibinfo {author} {\bibfnamefont {Chao}\ \bibnamefont
  {Yin}}\ and\ \bibinfo {author} {\bibfnamefont {Andrew}\ \bibnamefont
  {Lucas}},\ }\bibfield  {title} {\enquote {\bibinfo {title} {Low-density
  parity-check codes as stable phases of quantum matter},}\ }\href {\doibase
  10.1103/361k-nj4b} {\bibfield  {journal} {\bibinfo  {journal} {PRX Quantum}\
  }\textbf {\bibinfo {volume} {6}},\ \bibinfo {pages} {030329} (\bibinfo {year}
  {2025})}\BibitemShut {NoStop}%
\bibitem [{\citenamefont {De~Roeck}\ \emph {et~al.}(2025)\citenamefont
  {De~Roeck}, \citenamefont {Khemani}, \citenamefont {Li}, \citenamefont
  {O'Dea},\ and\ \citenamefont {Rakovszky}}]{vedikacode}%
  \BibitemOpen
  \bibfield  {author} {\bibinfo {author} {\bibfnamefont {Wojciech}\
  \bibnamefont {De~Roeck}}, \bibinfo {author} {\bibfnamefont {Vedika}\
  \bibnamefont {Khemani}}, \bibinfo {author} {\bibfnamefont {Yaodong}\
  \bibnamefont {Li}}, \bibinfo {author} {\bibfnamefont {Nicholas}\ \bibnamefont
  {O'Dea}}, \ and\ \bibinfo {author} {\bibfnamefont {Tibor}\ \bibnamefont
  {Rakovszky}},\ }\bibfield  {title} {\enquote {\bibinfo {title} {Low-density
  parity-check stabilizer codes as gapped quantum phases: Stability under
  graph-local perturbations},}\ }\href {\doibase 10.1103/7x71-8j7k} {\bibfield
  {journal} {\bibinfo  {journal} {PRX Quantum}\ }\textbf {\bibinfo {volume}
  {6}},\ \bibinfo {pages} {030330} (\bibinfo {year} {2025})}\BibitemShut
  {NoStop}%
\bibitem [{\citenamefont {Kastner}(2011)}]{longrange_metastab_Ising}%
  \BibitemOpen
  \bibfield  {author} {\bibinfo {author} {\bibfnamefont {Michael}\ \bibnamefont
  {Kastner}},\ }\bibfield  {title} {\enquote {\bibinfo {title} {Diverging
  equilibration times in long-range quantum spin models},}\ }\href {\doibase
  10.1103/PhysRevLett.106.130601} {\bibfield  {journal} {\bibinfo  {journal}
  {Phys. Rev. Lett.}\ }\textbf {\bibinfo {volume} {106}},\ \bibinfo {pages}
  {130601} (\bibinfo {year} {2011})}\BibitemShut {NoStop}%
\bibitem [{\citenamefont {Defenu}(2021)}]{longrange_metastab_spec}%
  \BibitemOpen
  \bibfield  {author} {\bibinfo {author} {\bibfnamefont {Nicolò}\ \bibnamefont
  {Defenu}},\ }\bibfield  {title} {\enquote {\bibinfo {title} {Metastability
  and discrete spectrum of long-range systems},}\ }\href {\doibase
  10.1073/pnas.2101785118} {\bibfield  {journal} {\bibinfo  {journal}
  {Proceedings of the National Academy of Sciences}\ }\textbf {\bibinfo
  {volume} {118}},\ \bibinfo {pages} {e2101785118} (\bibinfo {year}
  {2021})}\BibitemShut {NoStop}%
\bibitem [{\citenamefont {Collura}\ \emph {et~al.}(2022)\citenamefont
  {Collura}, \citenamefont {De~Luca}, \citenamefont {Rossini},\ and\
  \citenamefont {Lerose}}]{lerose}%
  \BibitemOpen
  \bibfield  {author} {\bibinfo {author} {\bibfnamefont {Mario}\ \bibnamefont
  {Collura}}, \bibinfo {author} {\bibfnamefont {Andrea}\ \bibnamefont
  {De~Luca}}, \bibinfo {author} {\bibfnamefont {Davide}\ \bibnamefont
  {Rossini}}, \ and\ \bibinfo {author} {\bibfnamefont {Alessio}\ \bibnamefont
  {Lerose}},\ }\bibfield  {title} {\enquote {\bibinfo {title} {Discrete
  time-crystalline response stabilized by domain-wall confinement},}\ }\href
  {\doibase 10.1103/PhysRevX.12.031037} {\bibfield  {journal} {\bibinfo
  {journal} {Phys. Rev. X}\ }\textbf {\bibinfo {volume} {12}},\ \bibinfo
  {pages} {031037} (\bibinfo {year} {2022})}\BibitemShut {NoStop}%
\bibitem [{\citenamefont {Chen}\ and\ \citenamefont
  {Lucas}(2019)}]{Chen:2019hou}%
  \BibitemOpen
  \bibfield  {author} {\bibinfo {author} {\bibfnamefont {Chi-Fang}\
  \bibnamefont {Chen}}\ and\ \bibinfo {author} {\bibfnamefont {Andrew}\
  \bibnamefont {Lucas}},\ }\bibfield  {title} {\enquote {\bibinfo {title}
  {{Finite speed of quantum scrambling with long range interactions}},}\ }\href
  {\doibase 10.1103/PhysRevLett.123.250605} {\bibfield  {journal} {\bibinfo
  {journal} {Phys. Rev. Lett.}\ }\textbf {\bibinfo {volume} {123}},\ \bibinfo
  {pages} {250605} (\bibinfo {year} {2019})},\ \Eprint
  {http://arxiv.org/abs/1907.07637} {arXiv:1907.07637 [quant-ph]} \BibitemShut
  {NoStop}%
\bibitem [{\citenamefont {Kuwahara}\ and\ \citenamefont
  {Saito}(2020)}]{Kuwahara:2019rlw}%
  \BibitemOpen
  \bibfield  {author} {\bibinfo {author} {\bibfnamefont {Tomotaka}\
  \bibnamefont {Kuwahara}}\ and\ \bibinfo {author} {\bibfnamefont {Keiji}\
  \bibnamefont {Saito}},\ }\bibfield  {title} {\enquote {\bibinfo {title}
  {{Strictly linear light cones in long-range interacting systems of arbitrary
  dimensions}},}\ }\href {\doibase 10.1103/PhysRevX.10.031010} {\bibfield
  {journal} {\bibinfo  {journal} {Phys. Rev. X}\ }\textbf {\bibinfo {volume}
  {10}},\ \bibinfo {pages} {031010} (\bibinfo {year} {2020})},\ \Eprint
  {http://arxiv.org/abs/1910.14477} {arXiv:1910.14477 [quant-ph]} \BibitemShut
  {NoStop}%
\bibitem [{\citenamefont {Tran}\ \emph {et~al.}(2021)\citenamefont {Tran},
  \citenamefont {Guo}, \citenamefont {Baldwin}, \citenamefont {Ehrenberg},
  \citenamefont {Gorshkov},\ and\ \citenamefont {Lucas}}]{Tran:2021ogo}%
  \BibitemOpen
  \bibfield  {author} {\bibinfo {author} {\bibfnamefont {Minh~C.}\ \bibnamefont
  {Tran}}, \bibinfo {author} {\bibfnamefont {Andrew~Y.}\ \bibnamefont {Guo}},
  \bibinfo {author} {\bibfnamefont {Christopher~L.}\ \bibnamefont {Baldwin}},
  \bibinfo {author} {\bibfnamefont {Adam}\ \bibnamefont {Ehrenberg}}, \bibinfo
  {author} {\bibfnamefont {Alexey~V.}\ \bibnamefont {Gorshkov}}, \ and\
  \bibinfo {author} {\bibfnamefont {Andrew}\ \bibnamefont {Lucas}},\ }\bibfield
   {title} {\enquote {\bibinfo {title} {{Lieb-Robinson Light Cone for Power-Law
  Interactions}},}\ }\href {\doibase 10.1103/PhysRevLett.127.160401} {\bibfield
   {journal} {\bibinfo  {journal} {Phys. Rev. Lett.}\ }\textbf {\bibinfo
  {volume} {127}},\ \bibinfo {pages} {160401} (\bibinfo {year} {2021})},\
  \Eprint {http://arxiv.org/abs/2103.15828} {arXiv:2103.15828 [quant-ph]}
  \BibitemShut {NoStop}%
\bibitem [{\citenamefont {Lemm}\ \emph {et~al.}(2025)\citenamefont {Lemm},
  \citenamefont {Rubiliani},\ and\ \citenamefont {Zhang}}]{lemm2025quantum}%
  \BibitemOpen
  \bibfield  {author} {\bibinfo {author} {\bibfnamefont {Marius}\ \bibnamefont
  {Lemm}}, \bibinfo {author} {\bibfnamefont {Carla}\ \bibnamefont {Rubiliani}},
  \ and\ \bibinfo {author} {\bibfnamefont {Jingxuan}\ \bibnamefont {Zhang}},\
  }\bibfield  {title} {\enquote {\bibinfo {title} {{On the quantum dynamics of
  long-ranged Bose-Hubbard Hamiltonians}},}\ }\href {\doibase
  10.48550/arXiv.2505.01786} {\bibfield  {journal} {\bibinfo  {journal} {arXiv
  preprint arXiv:2505.01786}\ } (\bibinfo {year} {2025}),\
  10.48550/arXiv.2505.01786}\BibitemShut {NoStop}%
\bibitem [{\citenamefont {Hastings}(2010)}]{hastings_mobilitygap}%
  \BibitemOpen
  \bibfield  {author} {\bibinfo {author} {\bibfnamefont {M.~B.}\ \bibnamefont
  {Hastings}},\ }\href {https://arxiv.org/abs/1001.5280} {\enquote {\bibinfo
  {title} {Quasi-adiabatic continuation for disordered systems: Applications to
  correlations, lieb-schultz-mattis, and hall conductance},}\ } (\bibinfo
  {year} {2010}),\ \Eprint {http://arxiv.org/abs/1001.5280} {arXiv:1001.5280
  [math-ph]} \BibitemShut {NoStop}%
\bibitem [{\citenamefont {Ong}(1984)}]{shurmultipliers}%
  \BibitemOpen
  \bibfield  {author} {\bibinfo {author} {\bibfnamefont {Sing-Cheong}\
  \bibnamefont {Ong}},\ }\bibfield  {title} {\enquote {\bibinfo {title} {On the
  schur multiplier norm of matrices},}\ }\href@noop {} {\bibfield  {journal}
  {\bibinfo  {journal} {Linear algebra and its applications}\ }\textbf
  {\bibinfo {volume} {56}},\ \bibinfo {pages} {45--55} (\bibinfo {year}
  {1984})}\BibitemShut {NoStop}%
\bibitem [{\citenamefont {Bhatia}(2000)}]{bhatia2000pinching}%
  \BibitemOpen
  \bibfield  {author} {\bibinfo {author} {\bibfnamefont {Rajendra}\
  \bibnamefont {Bhatia}},\ }\bibfield  {title} {\enquote {\bibinfo {title}
  {Pinching, trimming, truncating, and averaging of matrices},}\ }\href@noop {}
  {\bibfield  {journal} {\bibinfo  {journal} {The American Mathematical
  Monthly}\ }\textbf {\bibinfo {volume} {107}},\ \bibinfo {pages} {602--608}
  (\bibinfo {year} {2000})}\BibitemShut {NoStop}%
\bibitem [{\citenamefont {Hoorfar}\ and\ \citenamefont
  {Hassani}(2008)}]{Wfunction}%
  \BibitemOpen
  \bibfield  {author} {\bibinfo {author} {\bibfnamefont {Abdolhossein}\
  \bibnamefont {Hoorfar}}\ and\ \bibinfo {author} {\bibfnamefont {Mehdi}\
  \bibnamefont {Hassani}},\ }\bibfield  {title} {\enquote {\bibinfo {title}
  {Inequalities on the lambert w function and hyperpower function},}\
  }\href@noop {} {\bibfield  {journal} {\bibinfo  {journal} {J. Inequal. Pure
  and Appl. Math}\ }\textbf {\bibinfo {volume} {9}},\ \bibinfo {pages} {5--9}
  (\bibinfo {year} {2008})}\BibitemShut {NoStop}%
\bibitem [{\citenamefont {Garey}\ \emph {et~al.}(1990)\citenamefont {Garey},
  \citenamefont {Johnson} \emph {et~al.}}]{garey1990guide}%
  \BibitemOpen
  \bibfield  {author} {\bibinfo {author} {\bibfnamefont {Michael~R}\
  \bibnamefont {Garey}}, \bibinfo {author} {\bibfnamefont {David~S}\
  \bibnamefont {Johnson}},  \emph {et~al.},\ }\bibfield  {title} {\enquote
  {\bibinfo {title} {A guide to the theory of np-completeness},}\ }\href@noop
  {} {\bibfield  {journal} {\bibinfo  {journal} {Computers and intractability}\
  ,\ \bibinfo {pages} {37--79}} (\bibinfo {year} {1990})}\BibitemShut {NoStop}%
\bibitem [{\citenamefont {Iaconis}\ \emph {et~al.}(2019)\citenamefont
  {Iaconis}, \citenamefont {Vijay},\ and\ \citenamefont
  {Nandkishore}}]{Iaconis:2019hab}%
  \BibitemOpen
  \bibfield  {author} {\bibinfo {author} {\bibfnamefont {Jason}\ \bibnamefont
  {Iaconis}}, \bibinfo {author} {\bibfnamefont {Sagar}\ \bibnamefont {Vijay}},
  \ and\ \bibinfo {author} {\bibfnamefont {Rahul}\ \bibnamefont
  {Nandkishore}},\ }\bibfield  {title} {\enquote {\bibinfo {title} {{Anomalous
  Subdiffusion from Subsystem Symmetries}},}\ }\href {\doibase
  10.1103/PhysRevB.100.214301} {\bibfield  {journal} {\bibinfo  {journal}
  {Phys. Rev. B}\ }\textbf {\bibinfo {volume} {100}},\ \bibinfo {pages}
  {214301} (\bibinfo {year} {2019})},\ \Eprint
  {http://arxiv.org/abs/1907.10629} {arXiv:1907.10629 [cond-mat.stat-mech]}
  \BibitemShut {NoStop}%
\bibitem [{\citenamefont {Gromov}\ \emph {et~al.}(2020)\citenamefont {Gromov},
  \citenamefont {Lucas},\ and\ \citenamefont {Nandkishore}}]{Gromov:2020yoc}%
  \BibitemOpen
  \bibfield  {author} {\bibinfo {author} {\bibfnamefont {Andrey}\ \bibnamefont
  {Gromov}}, \bibinfo {author} {\bibfnamefont {Andrew}\ \bibnamefont {Lucas}},
  \ and\ \bibinfo {author} {\bibfnamefont {Rahul~M.}\ \bibnamefont
  {Nandkishore}},\ }\bibfield  {title} {\enquote {\bibinfo {title} {{Fracton
  hydrodynamics}},}\ }\href {\doibase 10.1103/PhysRevResearch.2.033124}
  {\bibfield  {journal} {\bibinfo  {journal} {Phys. Rev. Res.}\ }\textbf
  {\bibinfo {volume} {2}},\ \bibinfo {pages} {033124} (\bibinfo {year}
  {2020})},\ \Eprint {http://arxiv.org/abs/2003.09429} {arXiv:2003.09429
  [cond-mat.str-el]} \BibitemShut {NoStop}%
\end{thebibliography}%
\end{document}